\documentclass[letterpaper,11pt]{article}
\usepackage{amsfonts}
\usepackage[pdftex]{graphicx}
\usepackage{amsmath,amssymb,amsthm}
\usepackage{natbib}
\usepackage{hyperref}
\usepackage{bm}
\usepackage{fancyhdr}
\usepackage{sectsty}
\usepackage[notcite,notref,final]{showkeys}
\usepackage{pdfsync}
\usepackage{xr}
\usepackage{paralist} 
\usepackage{tikz}
\usepackage{stmaryrd}
\usepackage{setspace}
\usepackage{todonotes}
\usepackage{multirow}
\usepackage{booktabs}
\usepackage{threeparttable}
\usepackage{float}
\usepackage{enumitem}
\usepackage{mathrsfs}
\usepackage{comment}
\usepackage{bbm}

\usepackage{hyperref}
\hypersetup{colorlinks=true, linkcolor=blue, citecolor=blue}

\def\argmin{\mathop{\rm arg\,min}}

\newcommand{\R}{\mathbb{R}}

\newcommand{\cP}{\mathcal{P}}
\newcommand{\cW}{\mathcal{W}}

\newcommand{\cK}{\mathcal{K}}  
\newcommand{\cC}{\mathcal{C}}

\DeclareMathOperator{\Bin}{Bin}

\numberwithin{equation}{section}

\theoremstyle{plain}
\newtheorem{theorem}{Theorem} \newtheorem{proposition}{Proposition}[section] \newtheorem{lemma}{Lemma}[section] \newtheorem{corollary}{Corollary} \newtheorem{assumption}{Assumption} \newtheorem{remark}{Remark}[section]  

\theoremstyle{definition}
\newtheorem{example}{Example}

\title{Robust Likelihood Ratio Tests for Incomplete Economic Models\thanks{We are grateful to Tim Armstrong, Federico Bugni, Shuowen Chen, Andrew Chesher, Larry Epstein, Marc Henry, Hidehiko Ichimura, Toru Kitagawa, Michal Kolesar, Xun Lu, Francesca Molinari, Ulrich M\"{u}ller,  Whitney Newey, and Adam Rosen for their helpful comments and discussions. We thank, for their comments, the seminar, lecture, and conference participants at the following: BU, Cornell, Columbia, Duke, Hitotsubashi, HKUST, Kobe, Kyoto, Michigan State, Princeton, Rutgers, Texas A\&M, UCSD, U. Tokyo (CREPE Lectures), UWO, CEMMAP Workshops: Advances in Microeconometrics (Doshisha U., 2016, Vanderbilt, 2019), CEME, JEA Fall Meeting 2017, BU-BC joint econometrics workshop, ASSA Meeting 2018, NASMES 2018, New Frontiers in Econometrics (UConn), and Conference on Incomplete Models (Northwestern/Cemmap). We are grateful to Undral Byambadalai and Junwen Lu for their excellent research assistance.}}
\author{Hiroaki Kaido\thanks{Boston University, Email: hkaido@bu.edu. Financial support from NSF grants SES-1357643 and SES-1824344 is gratefully acknowledged. The author thanks Microsoft Research New England at which part of the research was conducted.
} \qquad Yi Zhang\thanks{Jinan University, Email: yzhangjnu@outlook.com. Financial support from the Young Scientists Fund of the National Natural Science Foundation of China (NSFC) grant No.72203077 is gratefully acknowledged.}}

\date{}

\begin{document} 
	\maketitle
\begin{abstract}
Economic models with multiple equilibria, self-selection, or weak behavioral restrictions often make set-valued predictions and therefore do not imply a unique likelihood. This paper develops robust likelihood-ratio tests for structural hypotheses in such incomplete models. We evaluate tests by their power guarantee, defined as the smallest rejection probability over all selection mechanisms compatible with the alternative, while requiring uniform size control over all null selections. Using the Huber--Strassen theory of least favorable pairs, we construct finite-sample minimax likelihood-ratio tests. The main result shows that, in repeated experiments, the least favorable pair is the product of the single-experiment least favorable pairs whenever the latent variables are independent across experiments. This product structure holds even though unrestricted selection may induce arbitrary heterogeneity and dependence in the observed outcomes. It reduces a high-dimensional robust testing problem to single-experiment calculations and yields exact finite-sample critical values and Gaussian approximations. For directed one-sided alternatives, we provide conditions under which conditioning on a selection-invariant statistic delivers exact conditional uniformly most powerful (UMP) tests with respect to the power guarantee. In the examples, these optimal tests are simple and interpretable, using selection-invariant features of the data that are directly tied to the hypothesis of interest. Monte Carlo experiments in entry-game and Roy-model designs illustrate size control, power, and the role of robust testability.

\vspace{0.3in}

\end{abstract}

\clearpage
\onehalfspacing
\section{Introduction}
Economic models often make set-valued predictions. In a market entry game, for example, the primitives of the model may imply several equilibria \citep{bresnahan1990entry,bresnahan1991empirical,berry1992estimation,ciliberto2009market}; in a Roy model, the theory may not uniquely determine sector choice when potential outcomes are tied \citep{heckman1990empirical,mourifie2020sharp}; in auction models, weak behavioral restrictions may allow multiple bidding profiles \citep{Haile:2003to,Aradillas-Lopez:2008ab}.\footnote{Other examples include models of voting \citep{kawai2013inferring}, choice of product variety \citep{EIZENBERG:2014aa}, and network formation models \citep{Miyauchi:2016aa}.} A broad class of such \emph{incomplete models} share the structure that, given structural parameter $\theta\in \Theta$ and latent variable $u\in U$, the model predicts the set $G(u|\theta)$ of values for discrete outcome $s$.
In such settings, a structural parameter value does not imply a single distribution of the observable outcome. It implies a set of distributions, one for each admissible selection mechanism. 

This raises several challenges for hypothesis testing. First, without further assumptions, the model permits multiple distributions of the observables even if each hypothesis fully specifies the value of $\theta$. To see this, consider a simple problem in which $H_0:\theta=\theta_0~v.s. ~H_1:\theta=\theta_1$. This problem may appear as testing a simple null hypothesis against a simple alternative. However, under each hypothesis, multiple distributions may be compatible with the theory because any distribution $P$ of the outcome is consistent with $\theta$ as long as one can augment the model by finding a suitable selection mechanism that induces $P$. Therefore, even under the simplest setting, both null and alternative hypotheses can be composite (in terms of permitted distributions).
The problem becomes even more challenging when data are obtained from a sequence of experiments. If one stays agnostic about the selection, the unknown selection mechanism is allowed to be arbitrary across experiments. For example, across experiments, the true selection mechanism may vary with and be correlated through a specific variable; however, the researcher does not even know the identity of this variable.
From the researcher's viewpoint, the resulting outcome sequence is then heterogeneous and dependent in an unknown way. As a result, the usual distributional approximations underlying likelihood-ratio tests, bootstrap procedures, and central-limit arguments may fail when they are applied directly to the observed outcome sequence.  

This paper develops a likelihood-based theory of robust optimal testing for incomplete economic models. To the best of our knowledge, it is the first framework to construct finite-sample optimal tests while allowing the observed outcomes to exhibit arbitrary heterogeneity and dependence induced by unrestricted selection across repeated experiments. We impose independence on the latent variables across experiments, but place no independence, homogeneity, or parametric restrictions on the selection mechanism or on the resulting outcome sequence. We evaluate tests by their \emph{power guarantee}: the smallest rejection probability over all selection mechanisms compatible with the alternative.
An important qualification is that a nontrivial power guarantee is possible only against robustly testable alternatives. If the sets of observable distributions compatible with the null and an alternative overlap, some distribution can be generated under either hypothesis by suitable selection mechanisms. The power guarantee of any level-$\alpha$ test against that alternative can then be no greater than $\alpha$. We call such an alternative not robustly testable relative to the null. Robust testability thus determines when nontrivial power guarantees are attainable.

At the single-experiment level, the theory of \cite{Huber:1973aa} identifies the relevant likelihoods. A \emph{least favorable pair (LFP)} consists of one
distribution compatible with the null that is hardest for size control and one
distribution compatible with the alternative that is hardest for the power
guarantee. The likelihood-ratio test formed from this pair controls size
uniformly over all null selections and maximizes the power guarantee. We make
this construction operational by showing that, for the class of incomplete
models considered here, the LFP can be computed by solving a finite-dimensional
convex program whose constraints are precisely the model's sharp identifying
restrictions. Thus, sharp identifying restrictions are used not only to
characterize identified sets, but also to construct the least favorable
likelihoods that define robust optimal tests. In our leading examples, the LFP
and the resulting minimax test can be derived in closed form.

Our main result extends this construction to repeated experiments. Although
unrestricted selection permits the observed outcome sequence to be
heterogeneous and dependent, the least favorable pair for the repeated
experiment is the product of the least favorable pairs from the individual
experiments whenever the latent variables are independent. Thus, the hardest
null law for size control and the hardest alternative law for power are product
measures even though the true data-generating process need not be. This result
reduces a high-dimensional minimax problem over joint outcome distributions to
a collection of single-experiment problems. In the i.i.d. latent-variable case,
implementation requires solving the single-experiment convex program only once
and multiplying the resulting likelihood ratios across observations. The
product structure also yields exact finite-sample critical values and, under
standard moment conditions, a Gaussian approximation whose validity does not
depend on the unknown selection mechanism.

For directed one-sided alternatives, we obtain an additional uniform optimality
result. Suppose that the alternatives are indexed by a scalar $w$ along a path
$\{\theta_w:w\in\mathcal W_1\}$. If conditioning on a selection-invariant
statistic yields least favorable distributions whose likelihood ratios are
monotone in a common ordering statistic, the resulting conditional LR test is
exact and uniformly most powerful with respect to the power guarantee over the
entire directed family. Thus, while the product-LFP result gives minimax
optimality against each fixed alternative, this conditional result delivers
uniform optimality along a directed path.

This result also gives an interpretable prescription for test construction:
the optimal statistic discards variation generated only by selection and uses
the selection-invariant feature of the model that is most relevant to the
hypothesis of interest. We verify these conditions in a canonical two-player
entry game and in a Roy model. In the entry game, we test no strategic
interaction against a directed family of negative interaction effects. The
relevant feature is the monopoly-versus-duopoly comparison: conditioning
removes the irrelevant split between the two monopoly equilibria and reduces
the problem to a one-sided binomial experiment. The Roy model highlights the
role of robust testability for a directed change in the distribution of
potential outcomes. Along the chosen path, some alternatives can be made
observationally indistinguishable from the null through selection, so no
level-$\alpha$ test can guarantee lower power exceeding $\alpha$; once the
sets of admissible distributions separate, the conditional UMP test delivers
nontrivial guaranteed power.

We illustrate main implications of the theory through Monte Carlo experiments. First, a simulation result demonstrates that the robust LR test controls size even under selection mechanisms that involve unknown cluster dependence or a common shock, whereas a likelihood-ratio test calibrated under i.i.d. selection can substantially overreject. Second, in a binary entry-game, the exact conditional UMP test closely tracks its theoretical power envelope and is insensitive to how the monopoly equilibria are selected. Third, simulations based on the Roy model shows how robust testability may arise only beyond a separation threshold and how, once that threshold is crossed, the conditional UMP test delivers nontrivial guaranteed power.

Finally, we develop two extensions. For models with nuisance parameters, we
formulate a decision problem and construct Bayes--Dempster--Shafer (BDS) LR
tests that combine Bayesian parameter uncertainty with ambiguity arising from
incompleteness. A minimax theorem shows that a level-$\alpha$ test maximizing
a weighted average of power guarantees can be approximated by a sequence of
BDS tests. We also extend the framework to incorporate observable covariates.
We keep these extensions concise because the present paper focuses on optimal
testing, whereas the companion paper \citep{kaido2025universal} uses least
favorable likelihoods and universal inference to construct finite-sample-valid
confidence sets for functions and subvectors of $\theta$ in the presence of
nuisance parameters and covariates.

\subsection{Relation to the Literature}\label{ssec:literature} 
Our study is most closely related to \cite{Epstein:2016qv}, who develop robust confidence regions for repeated incomplete experiments while leaving selection-induced heterogeneity and dependence unrestricted. We adopt their repeated-experiment framework but address a different question: how to construct tests with finite-sample optimality properties. To the best of our knowledge, the present paper is the first to characterize minimax and, under additional structure, exact conditional UMP tests in this setting.

Our study builds on earlier work on incomplete models.\footnote{The analysis of an incomplete system of equations dates back to the early work of \cite{Wald1950}. Here, we focus on reviewing more recent developments in models with multiple equilibria.} In particular, the framework for the single experiment builds on that of \cite{jovanovic1989observable}, who pointed out that models with multiple equilibria lead to incomplete structures and face potential difficulty in identifying structural parameters. \cite{tamer2003incomplete} studied identifying restrictions in an incomplete simultaneous discrete response model with multiple equilibria. Since his seminal work, it has become common to use partially identifying inequality restrictions to bound parameters of interest. \cite{galichon2011set}, \cite{beresteanu2011sharp}, and \cite{Chesher:2014aa} characterized sharp identifying restrictions for a wide range of incomplete models using the theory of random sets. We also use, as a central tool, the capacities associated with random sets. As discussed above, the sharp identifying restrictions play an important role in the construction of tests that achieve robustness and statistical optimality.

Commonly used identifying restrictions take the form of moment inequalities. As such, inference methods developed for moment-inequality models \citep{chernozhukov2007estimation,andrews2010inference,bugni2010bootstrap,andrews2012inference,romano2008inference,Shicanaybugni2015inference,kaido:2019aa} are widely applied. Some procedures \citep[e.g.][]{galichon2021inference,galichon2009test,galichon2013dilation} use capacity-based statistics to construct confidence regions. These approaches generally translate the implications of incomplete models into moment restrictions and combine them with a sampling structure, often i.i.d. sampling. Complementary approaches exploit the model’s set-valued prediction structure directly, working with the induced set of admissible observable distributions. \cite{kaido_molinari_information_based_inference} use a Kullback–Leibler projection onto this set to construct score-based confidence regions that remain well defined and asymptotically valid under misspecification. \cite{li2026finitesampleinferenceincomplete} use the set-valued structure in an optimal-transport formulation to obtain confidence regions with exact finite-sample coverage under unrestricted dependence. Their focus is validity rather than optimality.

Another strand of the literature develops likelihood-based inference for partially identified and incomplete models. \cite{Chen:2011aa} treat the unknown selection mechanism as a nonparametric nuisance parameter and profile it out using a sieve likelihood. \cite{ChenChristensenTamer2018} construct simulation-based confidence sets for identified sets of full parameter vectors and subvectors, using likelihood and related criterion functions with critical values obtained from quasi-posterior draws. In companion work, \cite{kaido2025universal} combine least favorable likelihoods with universal inference to obtain finite-sample-valid confidence sets for functions and subvectors of $\theta$, including counterfactual objects, using split-sample and cross-fit LR statistics. These papers pursue different objectives: the former approaches emphasize likelihood-based confidence sets with asymptotic validity, whereas the companion paper emphasizes finite-sample validity for functional targets rather than optimality. 

Existing optimality results for partially identified models have largely been developed within moment-inequality frameworks. \cite{CANAY2010408} showed that a test based on the empirical likelihood-ratio statistic is optimal under a large-deviations criterion. For moment restrictions that are convex in the parameter, \cite{Kaido:2014aa} established that a test based on a semiparametrically efficient estimator of the identified set achieves the asymptotic power envelope against certain local alternatives. For models characterized by conditional moment inequalities, \cite{Armstrong:2014aa,Armstrong:2018aa} compared the local power properties of procedures based on Cramér–von Mises and weighted Kolmogorov–Smirnov statistics.
These studies focus on models represented through moment inequalities and develop optimality results within that framework. In contrast, we use the set-valued prediction structure directly to construct least favorable likelihoods. This produces finite-sample minimax tests and, for directed alternatives satisfying a conditional monotone-likelihood-ratio property, exact conditional UMP tests, while allowing unrestricted selection-induced heterogeneity and dependence across experiments.

Finally, our framework for inference is related to others that use limit theorems based on the theory of random sets.  As mentioned earlier, we use a Gaussian approximation to compute the critical value for the LR statistic, which is similar in spirit to the central limit theorem (CLT) in \cite{Epstein:2016qv}, whereas a different tool is used to obtain this result because of the nontrivial difference between the LR statistic we use here and their Kolmogorov--Smirnov-type statistic. In a different class of models, in which observations are set-valued, \cite{beresteanu2008asymptotic} applied a central limit theorem for random sets to make their inference.

Throughout, for any metric space $A$, we let $\Sigma_A$ denote its Borel $\sigma$-algebra. We then denote the set of Borel probability measures on $A$ by $\Delta(A)$ and equip it with the topology of weak convergence. Let $N(\mu,V)$ denote the law of a normal random vector with mean $\mu\in\mathbb R^k$ and variance-covariance matrix $V\in\mathbb R^{k\times k}$. For any integrable random vector $X$, we let $E_P[X]$ denote its expectation with respect to probability measure $P$. 

The remainder of the paper is organized as follows. Section \ref{sec:setup} introduces the model and provides illustrative examples. Section \ref{sec:theory} develops the main theoretical results. Section  \ref{sec:preliminaries} reviews belief functions, least favorable pairs, and the Huber--Strassen minimax testing result.  Section \ref{sec:repeated} establishes the product least-favorable-pair theorem for repeated experiments and derives the associated Gaussian approximation. Section \ref{ssec:computation} discusses computation of least favorable pairs using sharp identifying restrictions. Section \ref{ssec:one-sided} develops the exact conditional UMP result for directed one-sided alternatives.  Section \ref{sec:montecarlo} then provides simulation evidence. Section \ref{sec:extensions} contains extensions of the baseline framework and Section \ref{sec:conclusion} concludes. The appendices contain proofs and additional details for the examples.

\section{Setup}\label{sec:setup}
Let $S$ be a finite set of observable outcomes and let $u\in U$ denote a variable unobservable to the researcher, where $U$ is assumed to be a Polish space. Let $\Theta$ denote the parameter space. We let $m=\{m_\theta,\theta\in\Theta\}$ denote a family of Borel probability measures on $U$. For each $\theta\in\Theta$, let $G(\cdot|\theta):U\twoheadrightarrow S$ be a weakly measurable correspondence. This map shows how latent variable $u$ is mapped to a set of permissible outcomes.
Observable outcome $s$ is then a measurable selection of random set $G(u|\theta)$. As such, the model does not impose any restrictions on how $s$ is selected.
One may also introduce observable covariates to this model. As the core analysis remains unaffected, we defer the analysis of this case to Section \ref{sec:covariates}.

The incomplete structure above is summarized by tuple $(S,U,m,\Theta;G)$. Such structures arise in various economic models. To fix the ideas, we present several examples based on simplifications of well-known models. The first example is a binary response game, which is commonly used to analyze environments such as firms' entry into markets and households' joint labor supply decisions \citep{bresnahan1990entry,bresnahan1991empirical,berry1992estimation,ciliberto2009market}.
\begin{example}
	[Binary response game]\rm \label{ex:game} Consider a two-player binary response game with the following payoff:
	\begin{center}
		\begin{tabular}
			{ r|c|c| } \multicolumn{1}{r}{} & \multicolumn{1}{c}{out} & \multicolumn{1}{c}{in} \\
			\cline{2-3} out & $0,0$ & $0,u^{(2)}$\\
			\cline{2-3} in & $u^{(1)},0$ & $u^{(1)}+\theta^{(1)}, u^{(2)}+\theta^{(2)}$ \\
			\cline{2-3}
		\end{tabular}
	\end{center}
The effect of the other player's action (e.g., entry) on player $k$'s payoff is represented by $\theta^{(k)}$.
Throughout, we call $\theta=(\theta^{(1)},\theta^{(2)})'\in\Theta\subset \mathbb R^2$ the players' \emph{strategic interaction effects}. Let $U=\mathbb{R}^2$. The latent payoff shifter $u=(u^{(1)},u^{(2)})'$ follows a continuous distribution $m_\theta$. Consider pure strategy Nash equilibria in this game when $\theta^{(1)}\le 0$ and $\theta^{(2)}\le 0$.\footnote{For simplicity, we focus on games with strategic substitutes throughout. Games with strategic complements, in which $\theta^{(1)}>0,\theta^{(2)}>0$, can be analyzed similarly.} There are four possible equilibrium outcomes: $S=\{(0,0),(1,1),(1,0),(0,1)\}$. How $u$ and $\theta$ are mapped to the equilibrium outcomes is summarized by the following correspondence:
	\begin{equation}
		G(u|\theta) = \left\{
		\begin{array}{cl}
			\{(0,0)\} & u^{(1)}<0,~u^{(2)}<0\\
			\{(1,1)\} & u^{(1)}\ge-\theta^{(1)}, ~u^{(2)}\ge-\theta^{(2)}\\
			\{(1,0)\} & u\in U_1,\\
			\{(0,1)\} & u\in U_2,\\
			\{(1,0),(0,1)\} & 0\le u^{(1)}<-\theta^{(1)},~0\le u^{(2)}<-\theta^{(2)},
		\end{array}
		\right.
	\label{eq:defG} \end{equation}
	where $U_1=\{u:u^{(1)}\ge -\theta^{(1)}, u^{(2)}<-\theta^{(2)})\cup\{u:0\le u^{(1)}<-\theta^{(1)}, u^{(2)}<0\}$ and $U_2=\{u:0\le u^{(1)}<-\theta^{(1)}, u^{(2)}\ge -\theta^{(2)}\}\cup \{u:u^{(1)}<0, u^{(2)}\ge 0\}$. The model predicts multiple equilibria when each player's latent payoff shifter is between the two thresholds ($0$ and $-\theta^{(k)}$, $k=1,2$).
\end{example}

The second example is the (binary) Roy model studied in \cite{mourifie2020sharp}. 
\begin{example}[Roy model]\rm\label{ex:roy}
	Consider an individual who chooses a sector of activity $D\in \{0,1\}$ and whether to work $Y\in \{0,1\}$ in the sector. 
The binary outcome is given by $Y=Y_1D+Y_0(1-D)$, where selection indicator $D$ is determined by binary potential outcomes $(Y_0,Y_1)$ through the following structure:
\begin{align}
	D=\begin{cases}
		1&Y_1>Y_0\\
		0 \text{ or }1& Y_1=Y_0\\
		0&Y_1<Y_0.
	\end{cases}
\label{eq:Roy1}
\end{align}
Binary potential outcome $Y_d$ represents whether one has good economic prospects in sector $d\in\{0,1\}$.\footnote{ \cite{mourifie2020sharp} extended their analysis to more general settings in which $Y_d$ is discrete or continuous (or both).
As discussed in Section 2 of their paper, one could also think of the binary Roy model as a consequence of a two-step decision process in which $D$ is determined first by potential wage $Y^*_d$ in sector $d$, and whether to work in section $d$ is determined by whether $Y^*_d$ crosses a threshold.} 
This model can be mapped to the present framework by letting $s=(y,d)\in S=\{(0,0),(0,1),(1,0),(1,1)\}$ be observable outcomes and $u=(Y_0,Y_1)\in U\equiv \{(0,0),(0,1),(1,0),(1,1)\}$ be latent variables. Since $u$ is discrete, we take the probability mass function of $u$ as a parameter vector. For this, let  $\theta=(\theta^{(0,0)},\theta^{(0,1)},\theta^{(1,0)})'\in\Theta$, where $\theta^{(0,0)}=m_\theta((Y_0,Y_1)=(0,0))$, for instance, and $\Theta=\{\theta\in [0,1]^3:\theta^{(0,0)}+\theta^{(0,1)}+\theta^{(1,0)}\le 1\}$. The Roy selection in \eqref{eq:Roy1} then yields the following correspondence:
\begin{align}
	G(u)=\begin{cases}
		\{(0,0),(0,1)\}&u=(0,0)\\
		\{(1,1)\}&u=(0,1)\\
		\{(1,0)\}&u=(1,0)\\
		\{(1,0),(1,1)\}&u=(1,1).
	\end{cases}
\end{align}
The model implies a unique outcome only if the potential outcomes are ordered (e.g., an individual works in sector 1 when $Y_0=0$ and $Y_1=1$). Otherwise, it predicts multiple outcome values. The unknown selection mechanism determines tie breaking.
\end{example}

\subsection{Set of Permitted Distributions and Robustness}
We start by defining the family of probability distributions compatible with the model structure. 
For a probability measure $m$ on $ U$ and a weakly measurable correspondence $K:U\twoheadrightarrow S$, let
$\mathfrak K_m(K)$ be the set of Markov kernels $\eta$ from $U$ to $S$ such that $\eta(K(u)|u)=1$, $m$-a.s. 

For each $\theta\in\Theta$, define
\begin{align}
	\mathcal{P}_{\theta}\equiv \Big\{P\in \Delta(S): P(A)=\int_{U} \eta(A|u) dm_{\theta}(u), \text{ for some }\eta \in \mathfrak K_{m_\theta}(G(\cdot|\theta)),~A\subseteq S \Big\},
\label{eq:P_theta} \end{align}
where $\eta$ represents the unknown \emph{selection mechanism}. This set collects probability distributions $P$, for which one can find a suitable selection mechanism and make it consistent with a given parameter value $\theta$ and the model structure.
Economic theory rarely provides any guidance on selection. The researcher therefore views any distribution in $\mathcal P_\theta$ as consistent with $\theta$.

Consider testing parameter value $\theta_0$ against another value $\theta_1$ on the basis of observed outcome $s\in S.$ This is equivalent to testing the null hypothesis, $P\in\mathcal{P}_{\theta_0}$, against the alternative hypothesis, $P\in\mathcal{P}_{\theta_1}$.
Note that $\mathcal P_{\theta_0}$ and $\mathcal P_{\theta_1}$ may contain multiple (typically infinitely many) elements because the selection is left unspecified. Therefore, even for testing a single value of $\theta$ against another value, the hypotheses are composite in terms of the permitted distributions.\footnote{The composite nature of the hypotheses arises because the unknown selection is a nuisance parameter. It is possible to allow some components of structural parameter $\theta$ to be additional nuisance parameters. We analyze this extension in Section \ref{sec:nuisance}.}
Given this challenge, we pursue a robust approach to inference. That is, we construct tests that (i) control the size uniformly across distributions permitted under the null and (ii) maximize certain measures of power under the alternative.

\section{Robust Tests for Incomplete Models}\label{sec:theory}
We start with preliminaries including an introduction of the key technical tools and the extension of the Neyman--Pearson lemma by \cite{Huber:1973aa}. We then present novel results on minimax tests, LFPs in repeated experiments, and the computational aspects of the proposed tests.

\subsection{Preliminaries}\label{sec:preliminaries}
\subsubsection*{Belief Functions:} For any $\mathcal P\subseteq\Delta(S)$, define the upper and lower probabilities of $\mathcal P$ pointwise by $\nu^{*}(A)\equiv \sup_{P \in \mathcal{P} } P(A)$ and $\nu(A)\equiv\inf_{P \in \mathcal{P}}P(A),A\subset S$, respectively. These functions are conjugate to each other in the sense that $\nu^{*}(A)=1-\nu(A^{c})$ for any $A\subset S$. Under mild restrictions on $\mathcal P$, they define set functions called \emph{capacities}.\footnote{Appendix A provides the details. Some authors distinguish a capacity from its conjugate (co-capacity). For simplicity, we call both of these ``capacities'' throughout.}

For each $\theta\in\Theta$ and $A\subset S$, define $\nu_\theta$ and $\nu_\theta^*$ as the lower and upper probabilities of $\mathcal P_\theta$ defined in \eqref{eq:P_theta}:
\begin{align}
	\nu_\theta(A)\equiv\inf_{P \in \mathcal{P}_\theta}P(A),\hspace{0.4in}\text{ and }\hspace{0.4in}\nu^{*}_\theta(A)\equiv \sup_{P \in \mathcal{P}_\theta} P(A).
\label{eq:defnutheta} \end{align}
The key factor for our analysis is that  $\nu_\theta$ is a \emph{belief function} (or infinitely monotone capacity).\footnote{The infinite monotonicity of $\nu_\theta$ follows from \cite{philippe1999decision} (Theorem 3). The foundations of belief functions are given by \cite{dempster1967upper} and \cite{shafer1982belief}. See \cite{Gul:2014aa} and \cite{epstein2015exchangeable} for the axiomatic foundations of the use of belief functions in incomplete models.} From Choquet's theorem \citep[e.g.,][]{Choquet1954,philippe1999decision,Molchanov:2006aa}, it is related to the probability distribution of random set $G(u|\theta)$ as follows:
\begin{align}
	\nu_\theta(A)=m_\theta(G(u|\theta)\subseteq A),~\text{ for any }A\subset S.
\label{eq:choquet_eq} \end{align}
This representation allows us to obtain $\nu_\theta$ without explicitly solving the minimization (or maximization) in \eqref{eq:defnutheta} by computing the right-hand side of \eqref{eq:choquet_eq} directly. Another key property of the belief function is that $P\in\mathcal P_\theta$ is equivalent to the following statement:
\begin{align}
\nu_\theta(A)\le P(A),~A\subset S.\label{eq:sir}
 \end{align}
\cite{galichon2011set} used the restrictions above to characterize the smallest possible (or ``sharp'') identification region of the parameters.\footnote{\cite{galichon2011set} used the conjugate of $\nu_\theta$, which yields equivalent identifying restrictions.} Following the literature, we call these the \emph{sharp identifying restrictions} \citep[see also][]{beresteanu2011sharp,Chesher:2014aa}.\footnote{While the restrictions play a role in constructing robust tests, the sharp identified set does not play a role as the latter is an object of interest when the sampling processes reveals the unique data generating process in the limit, which is not guaranteed in our setting. See  \cite{Epstein:2016qv} for a discussion.}

\subsubsection*{Theory of \cite{Huber:1973aa}:} Our starting point is an analog of the Neyman--Pearson framework, which builds upon HS. For $\theta_0, \theta_1 \in \Theta$ such that $\mathcal P_{\theta_0}$ and $\mathcal P_{\theta_1}$ are disjoint, consider testing a simple null hypothesis, $H_0: \theta=\theta_0$, against a simple alternative hypothesis, $H_1:\theta=\theta_1$. In \emph{complete} models, in which $G$ is singleton-valued, a well-defined reduced form induces a unique likelihood function \citep{tamer2003incomplete}.
In such settings, an optimal test is an LR test, as is well known from the Neyman--Pearson lemma. In incomplete models, however, the model generally admits a (non-singleton) set $\mathcal P_\theta$ of likelihoods, which prevents us from directly applying the Neyman--Pearson lemma.

In this setting, it is useful to consider \emph{minimax tests} \citep[see][Ch. 8 for the general principles]{Lehmann:2006aa}.
Let $\phi: S \mapsto [0,1]$ denote a (possibly randomized) test. For each $P$ on $(S,\Sigma_S)$, the rejection probability of $\phi$ is
\begin{equation}
	E_{P}[\phi(s)]=\int \phi(s)dP.
\end{equation}
Let $\pi_{\theta_1}(\phi)\equiv\inf_{P_1 \in \mathcal{P}_{\theta_1}}E_P [\phi(s)]$ be the \emph{power guarantee} of $\phi$ under $\theta_1$, which is the power value certain to be obtained regardless of the unknown selection mechanism. We then call test $\phi$ a \emph{level-$\alpha$ minimax test} if it satisfies the following conditions:
\begin{equation}
	\sup_{P \in \mathcal{P}_{\theta_0}}E_P[\phi(s)]\le \alpha~,
\label{eq:size} \end{equation}
and
\begin{equation}
	\pi_{\theta_1}(\phi)\ge \pi_{\theta_1}(\tilde\phi),~\forall \tilde \phi \text{ satisfying }\eqref{eq:size}.
\label{eq:power} \end{equation}
The condition in \eqref{eq:size} imposes uniform size control. In \eqref{eq:power}, tests are ranked by their guaranteed power, reflecting a preference for robust performance across selections.

A belief function (and its conjugate) is a special case of a two-monotone (and two-alternating) capacity, a class central to robust inference \citep{Huber:1981aa}.\footnote{Capacity $\nu$ is said to be \emph{monotone of order $k$} or, for short, \emph{k-monotone} if for any $A_i\subset S,i=1\cdots,k$,
\begin{align}
	\nu\big(\cup_{i=1}^k A_i\big) \ge \sum_{I\subseteq\{1,\cdots,k\}, I\ne \emptyset}(-1)^{|I|+1}\nu\big(\cap_{i\in I}A_i\big).
\end{align}
Conjugate $\nu^*(A)=1-\nu(A^c)$ is then called a \emph{$k$-alternating} capacity.} For a class of models whose lower probabilities are two-monotone, HS showed that the rejection region of a minimax test takes the form $\{s:\Lambda(s)>t\}$ for a measurable function, $\Lambda:S\to\mathbb R$, which they called the Radon--Nikodym derivative of $\nu^*_{\theta_1}$ with respect to $\nu^*_{\theta_0}$.
Further, they showed that there exists an LFP of distributions $(Q_0,Q_1)\in\mathcal{P}_{\theta_0}\times\mathcal{P}_{\theta_1}$ such that for all $t\in\mathbb{R}_+$,
\begin{equation}\label{p:lfp*}
	Q_0(\Lambda>t)=\nu_{\theta_0}^{*}(\Lambda>t),
\end{equation}
and
\begin{equation}\label{p:lfp}
	Q_1(\Lambda>t)=\nu_{\theta_1}(\Lambda>t),
\end{equation}
where $\Lambda$ can be taken to be a version of the Radon--Nikodym derivative:
\begin{align}
	\frac{dQ_1}{dQ_0}=\Big\{\frac{q_1}{q_0}:q_j\in \frac{dQ_j}{d\upsilon},~q_j\ge 0,~j=0,1,~q_0+q_1>0\Big\},
\end{align}
where $\upsilon$ is a measure that dominates $Q_j,j=0,1$. Below, we take $\upsilon$ to be the counting measure.

Heuristically, $Q_0$ is the null-consistent distribution under which the test attains its maximal size, while $Q_1$ is the alternative-consistent distribution that is least favorable for power. The following extension of the Neyman–Pearson lemma, adapted to our setting, then follows from HS.
\begin{lemma}\label{thm:neyman_pearson}
	Let $\mathcal{P}_{\theta_0}$ and $\mathcal{P}_{\theta_1}$ be defined as in \eqref{eq:P_theta} with $\theta=\theta_0$ and $\theta=\theta_1$, respectively.
	Then, there is a level-$\alpha$ minimax test $\phi$: $S \rightarrow [0,1]$ such that
	\begin{align}
		\phi(s) = \left\{
		\begin{array}{ll}
			1 & \text{if} \quad \Lambda(s)>C\\
			\gamma & \text{if} \quad \Lambda(s)=C\\
			0 & \text{if} \quad \Lambda(s)<C,\\
		\end{array}
		\right.
	\end{align}
	where $\Lambda\in dQ_1/dQ_0$ is a version of the Radon--Nikodym derivative of the LFP $(Q_0,Q_1)\in\mathcal P_{\theta_0}\times\mathcal P_{\theta_1}$, and $(C,\gamma)$ solves $E_{Q_0}[\phi(s)]=\alpha$.
\end{lemma}
Lemma \ref{thm:neyman_pearson} characterizes a level-$\alpha$ minimax test as an LR test based on the LFP. Intuitively, a large value of the likelihood ratio provides evidence against the null, and the lemma shows that rejecting for sufficiently large values of this ratio is minimax optimal.
 This is an existence and characterization result useful for obtaining the more general results below. We discuss the computational aspects in Section \ref{ssec:computation}.

\subsubsection*{Testability of Hypotheses:}
Before proceeding further, we comment on the testability of the hypotheses. The theory of HS requires that $\mathcal P_{\theta_0}$ and $\mathcal P_{\theta_1}$ are disjoint.
Otherwise, any test is vacuous from the minimax viewpoint because probability distribution $P\in \mathcal P_{\theta_0}\cap \mathcal P_{\theta_1}$ is consistent with both hypotheses. If this is the case, we say $\theta_1$ is not \emph{robustly testable} relative to $\theta_0$ because the power guarantee of any level-$\alpha$ test cannot exceed $\alpha$.
One should therefore expect nontrivial power guarantee only if an alternative hypothesis induces set $\mathcal P_{\theta_1}$ that does not intersect with $\mathcal P_{\theta_0}$.
\setcounter{example}{0}
\begin{example}[Binary response game (continued)]
Consider testing $H_0:\theta=0$ against $H_1:\theta=\theta_1$ with $\theta_1<0$. Suppose $u$ is continuously distributed over $\mathbb R^2$ with strictly positive density. Let $\bar A=\{(1,1)\}$. Then, $\nu^*_{\theta_1}(\bar A)<\nu_{0}(\bar A)$, ensuring $\mathcal P_{0}\cap \mathcal P_{\theta_1}=\emptyset$ for any $\theta_1<0$.
\end{example}

The lack of robust testability is also related to the notion of \emph{observational equivalence} \citep[see][and references therein]{Chesher:2014aa}. Let $s$ follow distribution $P$ and suppose $P$ is known. Consider parameter values $\theta,\theta'\in\Theta$ such that $\theta\ne\theta'$, $P\in \mathcal P_{\theta}$, and $P\in \mathcal P_{\theta'}$. In other words, the true distribution can be justified by structure $\theta$ augmented with some selection or by another structure $\theta'$ (again augmented with some selection). When this holds, $\theta$ and $\theta'$ are said to be observationally equivalent with respect to $P$. In incomplete models, $P$ is not in general identifiable, as the sampling process does not necessarily reveal it even asymptotically \citep{Maccheroni:2005aa,Epstein:2016qv}. Following \cite{Chesher:2014aa}, we say that $\theta$ and $\theta'$ are \emph{potentially observationally equivalent} if two structures are observationally equivalent for some $P$.
Clearly, any pair of potentially observationally equivalent parameter values are not robustly testable, as $\mathcal P_{\theta}$ and $\mathcal P_{\theta'}$ share a distribution in common. We examine the practical consequences of this feature in the Monte Carlo experiments below.

\subsection{Minimax Tests in Repeated Experiments}\label{sec:repeated}
Consider a sequence of outcomes $s^n=(s_1,s_2,\cdots,s_n)$ generated from repeated experiments. We present a set of results that characterize a minimax test in such a setting and provide an asymptotic Gaussian approximation to its (upper) rejection probability.

For any set $A$, let $A^n$ denote the $n$-fold Cartesian product of $A$.  For each $n\in\mathbb N$, let $S^n$ and $U^n$ be the sets of outcome sequences $s^n=(s_1,s_2,\cdots,s_n)$ and latent variable sequences $u^n=(u_1,\cdots,u_n)$, respectively. Below, we use the abbreviation $m^n$ to denote the family $\{m^n_{\theta}\}_{\theta \in \Theta}$ of $u^n$'s joint laws permitted by the model. We then make the following assumption on each member of this family.
\begin{assumption}\label{as:csiid}
	For each $\theta\in\Theta$, $m^n_\theta\in\Delta(U^n)$ is a product measure.
\end{assumption}
This assumption requires that $u_i$'s are distributed independently across experiments. A leading case is that $(u_1,\dots,u_n)$ are i.i.d. This can also accommodate heteroskedasticity and other types of heterogeneity across cross-sectional units or clusters of them (e.g., group-specific effects).

The Huber--Strassen characterization provides an associated likelihood-ratio statistic $\Lambda_i$ satisfying
\begin{align}
Q_{0,i}(\Lambda_i>t)
&=
\nu^*_{\theta_0,i}(\Lambda_i>t),
\label{eq:single-null-tail}\\
Q_{1,i}(\Lambda_i>t)
&=
\nu_{\theta_1,i}(\Lambda_i>t),
\label{eq:single-alt-tail}
\end{align}
for all $t\geq0$. In what follows, we assume that $Q_{1,i}\ll Q_{0,i}$ and a finite-valued $\Lambda_i$ is selected.

Without further assumptions, $s^n$ takes values in the Cartesian product of the sets of permissible outcome values:
\begin{align}
	G^n(u^n|\theta)=\prod_{i=1}^n G(u_i|\theta),
\label{eq:ginf} \end{align}
where $G(\cdot|\theta)$ is given as in \eqref{eq:defG}.\footnote{It is possible to allow the functional form of $G(u_i|\theta)$ to vary across $i$ as well. For notational simplicity, we do not explicitly consider this extension here. We introduce the heterogeneity of $G$ due to covariates in Section \ref{sec:covariates}.} This set collects outcome sequences that are compatible with the model and $\theta$. We represent the repeated experiments by the tuple $(S^n,U^n,\Theta, G^n;m^n).$

For each $\theta\in\Theta$ and $n\in \mathbb N$, the set of distributions compatible with the model is
\begin{align}
	\mathcal P^n_\theta=\Big\{P\in \Delta(S^n):P(A)=\int \eta^n(A|u^n) dm^n_\theta(u^n), ~\text{ for some } \eta^n \in \mathfrak K_{m^n_\theta}(G^n(\cdot|\theta)), A\subseteq S^n\Big\}.
\label{eq:P_thetainfty} \end{align}
This set collects all distributions of $s^n$ consistent with $\theta$. $\eta^n$ is unrestricted in the sense that the selection may be heterogeneous and dependent across experiments. Hence, $\mathcal P^n_\theta$ contains a broad range of distributions that can exhibit arbitrary dependence and heterogeneity. In particular, $\mathcal P^n_\theta$ allows measures under which the distributions of sample moments are not well approximated by classical limit theorems---even in large samples (see \cite{Epstein:2016qv}).

Finding the LFP in such a rich set of distributions may be challenging. However, under Assumption \ref{as:csiid} and with the correspondence in \eqref{eq:ginf}, the model has a tractable ``product'' structure, which significantly simplifies the characterization of the LFPs.
Let $\nu_\theta^{n,*}$ and $\nu_\theta^n$ denote the upper and lower probabilities of $\mathcal P_\theta^n$. For each $i\in\{1,\dots,n\}$ and $\theta\in\Theta$, let
\begin{align}
	\mathcal{P}_{\theta,i}\equiv \Big\{P\in \Delta(S): P(A)=\int_{U} \eta_i(A|u) dm_{\theta,i}(u), \text{ for some }\eta_i \in \mathfrak K_{m_{\theta,i}}(G(\cdot|\theta)), A\subseteq S \Big\},
\label{eq:P_theta_i} \end{align}
where $m_{\theta,i}$ is the $i$-th marginal distribution of $m_\theta^n.$  The following theorem shows that the minimax test in repeated experiments is an LR test and that the LFPs are product measures.

\begin{theorem}\label{thm:cs}	
	Suppose Assumption \ref{as:csiid} holds. Suppose, for each $i$, $Q_{1,i}\ll Q_{0,i}$, and a finite-valued $\Lambda_i$ is selected.
    Then, (i) an LFP $(Q_{0}^n,Q_{1}^n)\in\mathcal P_{\theta_0}^n\times\mathcal P_{\theta_1}^n$ exists such that for all $t\in\mathbb R_+$,
	\begin{align}
	\nu^{*,n}_{\theta_0}(\Lambda_n>t)&=Q_0^n(\Lambda_n>t)\\	
	\nu^{n}_{\theta_1}(\Lambda_n>t)&=Q_1^n(\Lambda_n>t),	
	\end{align}
where $\Lambda_n$ is a version of the Radon--Nikodym derivative $dQ_1^n/dQ_0^n$. The LFP consists of the product measures:
\begin{align}
	Q^n_0=\bigotimes_{i=1}^n Q_{0,i},~~\text{ and }~~Q^n_1=\bigotimes_{i=1}^n Q_{1,i},
\end{align}
where, for each $i=1,\dots,n$, $(Q_{0,i},Q_{1,i})\in \mathcal P_{\theta_0,i}\times\mathcal P_{\theta_1,i}$ is the LFP in the $i$-th experiment:

\noindent
(ii) A minimax test $\phi_n$: $S^n \rightarrow [0,1]$ can be constructed as
	\begin{align}
		\phi_n(s^n) = \left\{
		\begin{array}{ll}
			1 & \text{if} \quad \Lambda_n(s^n)>C_n\\
			\gamma_n & \text{if} \quad \Lambda_n(s^n)=C_n\\
			0 & \text{if} \quad \Lambda_n(s^n)<C_n,\\
		\end{array}
		\right. ~~\text{with}~~~ 		\Lambda_n(s^n)=\prod_{i=1}^{n} \Lambda_i,
	\label{eq:pi-n} \end{align}
where $\Lambda_i\in  dQ_{1,i}/dQ_{0,i}$ for all $i$, and $C_n$ and $\gamma_n$ are chosen so that $E_{Q_0^n}[\phi_n(s^n)]=\alpha$.
\end{theorem}
The LFP consists of product measures.\footnote{This result does not follow from Corollary 4.2 of HS, who assumed that a sample is independently distributed (p. 258).} Heuristically, this means that either for controlling size or maximizing power, the least favorable distribution in $\mathcal P^n_{\theta_0}$ (or $\mathcal P^n_{\theta_1}$) is a law that multiplies up the least favorable distributions in the individual experiments. 

When the $u_i$'s are i.i.d., this characterization has a particularly useful implication for the implementation. To construct a minimax test, it suffices to find the LFP $(Q_0,Q_1)\in\mathcal P_{\theta_0}\times\mathcal P_{\theta_1}$ in a \emph{single} experiment $(S,U,\Theta, G,m)$. One may then obtain the LR statistic by taking the product of their ratios across experiments. We state this result as a corollary.

\begin{corollary}\label{cor:cs}
	Suppose $(u_1,\dots,u_n)$ are identically and independently distributed. Then, a minimax test $\phi_n$: $S^n \rightarrow [0,1]$ can be constructed as
	\begin{align}
		\phi_n(s^n) = \left\{
		\begin{array}{ll}
			1 & \text{if} \quad \Lambda_n(s^n)>C_n\\
			\gamma_n & \text{if} \quad \Lambda_n(s^n)=C_n\\
			0 & \text{if} \quad \Lambda_n(s^n)<C_n,\\
		\end{array}
		\right. ~~\text{with}~~~ 		\Lambda_n(s^n)=\prod_{i=1}^{n}\Lambda_i,
 \end{align}
where $\Lambda_i\in  dQ_{1}/dQ_{0}$, $(Q_0,Q_1)\in\mathcal P_{\theta_0}\times \mathcal P_{\theta_1}$ is the LFP in $(S,U,\Theta, G,m)$, and $C_n$ and $\gamma_n$ are chosen so that $E_{Q_0^n}[\phi_n(s^n)]=\alpha$.
\end{corollary}

The LFP consisting of the product measures in Theorem \ref{thm:cs} provides an important link through which we may connect the incomplete model to standard frameworks. Below, we demonstrate this by studying the large-sample approximations. The Gaussian approximation can also be analyzed without identical latent distributions, but we maintain the i.i.d. case to simplify notation.\footnote{If $u_i$ is not identically distributed, one should invoke a central limit theorem for independent and not identically distributed (i.n.i.d.) sequences under $Q_0^n$ \citep[e.g.,][]{White:2001aa} to obtain a Gaussian approximation. 
}

\subsubsection*{Gaussian Approximation}
One of the consequences of the product structure is that the upper rejection probability of $\phi_n$ admits a Gaussian approximation in large samples. For ease of exposition, we assume that the $u_i$'s are i.i.d., which in turn implies that $Q_0^n$ is an i.i.d. law from Corollary \ref{cor:cs}. Hence, properly normalized sample averages follow classical limit theorems under this law. 
We use this insight to obtain an asymptotically valid critical value. Since $Q_0^n$ is the least favorable, the asymptotic size is controlled under any distribution under the null hypothesis.

Let $z_\alpha$ be the $1-\alpha$ quantile of the standard normal distribution and let $\Lambda\in dQ_1/dQ_0$. Throughout this subsection, we assume that the least favorable pair is mutually absolutely continuous and use a likelihood-ratio version such that $0<\underline\lambda
    \leq \Lambda(s)
    \leq \bar\lambda<\infty, Q_0-a.s.$ for some $\underline\lambda, \bar \lambda$. 
Since $S$ is finite, this implies $E_{Q_0}[(\ln\Lambda)^2]<\infty$.

For each $n$, let
\begin{align}
	C^*_n\equiv\exp\big(n \mu_{Q_0}+\sqrt n z_\alpha\sigma_{Q_0}\big),\label{eq:gauss_cv}
\end{align}
where $\mu_{Q_0}\equiv E_{Q_0}[\ln \Lambda(s)],$ and $\sigma^2_{Q_0}\equiv\text{Var}_{Q_0}(\ln\Lambda(s))$. Observe that $\mu_{Q_0}$ and $\sigma_{Q_0}$ depend only on the LFP but not on the unknown DGP.
Once the LFP is found, computing $\mu_{Q_0}$ and $\sigma_{Q_0}$ is straightforward because $Q_0$ is a discrete distribution and $\Lambda$ is known.
The critical value in \eqref{eq:gauss_cv} is constructed in such a way that the following convergence holds:
\begin{align}
	\sup_{P^n\in\mathcal P^n_{\theta_0}}P^n\big(\Lambda_n(s^n)> C_n^*\big)=Q^n_{0}\big(\Lambda_n(s^n)> C_n^*\big)
	\to \text{Pr}\big(Z> z_\alpha\big)= \alpha,
\end{align}
where $Z$ is a standard normal random variable. This critical value is computed without any resampling or simulation and therefore can be done so easily. Despite its simplicity, it has the advantage of being asymptotically valid even if the true DGP is highly heterogeneous and dependent.

Let $\phi^*_n$ be a test that rejects the null hypothesis if and only if $\Lambda_n>C^*_n$. The following proposition then follows.
\begin{proposition}\label{prop:asymptotics}
	Suppose the conditions of Corollary~\ref{cor:cs} hold and $E_{Q_0}[(\ln\Lambda)^2]<\infty$ holds. Then, the test controls the asymptotic size:
	\begin{align}
		\limsup_{n\to\infty}\sup_{P\in\mathcal P^n_{\theta_0}}E_{P}[\phi^*_n(s^n)]\le \alpha. \label{eq:asycontrol}
	\end{align}
	Furthermore, \eqref{eq:asycontrol} holds with equality when $\sigma^2_{Q_0}>0$.
\end{proposition}

\subsection{Computing Least Favorable Pairs}\label{ssec:computation}
A key step toward implementing our tests is the computation of the LFPs, in which the sharp identifying restrictions play a role. 
Let $H:[0,1]\to\mathbb R$ be a twice-continuously differentiable convex function. Our proposal is to find the LFP through the following characterization:
\begin{align}\label{eq:optimization}
	(Q_0,Q_1)=\argmin_{(P_0,P_1)\in \Delta(S)^2}&~ \int H\big(\frac{dP_0}{d(P_0+P_1)}\big)d(P_0+P_1)\\
	s.t.&~ \nu_{\theta_0}(A)\le P_0(A),~A\subset S\notag\\
	&~ \nu_{\theta_1}(A)\le P_1(A),~A\subset S,\notag
\end{align}
where the constraints on $(P_0,P_1)$ are the sharp identifying restrictions.\footnote{An alternative approach would be to use the sharp identifying restrictions of \cite{beresteanu2011sharp}, which also yield a finite number of linear restrictions. While we do not pursue that here, the insights presented in this paper may be useful for constructing optimal tests in models with endogeneity. Such models are studied by \cite{Chesher:2014aa}, who obtained sharp identifying restrictions using generalized instrumental variables.} The number of restrictions can be reduced further by restricting the class of events to the \emph{core determining class} \citep[see][]{galichon2011set,Luo:2017aa}.
This is a convex program with a convex objective function and linear constraints.\footnote{The convexity of the objective function follows from the convexity of the perspective $g(x,t)= tH(x/t)$ on its domain \citep[][Sec. 3.2.6]{Boyd:2004jk}.}

Our proposal builds on Theorem 6.1 in HS, which characterizes the LFP as a solution to a more general and abstract optimization problem over $\mathcal P_{\theta_0}\times\mathcal P_{\theta_1}$. Since $\mathcal P_\theta$ is defined through an unknown selection (see \eqref{eq:P_theta}), these constraints are not directly tractable. Rewriting them in terms of the sharp identifying restrictions yields a convex program with linear constraints that can be solved efficiently \citep[e.g.,][]{Boyd:2004jk}. 

\begin{remark}
	\rm Since $S$ is finite, the program in \eqref{eq:optimization} can be simplified further. Let $p_0$ denote the probability mass function of $P_0\in \mathcal P_{\theta_0}$ and $p_1$ be defined similarly. For simplicity, suppose $p_0(s)>0$ for all $s\in S$ and let $H(x)=-\ln x$. Then, one may solve
	\begin{align}
		(q_0,q_1)=\argmin_{(p_0,p_1)\in \Delta(S)^2}&~\sum_{s\in S}\ln\Big(\frac{p_0(s)+p_1(s)}{p_0(s)}\Big)(p_0(s)+p_1(s)) \label{eq:opt_ex1}\\
		s.t. &~~~\nu_{\theta_0}(A)\le \sum_{s\in A}p_0(s),~A\subset S \notag\\
		& ~~~\nu_{\theta_1}(A)\le \sum_{s\in A}p_1(s),~A\subset S \notag.
	\end{align}
	In this finite-dimensional convex program, one minimizes Kullback--Leibler divergence $D_{KL}(p_0+p_1\|p_0)$ subject to linear constraints on $(p_0,p_1)$. One may then use efficient numerical solvers (e.g., \verb1CVX1) to obtain the LFP.
\end{remark}

We illustrate the computation of an LFP and minimax test using Example \ref{ex:game}.
\setcounter{example}{0}
\begin{example}
	[Binary response game (continued)]\rm Let $0<\alpha<1/2$. Consider testing $H_0:\theta=0$ against $H_1: \theta=\theta_1$, where $\theta_1^{(k)}<0,k=1,2$. Suppose that $u$ follows the standard bivariate normal distribution $N(0,I_2)$.

It is straightforward to calculate $\nu_\theta(A)$.

For example, let $A=\{(1,0),(1,1)\}$. From \eqref{eq:defG} and \eqref{eq:choquet_eq}, 
	\begin{align}
		\nu_\theta\big(\{(1,0),(1,1)\}\big)=m_\theta\big(G(u|\theta)\subseteq \{(1,0),(1,1)\}\big) 
		=\frac{1}{4}+\frac{\Phi(\theta^{(1)})}{2},
	\label{eq:belcalc} \end{align}
where $\Phi$ is the CDF of a standard normal random variable (see Table \ref{Table:Probbound} in Online Supplement \ref{sec:implementation} for $\nu_\theta(A)$ for other events). In more complex models, simulation-based methods can be used \citep{galichon2011set,ciliberto2009market,Epstein:2016qv}. 
	
	Suppose that $\Phi(\theta^{(k)}_1)(1-\Phi(\theta^{(-k)}_1)) \le \frac{1}{4}$ for $k=1,2$.\footnote{The form of the minimax test below depends on the relative magnitude of $\theta^{(1)}_1$ and $\theta^{(2)}_1$. See Proposition \ref{prop:entry_game} in the Online Supplement for the full description of the minimax test in Example \ref{ex:game}.} Solving \eqref{eq:opt_ex1}, we obtain the following probability mass functions of the LFP $(Q_0,Q_1)$:
	\begin{align}
		(q_0(0,0),q_0(1,1),q_0(1,0),q_0(0,1))&=\Big(\frac{1}{4},\frac{1}{4}, \frac{1}{4}, \frac{1}{4}\Big)\\
		(q_1(0,0),q_1(1,1),q_1(1,0),q_1(0,1))&=\Big(\frac{1}{4},\Phi(\theta_1^{(1)})\Phi(\theta_1^{(2)}), \frac{3-4\Phi(\theta_1^{(1)})\Phi(\theta_1^{(2)})}{8},\frac{3-4\Phi(\theta_1^{(1)})\Phi(\theta_1^{(2)})}{8}\Big).
	\end{align}
	The LR statistic $\Lambda$ is then given by
	\begin{align}
		\Lambda(s)&=
		\begin{cases}
			1& s=(0,0)\\
			4\Phi(\theta^{(1)})\Phi(\theta^{(2)})& s=(1,1)\\
			\frac{3-4\Phi(\theta_1^{(1)})\Phi(\theta_1^{(2)})}{2}&s=(1,0)\\
			\frac{3-4\Phi(\theta_1^{(1)})\Phi(\theta_1^{(2)})}{2}&s=(0,1).
		\end{cases}
	\end{align}
	An LR test based on $\Lambda$ is level-$\alpha$ when $C=\frac{3-4\Phi(\theta_1^{(1)})\Phi(\theta_1^{(2)})}{2}$ and $\gamma=2\alpha$. The test therefore reduces to
	\begin{align}
		\phi(s)=
		\begin{cases}
			2\alpha & s=(1,0) \text{ or }(0,1)\\
			0& \text{otherwise}.
		\end{cases}
		\label{eq:minmax_test_ex1} \end{align}

The intuition is straightforward. Under $H_0$, the model is complete and the four outcomes occur with equal probability because $u\sim N(0,I_2)$ (Figure \ref{fig:levelsets}, left). Under $H_1$, a region of incompleteness emerges in which both $(1,0)$ and $(0,1)$ are equilibria. Although the selection is unrestricted, the model implies a larger probability of observing $s\in\{(1,0),(0,1)\}$ than under $H_0$ (Figure \ref{fig:levelsets}, right). The robust LR test therefore treats these outcomes as evidence of strategic interaction and rejects accordingly, without requiring knowledge of the selection mechanism.

    \begin{figure}
		[htbp] \small
		\begin{center}
			\caption{Level sets of $G$ under $H_0$ (left) and $H_1$ (right)} \label{fig:levelsets}
			\begin{tikzpicture}
				[scale=0.9,domain=-3:3,>=latex] 
				\fill[fill=red,opacity=0.4] (0,0) -- (3,0) -- (3,3) -- (0,3) ; 
				\fill[fill=green,opacity=0.4] (0,0) -- (3,0) -- (3,-3) -- (0,-3) ; 
				\fill[fill=green,opacity=0.4] (0,0) -- (-3,0) -- (-3,3) -- (0,3) ;
				\draw[->] (-3,0) -- (3,0) node[right] {$u_1$}; \draw[->] (0,-3) -- (0,3) node[above] {$u_2$}; \draw (0,-0.275)node[right]{$\theta=(0,0)$} ; \draw (1,1.2)node[right]{$\{(1,1)\}$} ; \draw (1,-1.2)node[right]{$\{(1,0)\}$} ; \draw (-1,1.2)node[left]{$\{(0,1)\}$} ; \draw (-1,-1.2)node[left]{$\{(0,0)\}$} ; \draw (0,0) circle [radius = 0.035];
			\end{tikzpicture}
			\hspace{0.2in}
			\begin{tikzpicture}
				[scale=0.9,domain=-3:3,>=latex] 
				\fill[fill=red,opacity=0.4] (2,1.5) -- (3,1.5) -- (3,3) -- (2,3) ; 
				\fill[fill=green,opacity=0.4] (3,-3) -- (0,-3) -- (0,0) -- (-3,0) -- (-3,3) --  (0,3) -- (2,3) -- (2,0) -- (2,1.5) -- (3,1.5) ; \draw[-] (-3,0) -- (2,0) ; 
				\draw[->] (2,0) -- (2,3) node[above] {$u_2$}; \draw[->] (0,1.5) -- (3,1.5) node[right] {$u_1$}; 
				\draw[-] (0,-3) -- (0,1.5) ; \draw (1.9,2.1) node [right]{$\{(1,1)\}$} ; \draw (1,-1.2) node [right]{$\{(1,0)\}$} ; \draw (-1,1.2) node [left]{$\{(0,1)\}$} ; \draw (-1,-1.2) node [left]{$\{(0,0)\}$} ; \draw (0.3,0.9) node [right]{$\{(0,1),$}; \draw (0.5,0.5) node [right]{$(1,0)\}$}; \draw (2,1.5) circle [radius = 0.035]; \draw (2,1.8)node[left]{$(-\theta^{(1)},-\theta^{(2)})$} ;
			\end{tikzpicture}
		\end{center}
		Note: The area in red represents the values of $u$ under which $\{(1,1)\}$ is predicted. The area in green represents the values of $u$ under which $\{(0,1)\}$, $\{(1,0)\}$, or their union is predicted.
	\end{figure}
\end{example}
\subsection{Testing against One-sided Alternatives}\label{ssec:one-sided}
This section formalizes a useful structure that arises in examples. Specifically, once an alternative direction is fixed, conditioning on a selection-invariant statistic can reduce the robust testing problem to a one-sided experiment with a monotone likelihood ratio. In that case, the one-sided likelihood-ratio test is not only exact under the least favorable null distribution but also \emph{uniformly most powerful} in terms of conditional power guarantee. The structure appears in testing the presence of strategic interaction effects in Example \ref{ex:game} and testing the share of individuals following a specific self-selection pattern in Example \ref{ex:roy}.

Fix a null value $\theta_0\in\Theta$ and a scalar-indexed path $\{\theta_w\in\Theta, w\in \mathcal W\subseteq \mathbb R\}$. Consider testing
\begin{align*}
  H_0:\theta=\theta_0
\qquad v.s. \qquad
H_1:\theta=\theta_w,\quad w\in \mathcal W_1,  
\end{align*}
where $w$ is a real-valued index that parameterizes movement along a fixed direction or path in the parameter space. The researcher chooses a one-sided subset $\mathcal W_1\subseteq\mathcal W$ to specify the alternatives toward which power is to be directed. Throughout, $\mathcal W_1$ is taken to be a collection of robustly testable alternatives in the sense that
\[
\cP^n_{\theta_0}\cap\cP^n_{\theta_w}=\emptyset,
\qquad
w\in\mathcal W_1.
\]
This choice does not require all alternatives to $\theta_0$ to be robustly testable. In particular, there may exist values $w\notin\mathcal W_1$ for which $\cP^n_{\theta_0}\cap\cP^n_{\theta_w}\neq\emptyset$,
so that no level-$\alpha$ test can guarantee power exceeding $\alpha$. Rather, $\mathcal W_1$ identifies the portion of the directed path on which a nontrivial robust power comparison is meaningful. In applications, $\mathcal W_1$ may take forms such as $(0,\bar w]$ for some $\bar w>0$, or $(w^\ast,\bar w]$ when robust testability holds only beyond a threshold $w^\ast>0$. The latter occurs in the Roy-model example, where alternatives with $w\leq w^\ast$ are not robustly testable, while those with $w>w^\ast$ are (see Section \ref{sec:MC_Roy}). 

For each $w\in\{0\}\cup\mathcal W_1$, let $\cP^n_{\theta_w}$
denote the set of laws of the observed outcome sequence compatible with $\theta_w$. 
For each $n$, let $\mathcal K_n$ be a measurable space, and let
\begin{align*}
    K_n:S^n\to\cK_n,
    \qquad
    T_n:S^n\to\mathbb R
\end{align*}
be, respectively, a \emph{conditioning statistic} and a \emph{conditional ordering statistic}.  In examples below, $K_n$ is selection-invariant and $T_n$ is the statistic that orders the least-favorable likelihood ratios.  For $k\in\cK_n$, write
$\Omega_{n,k}=\{s^n\in S^n:K_n(s^n)=k\}.$ For each $k$ and $w$, define the set of conditional laws $\mathcal L_P(\cdot|K_n=k)$ of $s^n$ on $\Omega_{n,k}$ by
\begin{align*}
 \cC^n_{\theta_w,k}
    \equiv
    \left\{
    \mathcal L_P(s^n\mid K_n=k):
    P\in\cP_{\theta_w}^n,\;P(K_n=k)>0
    \right\}.
\end{align*}
For $w\in\mathcal W_1$, define the conditional power guarantee of a test $\phi_n$ at $\theta_w$ and $k$ by
\begin{align*}
    \underline\pi^n_{\theta_w,k}(\phi_n)
    \equiv
    \inf_{P\in\cP^n_{\theta_w}:P(K_n=k)>0}
    E_P[\phi_n(s^n)\mid K_n=k].
\end{align*}

We make the following high-level assumptions, which we verify in the examples.
For
$c\in\mathbb R$ and $\gamma\in[0,1]$, define $\tau_{c,\gamma}(t)=1\{t>c\}+\gamma 1\{t=c\}$. The assumptions are stated for upper tailed tests. The lower-tail case is obtained by replacing $T_n$ with $-T_n$.
\begin{assumption}[Conditional least favorable laws]\label{ass:lfp}
For every $k\in\cK_n$ and $w\in\mathcal W_1$, there exist distributions
\[
    Q_{\theta_0,k}\in\cC^n_{\theta_0,k},
    \qquad
    Q_{\theta_w,k}\in\cC^n_{\theta_w,k}
\]
on $\Omega_{n,k}$ such that $Q_{\theta_0,k}$ does not depend on $w$. For every $c\in\mathbb R$ and $\gamma\in[0,1]$,
\begin{align*}
     \sup_{R\in\cC^n_{\theta_0,k}}
    E_R[\tau_{c,\gamma}\{T_n(s^n)\}]
    =
    E_{Q_{\theta_0,k}}[\tau_{c,\gamma}\{T_n(s^n)\}],   
\end{align*}
and 
\begin{align*}
        \inf_{R\in\cC^n_{\theta_w,k}}
    E_R[\tau_{c,\gamma}\{T_n(s^n)\}]
    =
    E_{Q_{\theta_w,k}}[\tau_{c,\gamma}\{T_n(s^n)\}].
\end{align*}
\end{assumption}
Thus, $Q_{\theta_0,k}$ is least favorable for conditional size control against upper-tail tests, and $Q_{\theta_w,k}$ is least favorable for maximizing power against $\theta_w$.

\begin{assumption}[Monotone likelihood ratios]\label{ass:mlr}
For every $k\in\cK_n$, the family $\{Q_{\theta_0,k}\}\cup\{Q_{\theta_w,k}:w\in\mathcal W_1\}$
is dominated by a common measure $\mu_k$ on $\Omega_{n,k}$.  Moreover, for every $w\in\mathcal W_1$, the likelihood ratio can be written as
\[
    \frac{dQ_{\theta_w,k}}{dQ_{\theta_0,k}}(s^n)
    =
    \lambda_{w,k}\{T_n(s^n)\},
\]
where $\lambda_{w,k}:\mathbb R\to\mathbb R_+$ is nondecreasing.  
\end{assumption}
This means that the conditional least favorable likelihood-ratio rejection regions are upper sets in the common statistic $T_n$ for all $w\in\mathcal W_1$.
Finally, for each $k\in\cK_n$, choose $c_k\in\mathbb R$ and $\gamma_k\in[0,1]$ such that
\begin{align}
     Q_{\theta_0,k}(T_n>c_k)+\gamma_k Q_{\theta_0,k}(T_n=c_k)=\alpha .
\label{eq:cv}   
\end{align}
The following proposition characterizes the exact conditional UMP test.
\begin{proposition}[Exact conditional UMP test on a directed path]\label{prop:cond-ump}
Suppose Assumptions \ref{ass:lfp}--\ref{ass:mlr} hold.  Define
\[
    \phi_n^*(s^n)
    =
    1\{T_n>c_{K_n}\}
    +
    \gamma_{K_n}1\{T_n=c_{K_n}\},
\]
where $(c_k,\gamma_k)$ satisfies \eqref{eq:cv}.  Then:
\begin{enumerate}[label=(\roman*)]
\item $\phi_n^*$ has exact conditional robust level $\alpha$:
\[
    \sup_{P\in\cP^n_{\theta_0}:P(K_n=k)>0}
    E_P[\phi_n^*(s^n)\mid K_n=k]
    =\alpha,
    \qquad
\]
for every $k$ such that \(\cC_{\theta_0,k}^n\) is nonempty.
\item For every $w\in\mathcal W_1$ and $k\in\cK_n$, its conditional power guarantee is
\[
    \underline\pi^n_{\theta_w,k}(\phi_n^*)
    =
    Q_{\theta_w,k}(T_n>c_k)+\gamma_k Q_{\theta_w,k}(T_n=c_k).
\]
\item $\phi_n^*$ is conditionally UMP in power guarantee on $\{\theta_w:w\in\mathcal W_1\}$.  That is, for every test $\psi_n$ with conditional level $\alpha$,
\[
    \underline\pi^n_{\theta_w,k}(\phi_n^*)
    \ge
    \underline\pi^n_{\theta_w,k}(\psi_n),
    \qquad
    \forall k\in\cK_n,
    \quad
    \forall w\in\mathcal W_1.
\]
\end{enumerate}
\end{proposition}
In the next section, we apply the proposition to two examples and examine their finite sample properties.

\section{Monte Carlo Experiments}\label{sec:montecarlo}
This section evaluates the finite-sample performance of the proposed tests using three Monte Carlo simulation designs. The designs highlight the robustness of the proposed tests relative to a naive method based on an i.i.d. selection assumption, test against one-sided alternatives, and the robust testability.

The first design illustrates the product-LFP result in repeated experiments: the robust LR test controls size under
selection mechanisms that induce dependence, whereas a naive LR test calibrated under i.i.d. selection may overreject. The second design studies the exact
conditional UMP test in a symmetric entry-game alternative. The third design
uses the Roy model to illustrate robust testability: some alternatives are not robustly testable and therefore admit no nontrivial power guarantee, while
separated alternatives do.
In all experiments, we evaluate empirical size and power at the nominal level $\alpha=0.05$ and $S=10,000$ simulation replications.

\subsection{Repeated Experiments in a Binary Response Game}
First, we simulate data from repeated experiments of binary response games in Example \ref{ex:game}.

For each experiment $i = 1, \dots, n$, the payoff shifters are given by $u_i = (u_i^{(1)}, u_i^{(2)})' \sim \mathcal{N}(0, I_2)$, and drawn independently across experiments. The strategic interaction parameters are $\theta = (\theta^{(1)}, \theta^{(2)})'$, where $\theta^{(1)} \le 0$ and $\theta^{(2)} \le 0$. 
We test 
\[
H_0:\theta=\theta_0=(-0.2, -1.2)
\qquad v.s.\qquad
H_1:\theta=\theta_1=(-1.2, -0.2).
\]

As discussed in Example \ref{ex:game}, the model is incomplete when both coordinates of $\theta$ are strictly negative, and the multiple equilibrium region is $0 \le u_i^{(1)} < -\theta^{(1)}$ and $0 \le u_i^{(2)} < -\theta^{(2)}$. In this region, the set of equilibrium outcomes is $\{(1,0), (0,1)\}$. Therefore, a selection mechanism is required to resolve multiplicity. 

We consider four types of selection mechanisms. The latent variables are independent across experiments, while the observed outcomes may be dependent because of the selection rule. The first mechanism is least favorable selection. Under $H_0$, outcomes are
generated according to the least favorable null law $Q_0^n$; under $H_1$,
they are generated according to the least favorable alternative law $Q_1^n$.
Equivalently, in each experiment we implement a selection rule that induces
the corresponding one-experiment LFP marginal, and then take the product law
across experiments. The second mechanism is i.i.d. selection. It draws $V_i \sim \mathrm{Bernoulli}(1/2)$ independently across $i$ and sets $s_i = (0,1)$ if $V_i = 1$ and $s_i = (1,0)$ if $V_i = 0$. The third mechanism is clustered selection. We partition observations into clusters of size $b_n = \max\{1, \lfloor \rho n \rfloor\}$ and draw a common selection shock $Z_g \sim \mathrm{Bernoulli}(1/2)$ for each cluster $g$. If $Z_g=1$, every multiple equilibrium observation in cluster $g$ selects $(0,1)$; otherwise, it selects $(1,0)$. The tuning parameter $\rho$ therefore controls the dependence strength by determining the fraction of observations affected by a common selection shock within a cluster. The fourth mechanism is common-shock selection. We draw a selection shock $Z \sim \mathrm{Bernoulli}(1/2)$ for the entire sample. If $Z=1$, all multiple equilibrium observations select $(0,1)$; otherwise, they select $(1,0)$. This is the limiting case of clustered selection with $\rho=1$.

Under these DGPs, we compare the robust LR test with a naive i.i.d. selection LR test. Let $(Q_0^n,Q_1^n)$ be the repeated-experiment LFP for testing
$\theta_0$ against $\theta_1$, and let
\[
    \ell_n(s^n)=\log\frac{dQ_1^n}{dQ_0^n}(s^n)
    =
    \sum_{i=1}^n \log\frac{dQ_1}{dQ_0}(s_i)
\]
in the i.i.d. latent-variable case. The robust LR test rejects for large
values of $\ell_n$. The exact critical value $c_{n,\alpha}$ and
randomization probability $\gamma_{n,\alpha}$ are chosen so that
\[
    Q_0^n(\ell_n>c_{n,\alpha})
    +
    \gamma_{n,\alpha}Q_0^n(\ell_n=c_{n,\alpha})
    =
    \alpha .
\]
The naive LR test completes the model by imposing i.i.d. selection in
the multiple-equilibrium region and calibrates the LR statistic under the
corresponding product null law.

Table \ref{tab:empirical-size-representative} reports empirical size under $H_0$ across different sample size and DGPs for both tests. The robust LR test has rejection probability close to the nominal level under least favorable selection. Under other selection mechanisms, it is conservative. By contrast, the naive LR test is fragile. It has a rejection probability approximately close to the nominal level under i.i.d. selection, but it overrejects under least favorable selection and other selection dependencies.

\begin{table}[htbp]
\centering
\caption{Empirical size under $H_0$: $\theta_0=(-0.2,-1.2)$}
\label{tab:empirical-size-representative}

\small
\resizebox{\textwidth}{!}{
\begin{tabular}{l c cc cc cc cc}
\toprule
& & \multicolumn{2}{c}{$n = 100$}
  & \multicolumn{2}{c}{$n = 250$}
  & \multicolumn{2}{c}{$n = 500$}
  & \multicolumn{2}{c}{$n = 1000$} \\
\cmidrule(lr){3-4} \cmidrule(lr){5-6} \cmidrule(lr){7-8} \cmidrule(lr){9-10}
Selection mechanism & Dependence strength
& Robust & Naive
& Robust & Naive
& Robust & Naive
& Robust & Naive \\
\midrule
Least favorable & --
& 0.047 & 0.100
& 0.049 & 0.147
& 0.049 & 0.213
& 0.050 & 0.310 \\

i.i.d. selection & $\rho = 0$
& 0.023 & 0.051
& 0.013 & 0.050
& 0.007 & 0.053
& 0.002 & 0.047 \\

Clustered selection & $\rho = 0.10$
& 0.019 & 0.047
& 0.015 & 0.056
& 0.009 & 0.057
& 0.003 & 0.061 \\

Common-shock selection & $\rho = 1.00$
& 0.032 & 0.064
& 0.029 & 0.082
& 0.024 & 0.103
& 0.028 & 0.166 \\
\bottomrule
\end{tabular}
}

\end{table}

Figure \ref{fig:size_dependence} shows the empirical rejection probability of the robust and naive LR test under $H_0$ as dependence strength $\rho$ varies. The robust LR test's size remains below the nominal level $\alpha$ across all values of $\rho$. It guarantees finite-sample validity under any clustered selection mechanisms. The naive LR test is approximately valid near the i.i.d. benchmark, but its rejection probability increases and starts over-rejecting as selection dependence strengthens.

Figure~\ref{fig:power} reports power of the robust LR test under
$\theta_1$ as the sample size varies. The least favorable curve is the
power guarantee delivered by the product-LFP test. The other curves correspond
to alternative selection mechanisms and lie weakly above the least favorable
curve, as expected. In this calibration, the power curves are close across
selection mechanisms, indicating that the selected non-LFP DGPs are only
moderately more favorable than the least favorable alternative. The figure
therefore illustrates the lower-power interpretation of the robust LR test:
the LFP curve is the guaranteed rejection probability, while other selections
may generate higher power.

\subsection{Exact Conditional UMP Test}\label{ssec:ump_game}
In the second design, we continue to use Example \ref{ex:game} and consider a directed family of alternatives indexed by a scalar parameter
$w\in\mathcal W_1$:
\[
\theta_w=(-w,-w),
\qquad w\in\mathcal W_1=(0,\infty).
\]
We consider testing
\[
H_0:\theta=\theta_0=(0,0)
\qquad\text{v.s.}\qquad
H_1:\theta=\theta_w,\quad w\in\mathcal W_1.
\]
For this directed family, the robust likelihood-ratio problem admits an exact conditional UMP test.\footnote{
The symmetry of the path is not a restriction on the entry-game model. More
generally, one could consider a directed path $w\mapsto\theta_w$, such as
$\theta_w=-w(\zeta_1,\zeta_2)$ for a fixed direction
$(\zeta_1,\zeta_2)$. The conditional ordering statistic may then depend on
the path. The same high-level argument can be verified on any subset
$\mathcal W_1$ over which the conditional least favorable likelihood ratios
are monotone in a common statistic (See Online Supplement \ref{ssec:ump_game_results}).
We use the symmetric path because it
gives the simplest statistic and interpretation.
} The key observation is that the unknown equilibrium-selection rule affects only the allocation of probability between the two monopoly outcomes $(1,0)$ and $(0,1)$. It does not affect the aggregate information contained in the number of monopoly and duopoly markets.

Let $K_n(s^n)$ be the number of markets with at least one entrant, and $N_{11}(s^n)$ be the number of duopoly markets. Under $H_0$, conditional on $K_n=k$,
\[
    N_{11}\mid K_n=k \sim \mathrm{Bin}(k,p_0),
    \qquad p_0=\frac13 .
\]
This conditional experiment is easy to interpret. After conditioning
on $K_n=k$, we look only at markets with at least one entrant. Each such market
is either a duopoly or a monopoly. Selection can determine which firm is the
monopolist, but it cannot change the monopoly--duopoly classification. Hence
the robustly informative feature is the duopoly share among active markets.
Under the null this share is $1/3$, whereas under
$\theta_w=(-w,-w)$ it is $4\Phi(-w)^2/3<1/3$. The exact conditional UMP test
therefore rejects when duopolies are too rare relative to the number of active
markets.

For each $k\in\cK_n$, choose $c_k\in\{0,1,\ldots,k\}$ and $\gamma_k\in[0,1]$ such that 
\[
    P_{B\sim \mathrm{Bin}(k,p_0)}(B<c_k)+\gamma_k P_{B\sim \mathrm{Bin}(k,p_0)}(B=c_k)=\alpha.
\]
Define
\begin{align*}
    \phi_n^{\mathrm{E}}(s^n)
    =
    1\{N_{11}(s^n)<c_{K_n(s^n)}\}
    +
    \gamma_{K_n(s^n)}1\{N_{11}(s^n)=c_{K_n(s^n)}\}.
\end{align*}
This is an upper-tail test in the number of duopoly markets $ T_n(s^n)=-N_{11}(s^n).$ 
This conditional likelihood-ratio test is exact and uniformly most powerful (see Corollary \ref{cor:entry-cump}).

\begin{table}[htbp]
\centering
\caption{Empirical Size of the UMP Test}
\label{tab:ump_size}
\resizebox{0.9\textwidth}{!}{
\begin{tabular}{lcccccccc}
\toprule
& \multicolumn{8}{c}{Sample size $n$} \\
\cmidrule(lr){2-9}
& 25 
& 50 
& 100 
& 150 
& 200 
& 300 
& 500 
& 1000 \\
\midrule
Exact conditional UMP 
& 0.0506 & 0.0520 & 0.0504 & 0.0497 
& 0.0500 & 0.0490 & 0.0489 & 0.0516 \\
\bottomrule
\end{tabular}
}
\end{table}
Table \ref{tab:ump_size} reports the empirical size of the exact conditional UMP test across different sample sizes. The robust likelihood ratio test controls its size across all sample sizes, with rejection frequencies around $\alpha$.

Figure \ref{fig: localpowerandumppower} plots the test's rejection probability against the local alternative parameter $h\geq0$, where $\theta_n(h)=-h/\sqrt{n}$ for $n=1000$. The black solid line represents the theoretical envelope $1-\Phi(z_{1-\alpha}-h/\sqrt{3/\pi})$ for the power guarantee. The simulated curves are nearly indistinguishable across selection mechanisms and closely track the theoretical local-power envelope. This is expected: conditional on $K_n$, the statistic uses only the number of duopoly markets,
which is unaffected by how the two monopoly equilibria are selected. The figure confirms the selection-invariance of the conditional experiment and the exact conditional UMP property along the directed path.

\subsection{Robust Testability in Binary Roy Model}\label{sec:MC_Roy}
The third design examines the consequences of robust testability using Example \ref{ex:roy}. For $i=1,...,n$, we draw latent type $u_i=(Y_{0i},Y_{1i})\in\{(0,0),(0,1),(1,0),(1,1)\}$ independent across experiments. For size, the latent types are drawn from $\theta_0=(\frac{1}{6},\frac{1}{2},\frac{1}{6})$. For alternative path, the latent types are drawn from 
\begin{align*}
\theta_w
    =\theta_0+\big(0,-\frac{w}{6},\frac{w}{6}\big)',
    \qquad w\in[0,3].
\end{align*}
Along this path, the share of the latent type $(Y_0,Y_1)=(0,1)$ decreases whereas that of $(Y_0,Y_1)=(1,0)$ increases, offsetting the change, while keeping the shares of the other two types unchanged. The alternatives are robustly testable relative to $\theta_0$ only after the path is sufficiently separated from the null. Specifically, we take $w\in \cW_1=(1,3]$ for which $\theta_w$ is robustly testable.\footnote{See Corollary \ref{cor:roy-general-cump} and Remark \ref{rem:calibration} in Appendix \ref{ssec:UMPs}.
The choice $\theta_0=(1/6,1/2,1/6)$ is made for notational simplicity. 
For a general null $\theta_0$, the null interval for $P\{(1,0)\}$ is
$[\theta_0^{(1,0)},\, 1-\theta_0^{(0,0)}-\theta_0^{(0,1)}].$
If the directed path keeps
$\theta_w^{(0,0)}=\theta_0^{(0,0)}$ and $\mathcal W_1$ is chosen so that
$\theta_w^{(1,0)}>
    1-\theta_0^{(0,0)}-\theta_0^{(0,1)},$
the same separation, least-favorable binomial experiment, and monotone
likelihood-ratio argument go through.}  For each experiment, if $u_i=(0,0)$, we split the observation equally between $s_i=(0,0)$ and $s_i=(0,1)$. If $u_i=(1,1)$, the observed outcome is selected according to a selection mechanism.

We consider five types of selection mechanisms. The first mechanism is the least favorable or null-mimicking selection. For
$0\leq w\leq1$, it chooses $s_i=(1, 0)$ with probability $r_{LF}(w)=1-w$ when $u_i=(1, 1)$, which makes the observable
distribution mimic the least favorable null. For $w>1$, it sets $r_{LF}(w)=0$, which is least favorable for the upper-tail test.
The second mechanism is i.i.d. selection, which selects $s_i=(1, 0)$ with probability 1/2. The third mechanism is clustered selection. Observations are partitioned into clusters of size $b = \lfloor 0.10n \rfloor$. For each cluster $g$, we draw a cluster-level shock $Z_g \sim \text{Bernoulli}(1/2)$. If $Z_g = 1$, for all observations in cluster $g$ with  $u_i=(1, 1)$, we select outcome $s_i=(1, 0)$; otherwise, we select $s_i=(1, 1)$. The fourth mechanism is common-shock selection. A single shock $Z \sim \text{Bernoulli}(1/2)$ is drawn for the entire sample, and we apply the same selection rule to the entire sample. The fifth mechanism always selects
$s_i=(1,0)$ when $u_i=(1,1)$. This mechanism is favorable for detecting $\theta_w$ (see below). It is useful for illustrating how large the rejection probability can be under
a selection rule that reveals the alternative.

Define
\[
    N_{10}(s^n)=\sum_{i=1}^n 1\{s_i=(1,0)\},
    \qquad
    N_{11}(s^n)=\sum_{i=1}^n 1\{s_i=(1,1)\},
\]
and let $K_n(s^n)=N_{10}(s^n)+N_{11}(s^n)$ be the number of individuals with $Y=1$. After conditioning on $K_n$, the testing problem becomes an upper one-sided test with the conditional ordering statistic $T_n(s^n)=N_{10}(s^n)$. For each $k\in\cK_n$, choose an integer $c_k\in\{0,1,\ldots,k\}$ and $\gamma_k\in[0,1]$ such that
\[
    P_{B\sim \mathrm{Bin}(k,2/5)}(B>c_k)
    +
    \gamma_kP_{B\sim \mathrm{Bin}(k,2/5)}(B=c_k)
    =\alpha.
\]
For $k=0$, this convention gives $c_0=0$ and $\gamma_0=\alpha$.  Define
\[
    \phi_n^{\mathrm R}(s^n)
    =
    1\{N_{10}(s^n)>c_{K_n(s^n)}\}
    +
    \gamma_{K_n(s^n)}1\{N_{10}(s^n)=c_{K_n(s^n)}\}.
\]
This test is exact conditionally UMP against the path $w\mapsto\theta_w$ defined above (see Corollary \ref{cor:roy-general-cump}). The conditional experiment again has a simple interpretation. Conditioning on
$K_n=k$ restricts attention to individuals with $Y=1$. Among these individuals,
the outcome is either $(Y,D)=(1,0)$ or $(1,1)$. The unknown selection rule can affect
how individuals with tied potential outcomes, $u_i=(1,1)$, are allocated between
these two outcomes, but it cannot change whether an individual belongs to the
$Y=1$ group. Thus, after conditioning on $K_n$, the robustly informative feature for detecting $\theta_w$
is the share of $(1,0)$ outcomes among individuals with $Y=1$. Under the least
favorable null, this share is $2/5$. Along the testable part of the path
$w>1$, even the least favorable alternative raises this conditional probability
to $(1+w)/5>2/5$. The exact conditional UMP test therefore rejects when
$(1,0)$ outcomes are unusually frequent among individuals with $Y=1$.

Table~\ref{tab:roy_ump} reports the empirical size of $\phi_n^R$.\footnote{Favorable selection coincides with least favorable selection under the null and is therefore omitted from the size table, but it is included in the power figures.} Under the least favorable null selection, the rejection probability is essentially equal to the nominal level for all sample sizes. Under the remaining selection mechanisms, the test is conservative, often substantially so. In particular, the i.i.d. and clustered midpoint selections yield rejection probabilities close to zero, whereas the common-shock midpoint selection yields rejection probabilities around $2.5\%$. This pattern is consistent with the minimax construction of the test.

Figure~\ref{fig:roy_power} highlights the distinction between apparent and robust power. For $w\leq 1$, the null and alternative sets of observable distributions overlap. Consequently, the null-mimicking selection keeps the rejection probability at the nominal level, confirming that no level-$\alpha$ test can guarantee lower power exceeding $\alpha$ in this region. Although other selection mechanisms may yield higher rejection probabilities, such gains reflect selection-dependent apparent power rather than power guarantee. In contrast, when $w>1$, the null and alternative sets are separated. In this region, the least favorable rejection probability increases with $w$ and coincides with the exact conditional lower-power envelope, demonstrating that the test achieves meaningful robust power once the alternative becomes robustly testable.

\begin{table}[htbp]
\centering
\caption{Empirical Size of the Exact Conditional UMP Test}
\label{tab:roy_ump}
\begin{threeparttable}
\begin{tabular}{lcccc}
\toprule
{Selection mechanism} 
& \textbf{$n=100$} 
& \textbf{$n=250$} 
& \textbf{$n=500$} 
& \textbf{$n=1000$} \\
\midrule
Least favorable selection
& 0.0489 & 0.0498 & 0.0486 & 0.0493 \\

i.i.d. selection 
& 0.0004 & 0.0000 & 0.0000 & 0.0000 \\

Clustered selection $(\rho=0.10)$ 
& 0.0008 & 0.0002 & 0.0004 & 0.0000 \\

Common-shock selection $(\rho=1)$ 
& 0.0250 & 0.0245 & 0.0263 & 0.0234 \\
\bottomrule
\end{tabular}
\begin{tablenotes}
\footnotesize
\item Notes: The nominal level is $\alpha=0.05$, and the number of Monte Carlo replications is $R=10{,}000$. 
The least favorable null selection sets $r=1$, which maximizes the conditional probability of observing 
$(Y,D)=(1,0)$ among observations with $Y=1$.
\end{tablenotes}
\end{threeparttable}
\end{table}

Figure~\ref{fig:roy_power_samplesize} reports rejection probabilities under the separated alternative $w^\ast=1.50$. The least favorable curve is the
guaranteed lower-power curve and coincides closely with the analytical lower-power expression. The i.i.d., clustered, common-shock, and favorable
selection curves lie above it. As the sample size grows, all curves approach one quickly, showing that once the null and alternative sets of admissible
observable distributions are separated, the exact conditional UMP test has strong finite-sample power.

\section{Extensions}\label{sec:extensions}
\subsection{Tests in the Presence of Nuisance Components}\label{sec:nuisance} We now consider testing the hypotheses on subcomponents of $\theta$. Let $\theta=(\beta',\delta')'\in \Theta_\beta\times\Theta_\delta$, where $\beta$ is a $k\times 1$-subvector of interest and $\delta$ is a $(d-k)\times 1$ vector of nuisance parameters. Consider the following hypotheses:
\begin{align}
	H_0:\beta=\beta_0,~\delta\in\Theta_\delta, ~~\text{v.s.}~~H_1:\beta\ne \beta_0,~\delta\in\Theta_\delta.
\label{eq:hypothesis2} \end{align}

In this setting, both hypotheses are composite in the structural parameters, and therefore Lemma~\ref{thm:neyman_pearson} is not directly applicable. Nevertheless, it serves as a useful building block for constructing tests with desirable optimality properties. We partition the parameter space as
\[
\Theta_0=\{\beta_0\}\times\Theta_\delta
\qquad\text{and}\qquad
\Theta_1=\{\beta:\beta\neq\beta_0\}\times\Theta_\delta.
\]

The researcher's action is binary $a\in \{0,1\}$. Define a loss function $\mathsf L:\Theta\times \{0,1\}\to\mathbb R_+$ by
\begin{align*}
	\mathsf L(\theta,a)\equiv aI_{\Theta_0}(\theta)+\zeta(1-a) I_{\Theta_1}(\theta),
\end{align*}
where $\zeta>0$. The loss from the Type-I error is normalized to 1. The trade-off between the Type-I and Type-II errors is determined by parameter $\zeta$.

For each test $\phi$ and $\theta\in\Theta$, define the \emph{upper risk} by
\begin{align}
	R(\theta,\phi)&=\max_{P\in \mathcal P_\theta}\int \phi(s)I_{\Theta_0}(\theta)+\zeta(1-\phi(s)) I_{\Theta_1}(\theta)dP(s)\notag\\
	&=\int \phi(s)d\nu^*_\theta(s)I_{\Theta_0}(\theta)+\zeta(1-\int\phi(s)d\nu_\theta(s))I_{\Theta_1}(\theta),
\label{eq:riskfun} \end{align}
where the integrals in \eqref{eq:riskfun} are Choquet integrals (see Appendix \ref{sec:capacities}).

The upper risk determines the trade-off between the size ($R_0(\theta,\phi)\equiv\sup_{P\in\mathcal P_\theta}\int \phi dP=\int\phi d\nu^*_\theta$ for $\theta\in\Theta_0$) and the power guarantee ($\inf_{P\in\mathcal P_\theta}\int \phi dP= \int\phi d\nu_\theta$ for $\theta\in\Theta_1$). What remains is to incorporate parameter uncertainty. For this, let $\mu$ be a (prior) probability distribution over $\Theta$. We write $\mu$ as $\mu=\tau\mu_0+(1-\tau)\mu_1,$ where $\tau\in (0,1)$ and $\mu_0,\mu_1$ are suitable probability measures supported on $\Theta_0$ and $\Theta_1$, respectively. Define
\begin{align}
	r(\mu,\phi)&\equiv \int_\Theta R(\theta,\phi)d\mu(\theta)\notag\\
	&=\tau\int\phi(s)d\kappa_0^*(s)+(1-\tau)\zeta(1-\int\phi(s)d\kappa_1(s)),
\label{eq:bdsrisk} \end{align}
where $\kappa_0^*=\int_{\Theta_0}\nu_\theta^*d\mu_0$ and $\kappa_1=\int_{\Theta_1}\nu_\theta d\mu_1$.\footnote{The second equality in \eqref{eq:bdsrisk} is established in the proof of Theorem \ref{thm:bayes-dempster-shafer}.} This risk function uses the prior probability to reflect parameter uncertainty, while it uses the belief function (and its conjugate) to incorporate the decision maker's willingness to be robust against incompleteness. In what follows, we call $r$ the \emph{Bayes--Dempster--Shafer (BDS) risk}. We then call $\phi$ a \emph{BDS test} if it minimizes the BDS risk.\footnote{The axiomatic foundations for this type of preference (when $S$ is the payoff-relevant state space) is given in \cite{Gul:2014aa} and \cite{epstein2015exchangeable} (in the context of repeated experiments).} 

One of the components of the BDS risk is $\pi_{\kappa_1}(\phi)=\int \phi d\kappa_1$. We call this object the \emph{weighted average power guarantee} (WAPG). It is analogous to the standard weighted average power of \citet{Andrews:1994aa,Andrews:1995aa}: the choice of $\mu_1$ determines which alternatives receive greater weight. Unlike standard weighted average power, however, $\pi_{\kappa_1}(\phi)$ averages power guarantees rather than power under a specified selection rule or data-generating distribution.

The following theorem characterizes the BDS test. For this, let $core(\kappa)\equiv \{P\in\Delta(S):P(A)\ge \kappa(A),\forall A\subset S\}$. In what follows, we assume that $core(\kappa_0)\cap core(\kappa_1)=\emptyset$.\footnote{To ensure this condition, it is sufficient to have at least one $\bar A\subset S$ such that $P_0(\bar A)<P_1(\bar A)$ (or $P_0(\bar A)>P_1(\bar A)$) for all $(P_0,P_1)\in \mathcal P_{\theta_0}\times\mathcal P_{\theta_1}$ and $(\theta_0,\theta_1)\in \Theta_0\times\Theta_1$. In Example \ref{ex:game}, one may take $\bar A=\{(1,1)\}$.}
\begin{lemma}\label{thm:bayes-dempster-shafer}
	Let the BDS risk be defined as in \eqref{eq:bdsrisk}.
	Then, there exists a BDS test such that, for any $\zeta>0$,
	\begin{align}
		\phi(s) = \left\{
		\begin{array}{ll}
			1 & \text{if} \quad \Lambda(s)>C\\
			\gamma & \text{if} \quad \Lambda(s)=C\\
			0 & \text{if} \quad \Lambda(s)<C,\\
		\end{array}\label{eq:LRBDS}
		\right.
	\end{align}
	where $C=\tau/\zeta(1-\tau)$, and $\Lambda$ is a version of $dQ_1/dQ_0$ for the LFP $(Q_0,Q_1)\in core(\kappa_0)\times core(\kappa_1)$ such that, for all $t\in\mathbb R_+$,
	\begin{align}
		Q_0(\Lambda>t)=\kappa_0^{*}(\Lambda>t),~~~\text{ and }~~~ Q_1(\Lambda>t)=\kappa_1(\Lambda>t).
	\end{align}
\end{lemma}
One may view this as an analog of Lemma \ref{thm:neyman_pearson}. A key difference is that the LFP belongs to the product of the cores of the capacities $\kappa_0$ and $\kappa_1$. Hence, $\kappa_0$ and $\kappa_1$ are both belief functions, which in turn allows us to compute the LFP in a tractable way.

The analysis above is useful for constructing optimal tests for minimizing risk. However, the BDS tests are not designed to control size uniformly over $\Theta_0$. Therefore, we also consider a test that controls size and maximizes the WAPG. For this, we fix $\mu_1$ throughout and follow the developments on the tests in the presence of nuisance parameters \citep{Chamberlain:2000aa,Elliott:2015aa,Moreira:2013aa}. 

The following minimax theorem characterizes the minimax test as a BDS test for the least favorable prior (if it exists). For this, let $\mathcal M(\mu_1)\equiv\{\mu:\mu=\tau\mu_0+(1-\tau)\mu_1,\mu_0\in\Delta(\Theta_0),\tau\in[0,1]\}$, where $\mu_1$ is fixed. In what follows, we drop $\mu_1$ from the argument of $\mathcal M$, but its dependence should be understood. We then let $\mathbf\Phi$ be the set of randomized tests.

\begin{theorem}\label{thm:minimax}
	Let the upper risk $R$ be defined as in \eqref{eq:riskfun}. Suppose that $\Theta$ is compact. Then,
	\begin{align}
		\sup_{\mu\in\mathcal M}\inf_{\phi\in\mathbf\Phi}\int_\Theta R(\theta,\phi)d\mu(\theta)=\inf_{\phi\in\mathbf\Phi}(\sup_{\theta\in\Theta_0}R_{0}(\theta,\phi)\vee R_1(\phi)),
	\label{eq:minimax} \end{align}
	where $R_{0}(\theta,\phi)=\int \phi(s)d\nu^*_\theta(s)I_{\Theta_0}(\theta)$ and $R_1(\phi)=\zeta(1-\pi_{\kappa_1}(\phi))$. Furthermore, there exists $\phi^\dagger$ that achieves equality in \eqref{eq:minimax}.
\end{theorem}

\begin{remark}\rm
Suppose that $\zeta$ is chosen so that the maximum BDS risk, given by the left-hand side of \eqref{eq:minimax}, equals $\alpha$. Then $\phi^\dagger$ is a level-$\alpha$ test that maximizes the WAPG. Moreover, the theorem implies that $\phi^\dagger$ can be approximated, in terms of risk, by a sequence of BDS tests $\{\phi_\ell\}_{\ell=1}^\infty$, where each $\phi_\ell$ is optimal under some prior $\mu_\ell$ and satisfies
\begin{align}
\int_\Theta R(\theta,\phi_\ell)\,d\mu_\ell(\theta)
\to
\sup_{\mu\in\mathcal M}
\int_\Theta R(\theta,\phi^\dagger)\,d\mu(\theta).
\end{align}
A numerical algorithm, in the spirit of \cite{Chamberlain:2000aa}, for constructing such sequences is studied by \cite{KaidoZhang2023}.
\end{remark}

\subsection{Covariates}\label{sec:covariates}
This section extends the base framework to incorporate observable covariates. Each individual experiment is described by $(S,X,U,G,\Theta;\upsilon,m)$, where $S,U,G,\Theta$ are defined as before. We let $X$ denote the finite set of covariate values and $\{\upsilon_\theta,\theta\in\Theta\}$ be a family of distributions on $X$. Throughout, we assume that each $\upsilon_\theta\in\Delta(X)$ has full support on $X$. Measure $m_\theta(\cdot|x)$ then determines the conditional law of $u$ given $x$. The prediction of the model is then summarized by a weakly measurable correspondence $(u,x)\mapsto F(u,x|\theta)\equiv \{(s,x):s\in G(u|\theta,x) \}\subset S\times X$ for each $\theta\in\Theta$. As before, this correspondence induces a belief function on $S\times X$
\begin{align}
	\nu_\theta(A)=\int 1\{F(u,x|\theta)\subset A\}dm_\theta(u|x)d\upsilon_\theta(x),~ A\subset S\times X.
\end{align}
If $A$ is a rectangle $A=A_s\times A_x$ for some $A_s\subset S$ and $A_x\subset X$, one may write it as
\begin{align}
	\nu_\theta(A)=\int_{A_x}m_\theta(G(u|\theta,x)\subset A_s|x)d\upsilon_\theta(x)=\int_{A_x}\nu_\theta(A_u|x)d\upsilon_\theta(x),
\end{align}
which can be viewed as the mean of the conditional belief function $\nu_\theta(\cdot|x)$. 
The subsequent analysis, starting with the Neyman--Pearson lemma, is then essentially the same as before.

A simplification occurs when $\upsilon$ does not depend on $\theta$.
Consider $\theta_0,\theta_1\in\Theta$ with $\theta_0\ne\theta_1$. Because of the additivity of $\upsilon$, it suffices to consider sets of the form $B\times \{x\}$, where $B\subset S$ and $x\in X.$
Then, the program that determines the LFP is
\begin{align}
	(Q_0,Q_1)=\argmin_{P_0,P_1\in \Delta(S\times X)}&~ \int H\big(\frac{dP_0}{d(P_0+P_1)}\big)d(P_0+P_1)\\
	s.t.&~ \nu_{\theta_0}(B\times \{x\})\le P_0(B\times \{x\}),~B\subset S, x\in X\notag\\
	&~ \nu_{\theta_1}(B\times \{x\})\le P_1(B\times \{x\}),~B\subset S, x\in X.\notag
\end{align}
Observe that the constraints simplify to
\begin{align}
\nu_{\theta_j}(B\times \{x\})=\nu_{\theta_j}(B|x)\upsilon(x)\le \sum_{s\in B}p_j(s|x)\upsilon_j(x),~j=0,1.
\end{align}
Taking $B=S$, one obtains $\upsilon(x)\le \upsilon_j(x)$ for all $x\in X$ with $j=0,1.$ Since $\upsilon$ is a measure, this implies $\upsilon_0=\upsilon_1=\upsilon,$ and the resulting LR statistic does not depend on $\upsilon$. The LR statistic $dQ_1/dQ_0=q_1(s|x)/q_0(s|x)$ can then be calculated by solving, for each $x$, 
\begin{align}
	\min_{p_{0}(\cdot|x),p_{1}(\cdot|x)}&~\sum_{s\in S}\ln\Big(\frac{p_0(s|x)+p_1(s|x)}{p_0(s|x)}\Big)(p_0(s|x)+p_1(s|x))\\
	s.t.&~ \nu_{\theta_0}(B|\theta,x)\le \sum_{s\in B}p_j(s|x),~B\subset S, x\in X\notag\\
	&~ \nu_{\theta_1}(B|\theta,x)\le \sum_{s\in B}p_j(s|x),~B\subset S, x\in X.\notag
\end{align}

\section{Concluding Remarks}\label{sec:conclusion}
This paper develops robust likelihood-based tests for incomplete economic models. Our main result shows that, although unrestricted selection may induce arbitrary heterogeneity and dependence in the observed outcomes, the least favorable laws in repeated experiments are product measures whenever the latent variables are independent across experiments. This result yields finite-sample minimax likelihood-ratio tests and reduces their implementation to computing least-favorable likelihoods in individual experiments. We further provide conditions under which directed one-sided alternatives admit exact conditional uniformly most powerful tests in power guarantee. Finally, we show that sharp identifying restrictions make these procedures computationally tractable by characterizing the least favorable pair through a finite-dimensional convex program, and we extend the framework to accommodate nuisance parameters and observable covariates.

\bibliographystyle{ecta}

\clearpage
\appendix

\section{Figures}
\begin{figure}[H]
    \centering
    \caption{Empirical size under clustered selection, $n=1000$}
\includegraphics[width=0.75\textwidth]{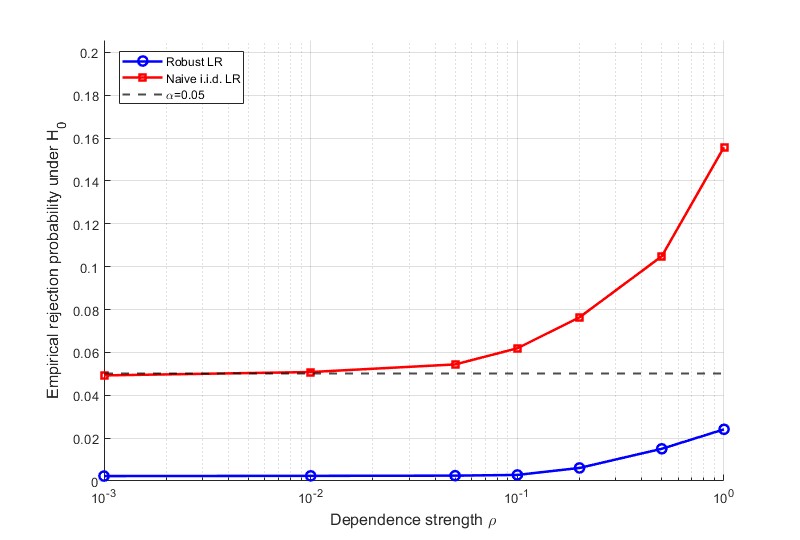}
    \label{fig:size_dependence}

    \caption{Power of the robust LR test under the fixed alternative}
    \includegraphics[width=0.75\textwidth]{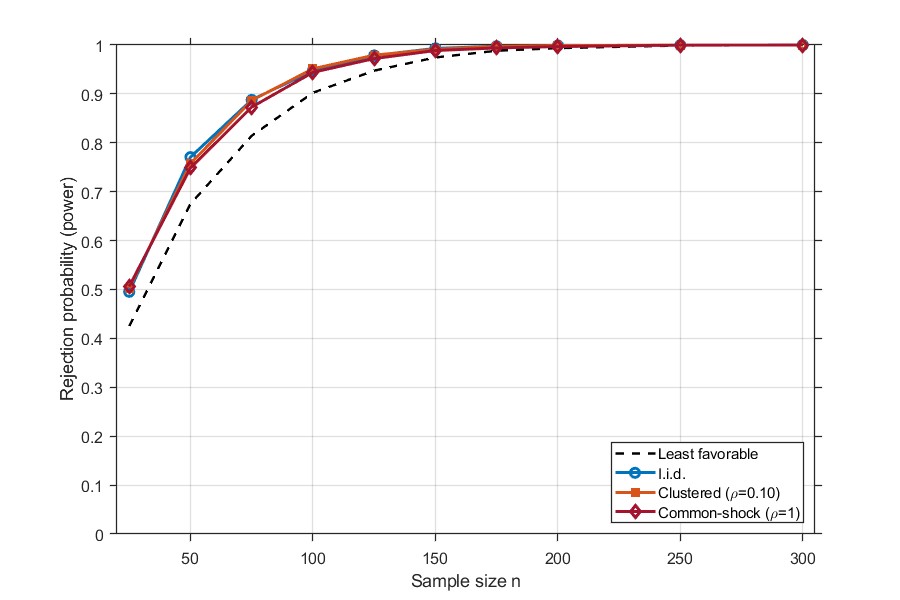}
    \label{fig:power}
\end{figure}

\begin{figure}
    \centering
     \caption{Local power of the exact conditional UMP test in the entry game}
     \label{fig: localpowerandumppower}
    \includegraphics[width=0.75\linewidth]{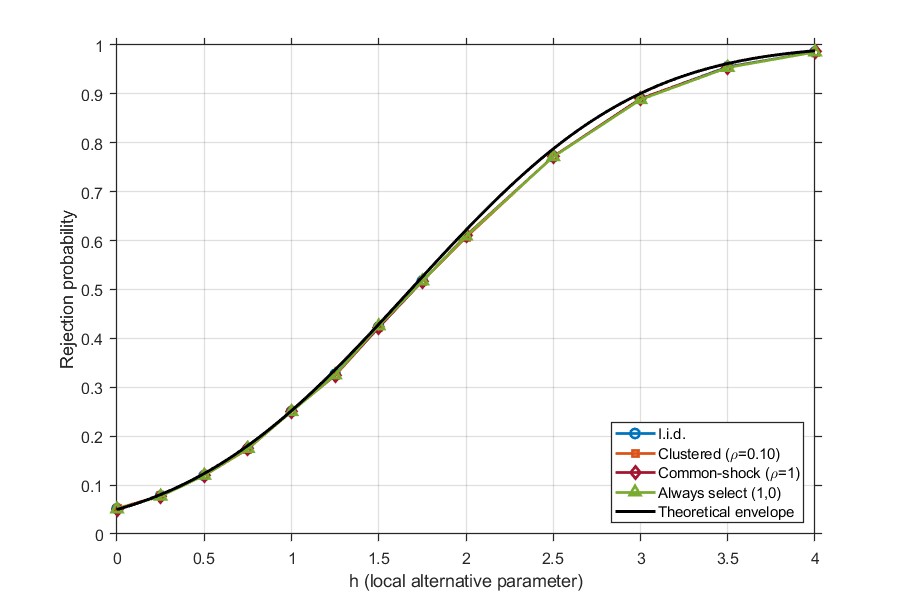}
\end{figure}

\begin{figure}
    \centering
        \caption{Rejection probabilities of $\phi^R$ under $\theta_w$}
    \includegraphics[width=0.75\linewidth]{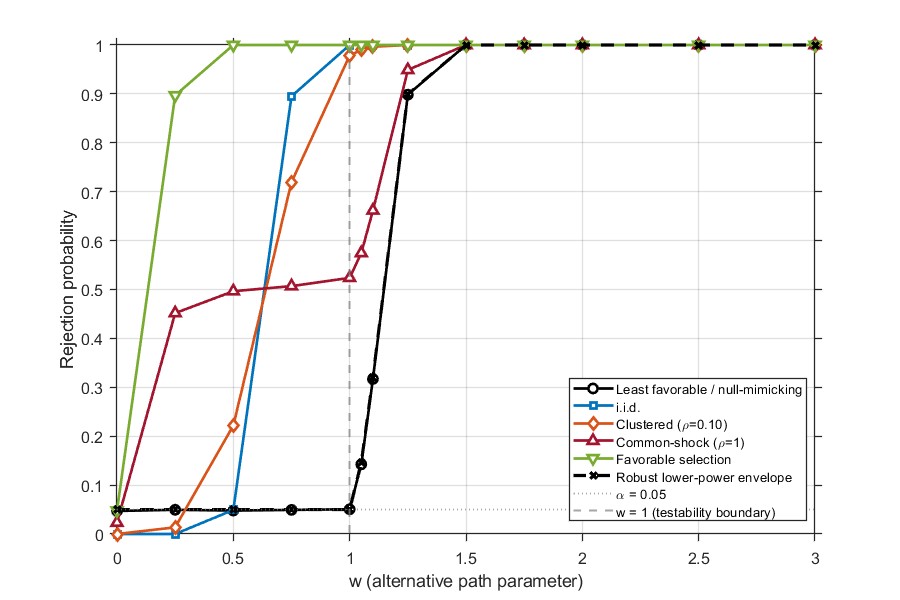}
    \label{fig:roy_power}
\end{figure}

\begin{figure}
    \centering
        \caption{Rejection probabilities of $\phi^R$ under the separated alternative $w^\ast=1.50$}
    \includegraphics[width=0.75\linewidth]{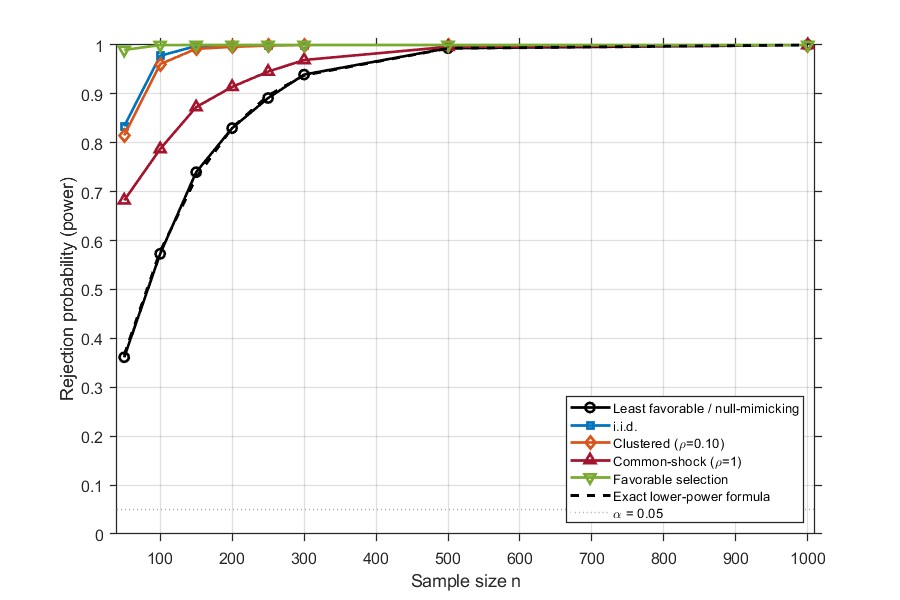}
    \label{fig:roy_power_samplesize}
\end{figure}

\clearpage
\section{Background material}
\subsection{Capacities}\label{sec:capacities} Throughout the appendix, let $\Omega$ be a compact metric space and let $\Sigma_\Omega$ denote its Borel $\sigma$-algebra. Let $\mathcal K(\Omega)$ be the set of compact subsets of $\Omega$ endowed with the Hausdorff metric. Let $\mathcal C(\Omega)$ be the set of continuous functions on $\Omega$. Let $\Delta(\Omega)$ be the set of  Borel probability measures on $\Omega$ endowed with the weak topology.

A set function $\nu^{*}$ is said to be a \emph{capacity} if $\nu^{*}$ satisfies the following conditions:
\begin{enumerate}[label=(\roman*)]\label{de:capac}
	\item{$\nu^{*}(\emptyset)=0, \nu^{*}(\Omega)=1$,\label{cd:capac1}}
	\item{$A\subset B \Rightarrow \nu^{*}(A) \leq \nu^{*}(B)$,~ for all $A,B\in \Sigma_\Omega$.\label{cd:capac2}}
	\item{$A_n \uparrow A \Rightarrow \nu^{*}(A_n) \uparrow \nu^{*}(A)$, for all $\{A_n,n\ge 1\}\subset \Sigma_\Omega$ and $A\in\Sigma_\Omega$.\label{cd:capac3}}
	\item{$F_n \downarrow F, F_n$ closed $\Rightarrow \nu^{*}(F_n) \downarrow \nu^{*}(F)$.\label{cd:capac4}}
\end{enumerate}
One may define integral operations with respect to capacities as follows. Let $f:\Omega\to\mathbb R$ be a measurable function. The \emph{Choquet integral} of $f$ with respect to $\nu$ is defined by
\begin{align}
\int fd\nu\equiv \int_{-\infty}^0(\nu(\{\omega:f(\omega)\ge t\})-\nu(\Omega))dt+\int_0^\infty \nu(\{\omega:f(\omega)\ge t\})dt,
\end{align}
where the integrals on the right hand side are Riemann integrals.

The following result due to Choquet  follows immediately from Theorems 1-3 in \cite{philippe1999decision}. We use this result repeatedly.
\begin{lemma}\label{lem:choquet}
	Let $\Omega$ be a Polish space. Let $M$ be a probability measure on $\mathcal K(\Omega)$. Let $\mathcal P=\{P\in\Delta(\Omega):P= \int P_KdM(K),P_K\in \Delta(K)\}$. Then, $\nu(\cdot)=\inf_{P\in\mathcal P}P(\cdot)$ is a belief function and satisfies
	\begin{align}
		\nu(A)=M(\{K\subset A\}).
	\end{align}
\end{lemma}
In each experiment characterized by the tuple $(S,U,\Theta,G;m)$, one may apply the lemma above with $\mathcal P=\mathcal P_{\theta}$, $\nu=\nu_\theta$, $K=G(u|\theta)$, and $M$ is the law of $G(u|\theta)$ induced by $G$ from $m_\theta$.

\section{Proofs}

\subsection{Proof of Theorem \ref{thm:cs}, Corollary \ref{cor:cs} and Auxiliary Lemmas}
We use Theorem 8.1.1 in \cite{Lehmann:2006aa} to show Theorem \ref{thm:cs}. For ease of reference, we copy their theorem below (with a slight change of notation to avoid conflicts).
For this, let $\mathcal P=\{P_\eta\in \Delta(S):\eta\in \mathcal E\cup \mathcal E'\}$ be families of probability distributions on $S$ with densities $p_\eta=dP_\eta/d\upsilon$, where $\eta$ belongs to the union of measurable spaces $\mathcal E,\mathcal E'$. Throughout, we assume that the map $(s,\eta)\mapsto p_\eta(s)$ is jointly measurable.

\begin{theorem}[Theorem 8.1.1. of \cite{Lehmann:2006aa}]\label{thm:Lehmann-Romano}
	For any distributions $\mu,\mu'$ over $\Sigma_{\mathcal E}$ and $\Sigma_{\mathcal E'}$, let $\phi_{\mu,\mu'}$ be the most powerful test for testing
\begin{align}
	f(s)=\int_{\mathcal E} p_\eta(s)d\mu(\eta)
\end{align}
at level $\alpha$ against
\begin{align}
	f'(s)=\int_{\mathcal E'} p_\eta(s)d\mu(\eta)
\end{align}
and let $\beta_{\mu,\mu'}$ be its power against the alternative $f'.$ If there exist $\mu$ and $\mu'$ such that
\begin{align}
	\sup_{\eta\in \mathcal E}E_{P_\eta}[\phi_{\mu,\mu'}(s)]&\le \alpha\label{eq:LRthm1}\\
	\inf_{\eta\in \mathcal E'}E_{P_\eta'}[\phi_{\mu,\mu'}(s)]&=\beta_{\mu,\mu'},\label{eq:LRthm2}
\end{align}
then:
\begin{itemize}
	\item[(i)] $\phi_{\mu,\mu'}$ maximizes $\inf_{\eta\in \mathcal E'}E_{P_\eta'}[\phi_{\mu,\mu'}(s)]$ among all level-$\alpha$ tests of the hypothesis $H:\eta\in E$ and is the unique test with this property if it is the unique most powerful level-$\alpha$ test for testing $f$ against $f'$.
	\item[(ii)] The pair of distributions $\mu,\mu'$ is least favorable in the sense that for any other pair $\tilde \mu,\tilde\mu'$ we have
	\begin{align}
		\beta_{\mu,\mu'}\le\beta_{\tilde\mu,\tilde\mu'}.
	\end{align} 
\end{itemize}
\end{theorem}

\begin{lemma}\label{lemma: prodmeasure}
Let $\nu_{\theta}$ be defined as in (\ref{eq:defnutheta}), and let $\nu_{\theta}^{*}$ be its conjugate. Let $f: S\rightarrow \mathbb{R}$ be a measurable function. Similarly, for each $i\in \mathbb N$, let  $f_i: S\rightarrow \mathbb{R}$ be a measurable function.

Then,
	\begin{itemize}
		\item [(i)] There exists a minimizing measure $Q \in \Delta(S)$ and a maximizing measure $Q^*\in\Delta(S)$ such that, for any $t \in \mathbb{R}$,
		\begin{equation}
			{\nu}_{\theta}(f(s) > t)=Q(f(s) > t),
		\label{eq:pm1} \end{equation}
		and
		\begin{equation}
			{\nu}_{\theta}^{*}(f(s) > t)=Q^{*}(f(s) > t).
		\label{eq:pm2} \end{equation}

		\item [(ii)] If $Q_i,Q_i^*$ are the minimizing and maximizing measures for $\{f_i(s)>t\}$ respectively, it follows that
		\begin{equation}
			\nu_\theta^n\Big(\sum_{i=1}^nf_i(s_i) > t\Big)=Q^n\Big(\sum_{i=1}^nf_i(s_i) > t\Big),
		\label{eq:pm3} \end{equation}
		and
		\begin{equation}
			\nu_\theta^{*,n}\Big(\sum_{i=1}^nf_i(s_i) > t\Big)=Q^{*n}\Big(\sum_{i=1}^nf_i(s_i) > t\Big),
		\label{eq:pm4} \end{equation}
		for all $t\in\mathbb R$, where $Q^{n}=\bigotimes_{i=1}^nQ_i$ and $Q^{* n}=\bigotimes_{i=1}^nQ_i^* \in \Delta(S^{n})$.
	\end{itemize}
\end{lemma}
\begin{proof}
	(i) By the same argument as in the proof of Lemma \ref{thm:neyman_pearson}, $\nu_{\theta}^{*}$ is a $2$-alternating capacity. Since $S$ is discrete, any function on $S$ is continuous. Therefore, $f$ is upper semi-continuous. By Lemma 2.4 in HS, there exists a probability measure $p^*\in \Delta(S)$ such that for all $t \in \mathbb{R}$, ${\nu}_{\theta}^{*}(f(s) > t)={p}^*(f(s) > t)$. This ensures \eqref{eq:pm2}. Similarly, let $g=-f$ and note that $g$ is again upper semicontinuous by the continuity of $f$. Applying Lemma 2.4 in HS to the event $\{g\ge -t\}$, there exists $p\in\Delta(s)$ such that
	\begin{align}
		&\nu^*_\theta(g\ge -t)=p(g\ge t)\notag\\
		&\qquad\Leftrightarrow 1-\nu^*_\theta(g< -t)=1-p(g<t)\notag\\
		&\qquad\Leftrightarrow \nu_\theta(f>t)=p(f>t),~\forall t\in\mathbb R.
	\label{eq:p-pm1} \end{align}
	This therefore establishes \eqref{eq:pm1}.

	(ii) For each $i$, let $Y_i\equiv \min_{s_i \in G(u_i|\theta)} f_i(s_i)$ and $Z_i\equiv f_i(s_i)$. The map $Y_i:U\to\mathbb R$ is measurable. Indeed, for every
$t\in\mathbb R$,
\[
    \{u_i:Y_i(u_i)>t\}
    =
    \left\{
    u_i:G(u_i|\theta)
    \subseteq
    \{s_i:f_i(s_i)>t\}
    \right\},
\]
which is measurable because $G(\cdot|\theta)$ is weakly measurable
and $S$ is finite.

We also use $m_{\theta,i}$ for its induced law. For each  experiment, we have
	\begin{equation}
		G(u|\theta) \subseteq \{s\in S: f_i(s)>t \} \Leftrightarrow \min_{s \in G(u|\theta)} f_i(s) > t.
	\label{eq:p-pm2} \end{equation}
	Therefore, by Lemma \ref{lem:choquet},
	\begin{equation}
		\begin{split}
			\nu_{\theta,i}(f_i(s_i) > t) & = m_{\theta,i}(\min_{s_i \in G(u_i|\theta)} f_i(s_i) > t)= m_{\theta,i}(Y_i> t),~ \forall t\in\mathbb R.
		\label{eq:p-pm3} \end{split}
	\end{equation}
	By (i), there is $Q_i\in\Delta (S)$ such that
	\begin{equation}
		\nu_{\theta,i}(f_i(s_i) > t)=Q_i(Z_i > t),~\forall t\in\mathbb R.
	\label{eq:p-pm4} \end{equation}
	Hence, by \eqref{eq:p-pm3}-\eqref{eq:p-pm4}, $Y_i \buildrel d \over = Z_i$ for all $i$.

	Let $\mathcal{P}_{\theta}^{n}$ be defined as in (\ref{eq:P_thetainfty}) and let  $\nu^{n}_{\theta}$,$\nu_{\theta}^{* n}$ be the lower and upper probabilities of $\mathcal{P}_{\theta}^{n}$ respectively. By Lemma \ref{lem:choquet}, $\nu_{\theta}^{n}$ is a belief function and $\nu_{\theta}^{*n}$ is its conjugate. Therefore,
	\begin{equation}
		\nu_{\theta}^{n}\Big(\sum_{i=1}^nf_i({s}_i) > t \Big) =m_{\theta}^{n}\Big(u^{n} \in U^{n}:G^{n}(u^{n}|\theta) \subseteq \{\sum_{i=1}^nf_i({s}_i) > t\} \} \Big).
	\label{eq:p-pm5} \end{equation}
	Since $G^{n} (u^{n}|\theta)=\prod_{i=1}^{n}G({u}_i|\theta)$, inside the parenthesis we have:
	\begin{align}
		\prod_{i=1}^{n}G({u}_i|\theta) &\subseteq \{{s}^{n}: \sum_{i=1}^nf_i(s_i) > t\}\notag\\
		&\Leftrightarrow \min_{s^{n} \in \prod_{i=1}^{n}G(u_i|\theta)} \sum_{i=1}^nf_i(s_i) > t \Leftrightarrow \sum_{i=1}^{n} \min_{s_i \in G(u_i|\theta)} f_i(s_i) > t.
	\label{eq:p-pm6} \end{align}
	By \eqref{eq:p-pm5}-\eqref{eq:p-pm6} and recalling that $Y_i=\min_{s_i\in G({u}_i|\theta)} f_i(s_i)$, we have
	\begin{align}
		\nu_{\theta}^{n}\Big(\sum_{i=1}^nf_i({s}_i) > t \Big)=m^{n}_{\theta}\Big(\sum_{i=1}^{n} \min_{s_i\in G({u}_i|\theta)} f_i(s_i) > t\Big)=m_{\theta}^{n} \Big(\sum_{i=1}^{n} Y_i >t \Big).
	\label{eq:p-pm7} \end{align}
	Let $\{Y_1,Y_2, \dotsi, Y_i, \dotsi \}$ be independently distributed according to $m_{\theta}^n$, and let $\{Z_1,Z_2, \dotsi, Z_i \dotsi \}$ be independently distributed according to $Q^n$. Then, $\sum_{i=1}^{n} Y_i \buildrel d \over = \sum_{i=1}^{n}Z_i$ because $(Y_1,\dots,Y_n) \buildrel d \over = (Z_1,\dots,Z_n)$. Therefore, for all $t\in\mathbb R$,
	\begin{align}
		m^{n}_{\theta}\Big(\sum_{i=1}^{n} Y_i > t\Big)=Q^{n}(\sum_{i=1}^{n} Z_i > t).
	\label{eq:p-pm8} \end{align}
	By \eqref{eq:p-pm7}-\eqref{eq:p-pm8} and $\nu^n_\theta$ being the lower probability of $\mathcal P_\theta^n$, we have
	\begin{align}
		\min_{P \in \mathcal{P}_{\theta}^{n}}P(\sum_{i=1}^nf_i(s_i) > t)=\nu_{\theta}^{n} \Big(\sum_{i=1}^{n}f_i(s_i) > t\Big)=Q^{n}(\sum_{i=1}^nf_i(s_i) > t)\label{eq:p-pm9}.
	\end{align}
	This establishes \eqref{eq:pm3}. One may show \eqref{eq:pm4} by a similar argument.
\end{proof}

\begin{proof}
	[\rm Proof of Theorem \ref{thm:cs}] 
For $j\in\{0,1\}$ and $i=1,\ldots,n$, let
$\eta_{j,i}\in\mathfrak K_{m_{\theta_j,i}}(G_{\theta_j})$
be a Markov kernel inducing $Q_{j,i}$. Define the product kernel by
\[
    \eta_j^n(\{s^n\}\mid u^n)
    =
    \prod_{i=1}^n
    \eta_{j,i}(\{s_i\}\mid u_i),
    \qquad s^n\in S^n.
\]
Since $S$ is finite, $\eta_j^n$ is measurable. Moreover,
\[
    \eta_j^n\{G^n(u^n\mid\theta_j)\mid u^n\}=1
    \qquad m_{\theta_j}^n\text{-a.s.}
\]
By Assumption~1 and Fubini's theorem, the distribution induced by
$\eta_j^n$ is $\bigotimes_{i=1}^n Q_{j,i}.$
Hence $Q_j^n=\bigotimes_{i=1}^nQ_{j,i}$ belongs to
$\cP_{\theta_j}^n$.

Recall that $Q^n_0=\otimes_{i=1}^n Q_{0,i}$, $Q^n_1=\otimes_{i=1}^n Q_{1,i}$, and $\Lambda_n$ is a version of the Radon-Nikodym derivative of them.
We follow Section 8.3 in \cite{Lehmann:2006aa}  and show the following statements:
\begin{itemize}
	\item[(a)] When $s^n$ is distributed according to a distribution in $\mathcal P^n_{\theta_0}$, the probability of the event $\{s^n:\Lambda_n>t\}$ is largest (for any $t$), i.e. $\Lambda_n$ is stochastically largest, when the distribution of $s^n$ is $Q^n_0=\otimes_{i=1}^n Q_{0,i}$.
	\item[(b)] When  $s^n$ is distributed according to a distribution in $\mathcal P^n_{\theta_1}$, the probability of the event $\{s^n:\Lambda_n>t\}$ is smallest (for any $t$), i.e. $\Lambda_n$ is stochastically smallest, when the distribution of $s^n$ is $Q^n_1=\otimes_{i=1}^n Q_{1,i}$.
	\item[(c)] $\Lambda_n$ is stochastically larger when the distribution of $s$ is $Q_1^n$ than when it is $Q_0^n$.
\end{itemize}
These statements are summarized by 
\begin{align}
	Q^{n,\prime}_0(\Lambda_n>t)\stackrel{(a)}{\le} Q_0^n(\Lambda_n>t)\stackrel{(c)}{\le} Q_1^n(\Lambda_n>t)\stackrel{(b)}{\le} Q^{n,\prime}_1(\Lambda_n>t),\label{eq:LR_ineq}
\end{align}
for all $t$, $Q^{n,\prime}_0\in\mathcal P^n_{\theta_0}$, and $Q^{n,\prime}_1\in\mathcal P^n_{\theta_1}.$

Below, we invoke Lemma \ref{lemma: prodmeasure}. For this, let $f_i(\cdot)=\ln \Lambda_i(\cdot)$, where $\Lambda_i\in dQ_{1,i}/dQ_{0,i}$.
Let  $(Q^{*n},Q^n)$ be the product measures  in Lemma \ref{lemma: prodmeasure} with $f_i=\ln \Lambda_i$ for $i=1,\dots,n$.

Note that $\Lambda_n>t$ is equivalent to $\sum_{i=1}^nf_i(s_i)>\ln t$.
By Lemma \ref{lemma: prodmeasure} with $t'=\ln t$, it then follows that
\begin{align}
		\nu^{*n}_{\theta_0}(\Lambda_n>t)=\nu_{\theta_0}^{*n}\big(\sum_{i=1}^n f_i(s_i)>t'\big)=Q^{*n}\big(\sum_{i=1}^n f_i(s_i)>t'\big)=Q^{*n}(\Lambda_n>t),
\label{eq:prod2} \end{align}
where $Q^{*n}=Q_0^n$. Recall that $\nu^{*n}_{\theta_0}(\Lambda_n>t)=\sup_{Q_0^{n,\prime}\in \mathcal P_{\theta_0}^n}Q_0^{n,\prime}(\Lambda_n>t)$. This therefore means $Q^{n}_0$ makes $\Lambda_n$ stochastically largest among all distributions in $\mathcal P^n_{\theta_0}$ and hence ensures inequality (a) in \eqref{eq:LR_ineq}.

	Similarly, again by Lemma \ref{lemma: prodmeasure},
	\begin{align}
		\nu^n_{\theta}(\Lambda_n>t)=\nu_{\theta}^n\big(\sum_{i=1}^n f_i(s_i)>t'\big)=Q^n\big(\sum_{i=1}^n f_i(s_i)>t'\big)=Q^n(\Lambda_n>t),\label{eq:prod1}.
	\end{align}
where 	$Q^n=Q^{n}_1$.
	 Therefore, $Q^n$ makes $\Lambda_n$ stochastically smallest and hence ensures inequality (b) in \eqref{eq:LR_ineq}.

The middle inequality in \eqref{eq:LR_ineq} follows from Corollary 3.2.1 in \cite{Lehmann:2006aa} and the Neyman-Pearson lemma.
Let $\mathcal E=\mathcal P^n_{\theta_0}$ and $\mathcal E'=\mathcal P^n_{\theta_1}$. Let $\mu,\mu'\in \Delta(\mathcal E)$ be distributions, each assigning probability 1 to a single distribution, $\mu$ to $Q_0^n\in\mathcal P^n_{\theta_0}$ and $\mu'$ to $Q_1^n\in\mathcal P^n_{\theta_1}$. Let $(C_n,\gamma_n)$ be chosen so that $E_{Q_0^n}[\phi_n(s^n)]=\alpha$, where $\phi_n$ is the likelihood-ratio test defined as in \eqref{eq:pi-n}.
The argument above shows that $\mu,\mu'$ satisfy \eqref{eq:LRthm1}-\eqref{eq:LRthm2}.
Equip $\mathcal E$ and $\mathcal E'$ with the Borel
$\sigma$-fields inherited from the finite-dimensional simplex
$\Delta(S^n)$. With counting measure as the dominating measure, the density
of $P\in\Delta(S^n)$ is $p_P(s^n)=P(\{s^n\})$
Because $S^n$ is finite, the map $(s^n,P)\mapsto p_P(s^n)$ is jointly
measurable. Thus the measurability condition in Theorem \ref{thm:Lehmann-Romano}  is satisfied.
The first claim of the theorem then follows from applying Theorem \ref{thm:Lehmann-Romano} to the present setting. 

Because $\Lambda_n$ has finite support, choose $C_n^-<C_n$ such that
$\Lambda_n$ takes no value in $[C_n^-,C_n)$. Then
\[
    \phi_n
    =
    (1-\gamma_n)1\{\Lambda_n>C_n\}
    +
    \gamma_n 1\{\Lambda_n>C_n^-\}.
\]
Hence the stochastic ordering applied
at $C_n$ and $C_n^-$, implies that $Q_0^n$ maximizes the rejection
probability under the null and $Q_1^n$ minimizes it under the alternative.
\end{proof}

\begin{proof}[\rm Proof of Corollary \ref{cor:cs}]
By Theorem \ref{thm:cs}, the LFP $(Q_0^n,Q_1^n)$ exists, and they are product measures. Note that, in the application of Lemma \ref{lemma: prodmeasure}, $Q_i$ (and $Q^*_i$) is identical across $i$ because $\nu_\theta$ (and $\nu^*_\theta$) is identical across $i$. The conclusion of the Corollary then follows by arguing as in the proof of Theorem \ref{thm:cs}.
\end{proof}

\subsection{Proof of Proposition  \ref{prop:asymptotics}}
\begin{proof}
	[\rm Proof of Proposition \ref{prop:asymptotics}] Note that
	\begin{align}
		\sup_{P\in\mathcal P^n_{\theta_0}}E_{P}[\phi_n^*(s^n)]&= \sup_{P\in\mathcal P^n_{\theta_0}}P(\Lambda_n(s^n)> C_n^*)\notag\\
		&=\nu^{*n}_{\theta_0}(\Lambda_n(s^n)> C_n^*)\notag\\
		&=Q_0^n\big(\Lambda_n(s^n)> C_n^*\big)\notag\\
		&=Q_0^n\Big(\frac{1}{\sqrt n}\sum_{i=1}^n\ln\frac{dQ_1}{dQ_0}(s_i)-E_{Q_0}[\ln \frac{dQ_1}{dQ_0}(s_i)]> \sigma_{Q_0}z_{\alpha}\Big),\label{eq:clt}
	\end{align}
	where the second equality follows from Choquet's theorem (Lemma \ref{lem:choquet}), and the third equality follows from $Q_0^n$ being the least favorable null distribution by Theorem \ref{thm:cs}. Let $Z_i\equiv \ln \frac{dQ_1}{dQ_0}(s_i)$. Under $Q_0^n$, $Z_i,i=1,\dots,n$ is an i.i.d. sequence with a finite variance due to $\sigma^2_{Q_0}<\infty$. Hence, if $\sigma_{Q_0}>0$, by the CLT for i.i.d. random variables, one obtains
	\begin{align}
		\lim_{n\to\infty} Q_0^n\Big(\frac{1}{\sqrt n}\sum_{i=1}^n\frac{\ln\frac{dQ_1}{dQ_0}(s_i)-E_{Q_0}[\ln \frac{dQ_1}{dQ_0}(s_i)]}{\sigma_{Q_0}}>z_{\alpha}\Big)=\text{Pr}\big(Z> z_{\alpha}\big)= \alpha,
	\end{align}
	where $Z\sim N(0,1)$. If $\sigma_{Q_0}=0$, the summand in \eqref{eq:clt} is identically 0 and hence the probability of the event is zero and hence $\limsup_{n\to\infty}\sup_{P\in\mathcal P^n_{\theta_0}}E_{P}[\phi_n^*(s^n)]\le \alpha.$
\end{proof}

\subsection{Proof of Proposition \ref{prop:cond-ump}}
\begin{proof}[\rm Proof of Proposition \ref{prop:cond-ump}]
Fix $k\in\cK_n$.  The conditional test induced by $\phi_n^*$ is
\[
    \phi_{n,k}^*(s^n)
    =
    1\{T_n(s^n)>c_k\}+\gamma_k 1\{T_n(s^n)=c_k\},
    \qquad s^n\in\Omega_{n,k}.
\]
This function is nondecreasing in $T_n$.  By Assumption \ref{ass:lfp},
\[
    \sup_{P\in\cP^n_{\theta_0}:P(K_n=k)>0}
    E_P[\phi_n^*(s^n)\mid K_n=k]
    =
    \sup_{R\in\cC^n_{\theta_0,k}}E_R[\phi_{n,k}^*(s^n)]
    =
    E_{Q_{\theta_0,k}}[\phi_{n,k}^*(s^n)].
\]
By \eqref{eq:cv}, the last display equals $\alpha$.  This proves exact conditional robust level.

Next fix $w\in\mathcal W_1$.  Since $\phi_{n,k}^*$ is nondecreasing in $T_n$, Assumption \ref{ass:lfp} gives
\[
    \underline\pi^n_{\theta_w,k}(\phi_n^*)
    =
    \inf_{R\in\cC^n_{\theta_w,k}}E_R[\phi_{n,k}^*(s^n)]
    =
    E_{Q_{\theta_w,k}}[\phi_{n,k}^*(s^n)],
\]
which proves the stated lower-power formula.

It remains to prove conditional UMP optimality.  Let $\psi_n$ be any test with conditional robust level $\alpha$, and let $\psi_{n,k}$ denote its restriction to $\Omega_{n,k}$.  Because $Q_{\theta_0,k}\in\cC^n_{\theta_0,k}$,
\[
    E_{Q_{\theta_0,k}}[\psi_{n,k}(s^n)]
    \le
    \sup_{R\in\cC^n_{\theta_0,k}}E_R[\psi_{n,k}(s^n)]
    \le \alpha .
\]
Thus $\psi_{n,k}$ is a level-$\alpha$ test for the ordinary simple null $Q_{\theta_0,k}$.  By Assumption \ref{ass:mlr}, the likelihood ratio $dQ_{\theta_w,k}/dQ_{\theta_0,k}$ is a nondecreasing function of $T_n$.  Hence, by the Neyman-Pearson lemma, with the usual monotone-likelihood-ratio nesting argument, the upper-tail test $\phi_{n,k}^*$ maximizes $E_{Q_{\theta_w,k}}[\cdot]$ among all tests with $E_{Q_{\theta_0,k}}[\cdot]\le\alpha$, for every $w\in\mathcal W_1$.  Therefore,
\[
    E_{Q_{\theta_w,k}}[\psi_{n,k}(s^n)]
    \le
    E_{Q_{\theta_w,k}}[\phi_{n,k}^*(s^n)].
\]
Since $Q_{\theta_w,k}\in\cC^n_{\theta_w,k}$,
\[
    \underline\pi^n_{\theta_w,k}(\psi_n)
    =
    \inf_{R\in\cC^n_{\theta_w,k}}E_R[\psi_{n,k}(s^n)]
    \le
    E_{Q_{\theta_w,k}}[\psi_{n,k}(s^n)].
\]
Combining the last two inequalities with the lower-power equality already established for $\phi_n^*$ yields
\[
    \underline\pi^n_{\theta_w,k}(\psi_n)
    \le
    E_{Q_{\theta_w,k}}[\psi_{n,k}(s^n)]
    \le
    E_{Q_{\theta_w,k}}[\phi_{n,k}^*(s^n)]
    =
    \underline\pi^n_{\theta_w,k}(\phi_n^*).
\]
This proves the conditional UMP optimality.
\end{proof}

\subsection{Corollaries to Proposition \ref{prop:cond-ump}}\label{ssec:UMPs}
\textbf{Entry game:}
Consider Example \ref{ex:game}. Consider testing $H_0:\theta=(0,0)$ against $H_1:\theta=\theta_w,w\in\mathcal W_1$, where $\theta_w=(-w,-w),w\in\cW_1$.

Let $U_i=(U_i^{(1)},U_i^{(2)})$, $i=1,\ldots,n$, be i.i.d. with a known continuous exchangeable law $F$ on $\R^2$ (e.g., $N(0,I_2)$).  Assume that $F$ has no mass on the relevant threshold sets and define $\kappa=
1-F\big((-\infty,0)\times(-\infty,0)\big),$ $d(w)=F\big([w,\infty)\times[w,\infty)\big),$ $p(w)=\frac{d(w)}{\kappa}$, and $p_0=p(0)$. Assume $0<\kappa<1$ and $0<p(w)<p_0<1$ for every $w\in\cW_1$.  For instance, one may take $\cW_1=(0,\infty)$ when this strict inequality holds for every $w>0$.

For each $k\in\cK_n$, choose $c_k\in\{0,1,\ldots,k\}$ and $\gamma_k\in[0,1]$ such that
\[
    P_{B\sim \mathrm{Bin}(k,p_0)}(B<c_k)+\gamma_k P_{B\sim \mathrm{Bin}(k,p_0)}(B=c_k)=\alpha.
\]
Define
\begin{align}
    \phi_n^{\mathrm{E}}(s^n)
    =
    1\{N_{11}(s^n)<c_{K_n(s^n)}\}
    +
    \gamma_{K_n(s^n)}1\{N_{11}(s^n)=c_{K_n(s^n)}\}.\label{eq:phiE}
\end{align}

\begin{corollary}[Exact conditional UMP test for the directed entry-game alternative]
\label{cor:entry-cump}
In the entry-game experiment described above, the following statements hold.

\begin{enumerate}
\item[(i)] The directed alternatives are robustly testable: for every $w\in\cW_1$, $ \cP^n_{\theta_0}\cap \cP^n_{\theta_w}=\emptyset.$

\item[(ii)] The test $\phi_n^{\mathrm{E}}$ is conditionally exact.  For every $k\in\cK_n$,
\[
    \sup_{R\in\cC^n_{\theta_0,k}} E_R[\phi_n^{\mathrm{E}}(s^n)]
    =
    \alpha.
\]
Consequently, since $K_n\sim\mathrm{Bin}(n,\kappa)$ under the null,
\[
    \sup_{P\in\cP^n_{\theta_0}} E_P[\phi_n^{\mathrm{E}}(s^n)]
    =
    \alpha.
\]

\item[(iii)] The test is conditionally UMP in power guarantee along the directed path. 
For every $w\in\cW_1$, every $k\in\cK_n$, and every test $\psi_n$ satisfying the conditional robust level condition
\[
    \sup_{R\in\cC^n_{\theta_0,k}}E_R[\psi_n(s^n)]\le \alpha,
\]
we have
\[
    \inf_{R\in\cC^n_{\theta_w,k}}E_R[\phi_n^{\mathrm{E}}(s^n)]
    \ge
    \inf_{R\in\cC^n_{\theta_w,k}}E_R[\psi_n(s^n)].
\]
\end{enumerate}
\end{corollary}

\begin{proof}
We verify the high-level conditions of Proposition \ref{prop:cond-ump} with the conditional ordering statistic $T_n=-N_{11}$.  The conclusion of parts (ii)--(iii) then follows directly from that proposition.

First, $K_i=1\{s_i\neq(0,0)\}$ and $D_i=1\{s_i=(1,1)\}$ are both selection-invariant, and
\[
    P_{\theta_w}(K_i=1)=\kappa,
    \qquad
    P_{\theta_w}(D_i=1)=d(w),
    \qquad
    P_{\theta_w}(D_i=1\mid K_i=1)=p(w).
\]
Since $U_i$ is i.i.d. across markets,
\[
    N_{11}\mid K_n=k \sim \mathrm{Bin}(k,p(w))
\]
under every law in $\cP^n_{\theta_w}$.  In particular, under $H_0$, $N_{11}\mid K_n=k\sim\mathrm{Bin}(k,p_0)$.

For any upper-tail critical function in $T_n=-N_{11}$,
\[
    \tau_{c,\gamma}\{T_n(s^n)\}
    =
    1\{-N_{11}>c\}+\gamma 1\{-N_{11}=c\},
\]
the expectation depends only on the conditional distribution of $N_{11}$ given $K_n=k$, which is selection-invariant.  Therefore, for each $k$ and $w$, we may take any feasible conditional laws $Q_{\theta_0,k}\in\cC^n_{\theta_0,k}$ and $Q_{\theta_w,k}\in\cC^n_{\theta_w,k}$ that induce the binomial laws $\mathrm{Bin}(k,p_0)$ and $\mathrm{Bin}(k,p(w))$ for $N_{11}$.  Then the supremum over $\cC^n_{\theta_0,k}$ and the infimum over $\cC^n_{\theta_w,k}$ are attained at these laws for all such critical functions. This verifies Assumption \ref{ass:lfp}.

It remains only to check the monotone likelihood-ratio condition.  Choose $Q_{\theta_0,k}$ and $Q_{\theta_w,k}$ so that, conditional on $K_n=k$, the probability assigned to the monopoly outcomes is split symmetrically between $(1,0)$ and $(0,1)$.  Exchangeability of $F$ and the symmetry of the directed path make these conditional laws feasible.  With respect to counting measure on $\Omega_{n,k}$, we have
\[
    \frac{dQ_{\theta_w,k}}{dQ_{\theta_0,k}}(s^n)
    =
    C_{w,k}
    \left(
    \frac{p(w)(1-p_0)}{p_0(1-p(w))}
    \right)^{N_{11}(s^n)}
\]
for a positive constant $C_{w,k}$.  Since $p(w)<p_0$, this likelihood ratio is decreasing in $N_{11}$, and hence nondecreasing in $T_n=-N_{11}$, ensuring Assumption \ref{ass:mlr}.  Proposition \ref{prop:cond-ump} therefore implies that $\phi_n^{\mathrm{E}}$ is conditionally exact and conditionally UMP.  

Finally, part (i) follows from the same selection-invariance.  The marginal probability of duopoly is $d(0)$ under every law in $\cP^n_{\theta_0}$ and is $d(w)$ under every law in $\cP^n_{\theta_w}$.  Since $d(w)<d(0)$ for $w\in\cW_1$, the two sets of laws cannot intersect.
\end{proof}

\bigskip
\noindent
\textbf{Roy Model}: Consider Example \ref{ex:roy}. Let
\[
    \theta=(\theta^{(0,0)},\theta^{(0,1)},\theta^{(1,0)})\in\Theta,
\]
where $\theta^{(a,b)}=m((Y_0,Y_1)=(a,b))$, and note that $ \theta^{(1,1)}
    \equiv
    1-\theta^{(0,0)}-\theta^{(0,1)}-\theta^{(1,0)}.$
Fix $\theta_0=(\theta_0^{(0,0)},\theta_0^{(0,1)},\theta_0^{(1,0)})$ satisfying $0<1-\theta_0^{(0,0)}<1$ and $0<
    \frac{1-\theta_0^{(0,0)}-\theta_0^{(0,1)}}
         {1-\theta_0^{(0,0)}}
    <1.$ Under $\theta_0$, $P(s=(1,0))$ belongs to 
\[[\theta_0^{(1,0)},\,
1-\theta_0^{(0,0)}-\theta_0^{(0,1)}].\]
Hence, any parameter value that induces $P(s=(1,0))>1-\theta_0^{(0,0)}-\theta_0^{(0,1)}$ is robustly testable.

Let $\{\theta_w:w\in\cW\}$ be a scalar-indexed directed path.  We assume that the path keeps the probability of type $(0,0)$ fixed,
\[
    \theta_w^{(0,0)}=\theta_0^{(0,0)},
    \qquad w\in\cW,
\]
and choose a directed, robustly testable subset $\cW_1\subseteq\cW$ such that, for every $w\in\cW_1$,
\begin{align}
	    1-\theta_0^{(0,0)}-\theta_0^{(0,1)}
    <
    \theta_w^{(1,0)}
    <
    1-\theta_0^{(0,0)}.\label{eq:roy_path1}
\end{align}
Define
\[
    N_{10}(s^n)=\sum_{i=1}^n 1\{s_i=(1,0)\},
    \qquad
    N_{11}(s^n)=\sum_{i=1}^n 1\{s_i=(1,1)\},
\]
and condition on
\[
    K_n(s^n)=N_{10}(s^n)+N_{11}(s^n).
\]
Thus $K_n$ is the number of observations with $Y=1$.  Let
\[
    \Omega_{n,k}=\{s^n:K_n(s^n)=k\},
    \qquad
    \cK_n=\{0,1,\ldots,n\}.
\]
For $\theta\in\Theta$, define the conditional experiment
\[
    \cC^n_{\theta,k}
    =
    \left\{
    \mathcal L_P(S^n\mid K_n=k):
    P\in\cP^n_\theta,\; P(K_n=k)>0
    \right\}.
\]

Set
\[
    p_0^+
    =
    \frac{1-\theta_0^{(0,0)}-\theta_0^{(0,1)}}
         {1-\theta_0^{(0,0)}},
    \qquad
    p_w^-
    =
    \frac{\theta_w^{(1,0)}}
         {1-\theta_0^{(0,0)}}.
\]
By \eqref{eq:roy_path1}, $p_w^->p_0^+$ for every $w\in\cW_1$.  The conditional ordering statistic is
\[
    T_n(s^n)=N_{10}(s^n),
\]
so the relevant tests are upper-tail tests.  For each $k\in\cK_n$, choose an integer $c_k\in\{0,1,\ldots,k\}$ and $\gamma_k\in[0,1]$ such that
\[
    P_{B\sim \Bin(k,p_0^+)}(B>c_k)
    +
    \gamma_kP_{B\sim \Bin(k,p_0^+)}(B=c_k)
    =\alpha.
\]
For $k=0$, take $c_0=0$ and $\gamma_0=\alpha$.  Define
\[
    \phi_n^{\mathrm R}(s^n)
    =
    1\{N_{10}(s^n)>c_{K_n(s^n)}\}
    +
    \gamma_{K_n(s^n)}1\{N_{10}(s^n)=c_{K_n(s^n)}\}.
\]
\begin{corollary}[Exact conditional UMP test for a general directed Roy-model alternative]
\label{cor:roy-general-cump}
In the Roy-model experiment described above, the following statements hold.

\begin{enumerate}[label=(\roman*)]
\item The directed alternatives are robustly testable: for every $w\in\cW_1$, $    \cP^n_{\theta_0}\cap\cP^n_{\theta_w}=\emptyset.$

\item The test $\phi_n^{\mathrm R}$ is conditionally exact.  For every $k\in\cK_n$,
\[
    \sup_{R\in\cC^n_{\theta_0,k}}
    E_R[\phi_n^{\mathrm R}(s^n)]
    =\alpha.
\]
Consequently, since $K_n\sim\Bin(n,1-\theta_0^{(0,0)})$ under $\theta_0$ for every admissible selection mechanism,
\[
    \sup_{P\in\cP^n_{\theta_0}}E_P[\phi_n^{\mathrm R}(s^n)]=\alpha.
\]

\item The test is conditionally UMP in robust lower power along the directed path.  For every $w\in\cW_1$, every $k\in\cK_n$, and every test $\psi_n$ satisfying
\[
    \sup_{R\in\cC^n_{\theta_0,k}}E_R[\psi_n(s^n)]\le\alpha,
\]
we have
\[
    \inf_{R\in\cC^n_{\theta_w,k}}E_R[\phi_n^{\mathrm R}(s^n)]
    \ge
    \inf_{R\in\cC^n_{\theta_w,k}}E_R[\psi_n(s^n)].
 \]
\end{enumerate}
\end{corollary}

\begin{proof}
We first verify robust testability.  In a single experiment, the sharp restrictions imply
\[
    P_\theta(s=(1,0))
    \in
    \left[
    \theta^{(1,0)},\;\theta^{(1,0)}+\theta^{(1,1)}
    \right]
    =
    \left[
    \theta^{(1,0)},\;1-\theta^{(0,0)}-\theta^{(0,1)}
    \right].
\]
Under $\theta_w$, the lower endpoint is $\theta_w^{(1,0)}$. By \eqref{eq:roy_path1},
\[
    \theta_w^{(1,0)}
    >
    1-\theta_0^{(0,0)}-\theta_0^{(0,1)}.
\]
Thus the single-experiment null and alternative sets are separated by the event $\{(1,0)\}$.  The same separation holds in the repeated experiment by considering the marginal event $\{S_1=(1,0)\}$.  Hence $\cP^n_{\theta_0}\cap\cP^n_{\theta_w}=\emptyset$ for every $w\in\cW_1$.

It remains to verify the high-level conditions of Proposition~\ref{prop:cond-ump}.  Fix $k\in\cK_n$.  Since $\theta_w^{(0,0)}=\theta_0^{(0,0)}$ along the path, the event $\{Y=1\}$ has probability $1-\theta_0^{(0,0)}$ for all $w\in\cW_1$ and is unaffected by selection.  Hence $K_n\sim\textrm{Bin}(n,1-\theta_0^{(0,0)})$ under the null and under every directed alternative.

Conditional on $K_n=k$, the $k$ observations with $Y=1$ have latent types in $\{(0,1),(1,0),(1,1)\}$.  Under $\theta_w$, their conditional probabilities are
\[
    m_{\theta_w}(u=(0,1)\mid Y=1)
    =
    \frac{\theta_w^{(0,1)}}{1-\theta_0^{(0,0)}},
\]
\[
    m_{\theta_w}(u=(1,0)\mid Y=1)
    =
    \frac{\theta_w^{(1,0)}}{1-\theta_0^{(0,0)}},
    \qquad
    m_{\theta_w}(u=(1,1)\mid Y=1)
    =
    \frac{\theta_w^{(1,1)}}{1-\theta_0^{(0,0)}}.
\]
Only type $u=(1,1)$ is ambiguous among observations with $Y=1$, because it can generate either $s=(1,0)$ or $s=(1,1)$.  Among the observations for which $K_n=k$, let $A_k$ be the number of observations with latent type $u=(1,0)$, and let $B_k$ be the number with latent type $u=(1,1)$.  Under any admissible selection mechanism,
\begin{align}
    A_k\le N_{10}\le A_k+B_k.	\label{eq:pointwise_ineq}
\end{align}

For upper-tail tests in $T_n=N_{10}$, the rejection probability under the null is maximized by selecting $s=(1,0)$ whenever $u=(1,1)$.  At $\theta_0$, this gives
\[
    N_{10}\mid K_n=k
    \sim
    \textrm{Bin}\left(k,
    \frac{\theta_0^{(1,0)}+\theta_0^{(1,1)}}{1-\theta_0^{(0,0)}}
    \right)
    =
    \textrm{Bin}(k,p_0^+).
\]
Similarly, for $w\in\cW_1$, the rejection probability under the alternative is minimized by selecting $s=(1,1)$ whenever $u=(1,1)$.  This gives
\[
    N_{10}\mid K_n=k
    \sim
    \textrm{Bin}\left(k,
    \frac{\theta_w^{(1,0)}}{1-\theta_0^{(0,0)}}
    \right)
    =
   \textrm{Bin}(k,p_w^-).
\]
For the observations with $Y=0$, use the same selection rule under the least favorable null and alternative; for example, split $u=(0,0)$ equally between $(0,0)$ and $(0,1)$.  This common component does not affect $T_n$ and cancels from the conditional likelihood ratio.  The inequalities \eqref{eq:pointwise_ineq} then verify Assumption \ref{ass:lfp}.

Conditional on $K_n=k$, the least favorable null has success probability $p_0^+$ and the least favorable alternative has success probability $p_w^-$.  Therefore, after cancelling factors that do not depend on $N_{10}$,
\[
    \frac{dQ_{\theta_w,k}}{dQ_{\theta_0,k}}(s^n)
    =
    C_{w,k}
    \left(
    \frac{p_w^-(1-p_0^+)}
         {p_0^+(1-p_w^-)}
    \right)^{N_{10}(s^n)}
\]
for a positive constant $C_{w,k}$.  Since $p_w^->p_0^+$, the base of the power is greater than one.  Hence the likelihood ratio is nondecreasing in $T_n$, verifying Assumption \ref{ass:mlr}.

The conditional exactness and conditional UMP conclusions now follow from Proposition~\ref{prop:cond-ump}.  
\end{proof}

\begin{remark}[Connection to the calibration in the simulations]\label{rem:calibration}
The numerical calibration used in the simulations is recovered by setting
\[
    \theta_0=\left(\frac16,\frac12,\frac16\right),
    \qquad
    \theta_w=
    \left(\frac16,\frac12-\frac{w}{6},\frac16+\frac{w}{6}\right),
    \qquad
    \cW_1=(1,3].
\]
Then, $\cW_1=[1,3]$ satisfies \eqref{eq:roy_path1}, and
\[
    p_0^+=\frac25,
    \qquad
    p_w^-=\frac{1+w}{5}.
\]

\end{remark}

\subsection{Proof of Theorems in Section \ref{sec:extensions}}
In what follows, we repeatedly use the fact that, for any nonnegative measurable function $g$ on $S$, belief function $\nu$ and its conjugate $\nu^*$, one has
\begin{align}
	\int g(s)d\nu^*(s)&=\int \max_{s\in K}g(s) dM_\nu\\
	\int g(s)d\nu(s)&=\int \min_{s\in K}g(s) dM_\nu,
\end{align}
where $M_\nu$ is the probability measure on $\mathcal K(S)$ associated with $\nu$ (see Lemma \ref{lem:choquet}). 
\begin{proof}
	[\rm Proof of Lemma \ref{thm:bayes-dempster-shafer}] 
We first start with showing \eqref{eq:riskfun} and \eqref{eq:bdsrisk}. For this, observe that
	\begin{align}
		R(\theta,\phi)&=\max_{P\in \mathcal P_\theta}\int \phi(s)I_{\Theta_0}(\theta)+\zeta(1-\phi(s)) I_{\Theta_1}(\theta)dP(s)\notag\\
		&=\max_{P\in \mathcal P_\theta}( I_{\Theta_0}(\theta)-\zeta I_{\Theta_1}(\theta))\int\phi(s)dP(s)+\zeta I_{\Theta_1}(\theta)\notag\\
		&= I_{\Theta_0}(\theta)\max_{P\in \mathcal P_\theta}\int\phi(s)dP(s) -\zeta I_{\Theta_1}(\theta)\min_{P\in \mathcal P_\theta}\int\phi(s)dP(s)+\zeta I_{\Theta_1}(\theta)\notag\\
		&=\int \phi(s)d\nu^*_\theta(s)I_{\Theta_0}(\theta)+\zeta(1-\int\phi(s)d\nu_\theta(s))I_{\Theta_1}(\theta),
	\end{align}
	where the third equality follows from the fact that $ I_{\Theta_0}(\theta)- \zeta I_{\Theta_1}(\theta)>0$ if and only if $I_{\Theta_0}(\theta)=1$ (and $ I_{\Theta_0}(\theta)-\zeta I_{\Theta_1}(\theta)\le 0$ if and only if $I_{\Theta_1}(\theta)=1$). The last equality follows from $core(\nu_\theta)=\mathcal P_\theta$ by Theorem 3 in \cite{philippe1999decision} and the fact that, for any nonnegative bounded function $g$ on $S$, $\int g d\nu\le \int gdP\le \int gd\nu^*$ for all $P\in core(\nu)$.

	Using this, write
	\begin{align}
		r(\mu,\phi)&\equiv \int_\Theta R(\theta,\phi)d\mu(\theta)\notag\\
		&=\tau\int\int \phi(s)I_{\Theta_0}(\theta)d\nu^*_\theta(s)d\mu_0(\theta)+\zeta(1-\tau)(1-\int\int\phi(s)I_{\Theta_1}(\theta)d\nu_\theta(s)d\mu_1(\theta)).
	\label{eq:rmphi} \end{align}
	By Lemma \ref{lem:choquet}, for each $\nu_\theta$, there is a unique Borel probability measure $M_\theta$ on $\mathcal K(S)$ such that
	\begin{align}
		\nu_\theta(A)=M_\theta(K\subset A),~\forall A\subset S.
	\end{align}
	Let $q$ be a $\sigma$-finite measure on $\mathcal K(S)$ such that $M_\theta\ll q$. Let $M_{\kappa_1}$ be a Borel probability measure on $\mathcal K(S)$ such that $dM_{\kappa_1}/dq=\int_{\Theta_1}\frac{dM_\theta}{dq}d\mu_1(\theta)$. For any $A\subset S$, it then follows that
	\begin{align}
		\int_{\Theta_1}\nu_\theta(A)d\mu_1(\theta)&=\int_{\Theta_1}M_\theta(K\subset A)d\mu_1(\theta)\notag\\
		&= \int_{\Theta_1}\int_{\mathcal K(S)} 1\{K\subset A\}dM_\theta d\mu_1(\theta)\notag\\
		&=\int_{\mathcal K(S)} 1\{K\subset A\}\int_{\Theta_1}\frac{dM_\theta}{dq}d\mu_1(\theta)dq(K)\notag\\
		&=\int_{\mathcal K(S)}1\{K\subset A\}dM_{\kappa_1}(K)\notag\\
		&=\int_{S}1\{s\in A\}d\kappa_1(s)\notag\\
		&=\kappa_1(A),
	\end{align}
	where the third equality follows from Fubini's theorem. The existence and uniqueness of $\kappa_1$ follows again from Choquet's theorem (Lemma \ref{lem:choquet}). Note that by $\phi\ge 0$ and the definition of the Choquet integral, one can show by the same argument
	\begin{align}
		\int\int \phi(s)I_{\Theta_1}(\theta)d\nu_\theta(s)d\mu_1(\theta)&=\int_{\Theta_1}\int \inf_{s\in K}\phi(s) dM_{\theta}(K)d\mu_1(\theta)\notag\\
		&=\int \inf_{s\in K}\phi(s)\int_{\Theta_1}\frac{dM_\theta}{dq} d\mu_1(\theta)dq(K)\notag\\
		&=\int \phi(s)d\kappa_1(s),
	\label{eq:kappa1_rep} \end{align}
	where the second equality follows from Fubini's theorem. Similarly, it follows that
	\begin{align}
		\int\int \phi(s)I_{\Theta_0}(\theta)d\nu^*_\theta(s)d\mu_0(\theta)=\int\phi(s)d\kappa_0^*(s).
	\label{eq:kappa0_rep} \end{align}
	By \eqref{eq:rmphi}, \eqref{eq:kappa1_rep}, and \eqref{eq:kappa0_rep}, we have
	\begin{align}
		r(\mu,\phi)=\tau\int\phi(s)d\kappa_0^*(s)+\zeta(1-\tau)(1-\int\phi(s)d\kappa_1(s)).
	\end{align}
	Therefore, \eqref{eq:bdsrisk} holds. Minimizing the BDS risk is then equivalent to minimizing
	\begin{align}
		\tilde r_t(\mu,\phi)=t\int\phi(s)d\kappa_0^*(s)-\int\phi(s)d\kappa_1(s),
	\end{align}
	where $t=\tau/\zeta(1-\tau)>0$. Let $A\equiv\{s:\phi(s)>0\}.$ Minimizing the risk function above with respect to $\phi$ is then equivalent to minimizing the 2-alternating function $w_t(A)\equiv t\kappa_0^*(A)-\kappa_1(A)$ with respect to $A\subset S$. By Lemmas 3.1 and 3.2 in HS, for each $t\in[0,\infty]$, there exists a set $A_t\subset S$ such that
	\begin{align}
		w_t(A_t)=\inf_{A\subset S}w_t(A),
	\end{align}
	and $\{A_t,t\ge 0\}$ forms an increasing family of sets. Now define
	\begin{align}
		\Lambda(s)\equiv\inf\{t|s\in A_t\}.
	\end{align}
	By Theorem 4.1 in HS, the conclusion of the theorem then follows.
\end{proof}

\begin{proof}
	[\rm Proof of Theorem \ref{thm:minimax}] Note that $\mu_1$ is fixed throughout and $\mathcal M=\{\mu:\mu=\tau\mu_0+(1-\tau)\mu_1,\mu_0\in\Delta(\Theta_0),\tau\in[0,1]\}$. In what follows, we therefore redefine $R$ in \eqref{eq:riskfun} as
	\begin{align}
		R(\theta,\phi)&=\int \phi(s)d\nu^*_\theta(s)I_{\Theta_0}(\theta)+\zeta(1-\int\phi(s)d\kappa_1(s))I_{\Theta_1}(\theta)\notag\\
		&=R_0(\theta,\phi)I_{\Theta_0}(\theta)+R_1(\phi)I_{\Theta_1}(\theta),
	\label{eq:newrisk} \end{align}
	where $\kappa_1=\int_{\Theta_1}\nu_\theta d\mu_1(\theta).$

	First, we show $\sup_{\mu\in \mathcal{M}}\inf_{\phi\in\mathbf\Phi}\int R(\theta,\phi)d\mu\le \inf_{\phi\in\mathbf\Phi}\sup_{\theta\in\Theta}R(\theta,\phi)$. This follows because for any $(\theta,\phi)$, one has
	\begin{align*}
		\inf_{\phi'}R(\theta,\phi') \le R(\theta,\phi)\le \sup_{\theta'}R(\theta',\phi),
	\end{align*}
	and hence
	\begin{align*}
		\sup_{\theta}\inf_{\phi'}R(\theta,\phi') \le \inf_{\phi} \sup_{\theta'}R(\theta',\phi).
	\end{align*}
	Note that $\sup_{\theta}\inf_{\phi'}R(\theta,\phi')\ge \inf_{\phi\in\mathbf\Phi}\int R(\theta,\phi)d\mu$ for any $\mu$, and hence, the first claim follows.

	The other direction $\sup_{\mu\in \mathcal{M}}\inf_{\phi\in\mathbf\Phi}\int R(\theta,\phi)d\mu\ge \inf_{\phi\in\mathbf\Phi}\sup_{\theta\in\Theta}R(\theta,\phi)$ follows from Lemma \ref{lem:lecam}. To see this, let
	\begin{align}
		\beta\equiv\sup_{\mu\in \mathcal{M}}\inf_{\phi\in\mathbf\Phi} \int R(\theta,\phi)d\mu.
	\end{align}
	If $\beta=\infty$, the result is trivial. If $\beta<\infty$, set $f(\theta)=\beta$ for all $\theta\in\Theta.$ By construction, $f(\theta)\ge \inf_{\phi\in\mathbf\Phi} \int R(\theta,\phi)d\mu$ for every $\mu.$ By Lemma \ref{lem:lecam}, this is equivalent to the existence of $\phi^\dagger\in\mathbf\Phi$ such that
	\begin{align}
		\beta\ge R(\theta,\phi^\dagger),~\forall\theta\in\Theta.
	\end{align}
	This implies
	\begin{align}
		\beta \ge \sup_{\theta\in\Theta}R(\theta,\phi^\dagger) \ge \inf_{\phi\in\mathbf\Phi}\sup_{\theta\in\Theta}R(\theta,\phi).
	\end{align}
	Finally, observe that by \eqref{eq:newrisk},
	\begin{align}
		\sup_{\theta\in\Theta}R(\theta,\phi)=\sup_{\theta\in\Theta_0}R_0(\theta,\phi)\vee R_1(\phi).
	\end{align}
	This completes the proof.
\end{proof}

\section{Auxiliary Lemmas} 
\subsection{Proof of Lemma \ref{thm:neyman_pearson}}
\begin{proof}[\rm Proof of Lemma \ref{thm:neyman_pearson}] 

It is straightforward to show $\nu^*_\theta$ is a capacity satisfying conditions (i)-(iv) in Appendix \ref{sec:capacities}.

 Since $G(\cdot|\theta)$ is weakly measurable, the map $u\mapsto G(u|\theta)$ defines a measurable map from $U$ to $\mathcal K(S)$. Let $\tilde m$ be the induced measure of $m_\theta$ on $\mathcal K(S)$ by this map. Then, by Lemma \ref{lem:choquet}, $P\in \mathcal P_\theta$ is equivalent to $P\in core(\nu)$ for an infinitely monotone capacity $\nu$ such that $\nu(A)=\tilde m(K\subset A)$ for all $A\in \Sigma_\Omega.$ By Lemma 2.5 in HS,
	\begin{align}
		\nu(A)=\inf_{P\in\mathcal P_\theta}P(A)=\nu_\theta(A),~\text{ for all }A\in \Sigma_S,
	\end{align}
	and hence $\nu_\theta$ is infinitely monotone.

	By the previous step, $\nu_{\theta_1}^{*}$ and $\nu_{\theta_0}^{*}$ are also 2-alternating and 2-monotone capacities respectively. Let $\Lambda$ be the Radon-Nikodym derivative of $\nu_{\theta_1}^{*}$ and $\nu_{\theta_0}^{*}$ in the sense of HS (Section 3). Then, by their Theorem 4.1, the conclusion of the lemma follows. 
\end{proof}

\subsection{A Lemma for results in Section \ref{ssec:computation}}
Consider
	\begin{align}
		(q_0,q_1)=\text{argmin}_{(p_0,p_1)\in \Delta(S)^2}&~\sum_{s\in S}H\Big(\frac{p_0(s)}{p_0(s)+p_1(s)}\Big)(p_0(s)+p_1(s)) \label{eq:opt_ex}\\
		s.t. &~~~\nu_{\theta_0}(A)\le \sum_{s\in A}p_0(s),~A\subset S \notag\\
		& ~~~\nu_{\theta_1}(A)\le \sum_{s\in A}p_1(s),~A\subset S \notag.
	\end{align}
\begin{lemma}\label{lem:convex}
Let $H$ be a twice continuously function on $[0,1]$ such that $H''>0$. Then, (i) $(q_0,q_1)$ solving \eqref{eq:opt_ex} are the densities of LFP $(Q_0,Q_1)$; (ii) If, for each $s\in S$, $\nu_{\theta_j}(s)>0$ for either $j=0$ or $j=1$, the objective function in \eqref{eq:opt_ex} is convex.
\end{lemma}	

\begin{proof}[\rm Proof of Lemma \ref{lem:convex}]
(i) By Choquet's theorem, $P\in\mathcal P_\theta=core(\nu_\theta)$ is equivalent to $\nu_\theta(A)\le P(A)=\sum_{s\in A}p(s)$ for all $A\subset S$.  Then, the claim follows from Theorem 6.1 in HS.

\noindent
(ii) For any function $f:\mathbb R^d\to\bar{\mathbb R}$,  let $\text{dom}(f)\equiv\{x\in\mathbb R^d:f(x)<\infty\}$. Let $(x,t)\in\mathbb R\times\mathbb R_+$. Then, the perspective $g(x,t)= tH(x/t)$ is convex with $\text{dom}(g)=\{(x,t):x/t\in \text{dom}(H),t>0\}$ iff $H$ is convex \citep[][Sec. 3.2.6]{Boyd:2004jk}. 
Let $h(x,y)=(x,x+y)$ and note that this map is linear. Then, since $p_0(s)+p_1(s)>0$ for all $s$ by assumption, the objective function in \eqref{eq:opt_ex} can be expressed as
\begin{align}
\sum_{s\in S}H\Big(\frac{p_0(s)}{p_0(s)+p_1(s)}\Big)(p_0(s)+p_1(s))=	\sum_{s\in S}g\big(h(p_0(s),p_1(s))\big),
\end{align}
which is convex in $(p_0,p_1)$.
\end{proof}

\subsection{Lemmas for Theorems in Section \ref{sec:nuisance}}
Below, we identify each decision function (randomized test) $\phi$ with a Markov kernel $\phi:S\times \mathcal B_{\{0,1\}}\to [0,1]$ and let $\mathbf\Phi$ be the set of all decision functions. We then equip $\mathbf\Phi$ with the weak topology \citep[see][Definition 2.2]{Hausler:2015aa}. The following lemma is an analog of Lemma 46.1 in \cite{Strasser:1985aa} for the BDS risk. We state it as a lemma because the BDS risk $R$ is defined through Choquet integrals with respect to capacities (instead of measures) and hence Lemma 46.1 in \cite{Strasser:1985aa} is not directly applicable.\footnote{They also show their results to generalized decision functions, which we do not pursue here.}
\begin{lemma}\label{lem:lecam}
Suppose $S$ is finite and $\Theta$ is compact.	Let $R$ be defined as in \eqref{eq:riskfun}. For every $f:\Theta\to\mathbb R$, the following assertions are equivalent.

	(i) There exists $\phi\in\mathbf\Phi$ such that $f(\theta)\ge R(\theta,\phi)$ for every $\theta\in\Theta$.

	(ii) $\int fd\mu\ge \inf_{\phi\in \mathbf\Phi}\int R(\theta,\phi)d\mu(\theta)$ for every $\mu\in \Delta(\Theta).$
\end{lemma}
\begin{proof}
	The implication $(i)\Rightarrow (ii)$ is obvious. We therefore prove the other implication. Consider the following sets of functions:
	\begin{align*}
		M_1&=\{f\}\\
		M_2&=\{h\in\mathcal C(\Theta):h(\cdot)=R(\cdot,\phi),\phi\in \mathbf\Phi\}.
	\end{align*}
	We mimic the proof of Lemma 46.1 in \cite{Strasser:1985aa} while replacing $M_2$ with the set above. For this, let $M\subseteq\mathcal C(\Theta)$ be an arbitrary set. For any $m\in \Delta(\Theta)$, the lower envelope of $M$ is defined as
	\begin{align}
		\psi_M(m)\equiv\inf\{\int fdm,f\in M\}.
	\end{align}
	Also define
	\begin{align}
		\alpha(M)\equiv \bigcup_{f\in M}\{g\in\mathcal C(\Theta):f\le g\}.
	\end{align}
	This is the set of continuous functions that dominate some function in $M$. Since $M_2$ is compact by Lemma \ref{lem:compact}, $\alpha(M_2)$ is closed and hence coincides with its closure $\overline{\alpha (M_2)}$  \citep[][Remark 45.4]{Strasser:1985aa}.

	By (ii), $\psi_{M_2}(m)\le \psi_{M_1}(m)$ for all $m\in \Delta(\Theta)$. By Lemma \ref{lem:subconvexity}, $M_2$ is subconvex. Then, by Theorem 45.6 in \cite{Strasser:1985aa}, for every $f\in M_1$, there is $g\in\alpha(M_2)=\overline{\alpha (M_2)}$ such that $g\le f$. By the construction of $\alpha(M_2)$, this means there exists $\phi\in \mathbf\Phi$ such that
	\begin{align}
		R(\cdot,\phi)\le g(\cdot)\le f(\cdot).
	\end{align} This completes the proof.
\end{proof}

Below, a set $M$ is said to be \emph{subconvex}, if for any $\alpha\in (0,1)$ and $h_1,h_2\in M,$ there exists $h_3\in M$ such that $h_3\le \alpha h_1+(1-\alpha)h_2$.
\begin{lemma}\label{lem:subconvexity}
	$M_2$ is subconvex.
\end{lemma}
\begin{proof}
	Let $h_1,h_2\in M_2$. Then, there exist $\phi_1,\phi_2\in\mathbf\Phi$ such that $h_j(\cdot)=R(\cdot,\phi_j),j=1,2$. Therefore, for any $\alpha\in(0,1)$,
	\begin{align}
		\alpha h_1(\theta)+(1-\alpha)h_2(\theta) =&\alpha R(\theta,\phi_1)+(1-\alpha)R(\theta,\phi_2)\notag\\
		=&\zeta\big(\alpha\int \phi_1(s)d\nu_{\theta}^*(s)+(1-\alpha)\int \phi_2(s)d\nu_{\theta}^*(s)\big)I_{\Theta_0}(\theta)\notag\\
		&+\big(\alpha(1-\int\phi_1(s)d\nu_{\theta}(s))+(1-\alpha)(1-\int\phi_2(s)d\nu_{\theta}(s))\big)I_{\Theta_1}(\theta).
	\label{eq:subcvx1} \end{align}
	Note that $\nu_{\theta}^*$ is 2-alternating. Therefore, the Choquet integral with respect to $\nu_{\theta}^*$ is subadditive. The Choquet integral is also positively homogeneous. Therefore,
	\begin{align}
		\alpha\int\phi_1(s)d\nu_{\theta}^*(s)+(1-\alpha)\int \phi_2(s)d\nu_{\theta}^*(s)&=\int \alpha\phi_1(s)d\nu_{\theta}^*(s)+\int (1-\alpha)\phi_2(s)d\nu_{\theta}^*(s)\notag\\
		&\ge \int \alpha\phi_1(s)+(1-\alpha)\phi_2(s)d\nu_{\theta}^*(s).
	\label{eq:subcvx2} \end{align}
	Similarly, by the 2-monotonicity of $\nu_{\theta}$, the Choquet integral with respect to it is superadditive and positively homogeneous. Therefore,
	\begin{align}
		\alpha(1-\int\phi_1(s)d\nu_{\theta}(s))+(1-\alpha)(1-\int\phi_2(s)d\nu_{\theta}(s))\ge 1-\int \alpha\phi_1(s)+(1-\alpha)\phi_2(s)d\nu_{\theta}(s).
	\label{eq:subcvx3} \end{align}
	Combining \eqref{eq:subcvx1}-\eqref{eq:subcvx3}, we obtain
	\begin{align}
		\alpha h_1(\cdot)+(1-\alpha)h_2(\cdot)\ge h_3(\cdot),
	\label{eq:subcvx4} \end{align}
	where $h_3(\cdot)=R(\cdot, \phi_3)$ with $\phi_3=\alpha\phi_1+(1-\alpha)\phi_2$. Hence, $h_3\in M_2$. Conclude that $M_2$ is subconvex.
\end{proof}
\begin{lemma}\label{lem:compact}
	Suppose that $S$ is finite and $\Theta$ is compact. Then, $M_2$ is weakly compact.
\end{lemma}
\begin{proof}
	Equip $M_2$ with the weak topology. Let $\varrho$ be the counting measure.
	We let $\mathbf\Phi=K^1(\varrho)$ denote the set of Markov kernels equipped with the weak topology. It is the coarsest topology that makes any functional of the following form continuous:
	\begin{align}
		L(\phi)= \int_S\int_{\{0,1\}}h(a)\phi(s,da)f(s)d\varrho(s),~f\in L^1_\varrho(S),h\in\mathcal C_b(\{0,1\}).
	\end{align}
	By Theorem 2.7 in \cite{Hausler:2015aa}, $\mathbf\Phi$ is compact if the set $\varrho\mathbf\Phi\equiv\{ \varrho\phi:\varrho\phi(\cdot)=\int_S \phi(s,\cdot)d\varrho(s),\phi\in \mathbf\Phi\}$ is relatively compact in $\Delta(\{0,1\}).$ Note that $\{0,1\}$ is compact. Hence, $\varrho\mathbf\Phi$ is uniformly tight, which implies that $\varrho\mathbf\Phi$ is relatively compact by Prohorov's theorem \citep[][Problem 1.12.1]{vandervaart:wellner:1996}. This ensures the compactness of $\mathbf\Phi$.

	Below, let $\mathcal K(S)$ the set of all nonempty (and necessarily closed) subsets of $S$. By Choquet's theorem, a belief function $\nu_\theta$ can be expressed by its canonical representation $(\mathcal K(S),K,\hat m_\theta)$, where $K$ is a random set following a probability measure $\hat m_\theta$ on $\mathcal K(S)$ such that $\nu_\theta(A)=\hat m_\theta(K\subset A)$ for all $A\in\mathcal K(S)$. Below, we adopt this canonical representation and also denote the measure on $\mathcal K(S)$ by $m_\theta$ rather than $\hat m_\theta.$ We also note that we write $\phi(s)=\int_{\{0,1\}}\phi(s,da)$ in what follows.

	Define $g:\mathbf\Phi\to\mathcal C(\Theta)$ pointwise by $\phi\mapsto R(\cdot,\phi)$. We argue that this map is continuous, where we equip $\mathbf\Phi$ and $\mathcal C(\Theta)$ with weak topologies. Let $\phi_n,n=1,2,\cdots$ be a sequence such that $\phi_n\to\phi\in \mathbf\Phi$ weakly.
\begin{align}
		\int_S\int_{\{0,1\}}\phi_n(s,da)d\nu^*_\theta(s)= \int_S \phi_n(s)d\nu^*_\theta(s)=\int_S \max_{s\in K}\phi_n(s)dm_\theta(K)
	\end{align}
	Since $S$ is finite and $1\{s=s'\}\in L^1_\varrho(S)$ for any $s'\in S$, $\phi_n$ converging weakly implies
	\begin{align}
		\phi_n(s')=\sum_{s\in S}\phi_n(s)1\{s=s'\}\to \sum_{s\in S}\phi(s)1\{s=s'\}=\phi(s'),~\forall s'\in S.
	\end{align}
	Therefore, $\phi_n$ converges pointwise to $\phi$.

	Fix $K\in\mathcal K(S)$. Note that $(s,n)\mapsto \phi_n(s)$ is continuous with respect to the discrete topology. By Berge's maximum theorem, $ \max_{s\in K}\phi_n(s)\to \max_{s\in K}\phi(s)$.
	Hence,
	\begin{align}
		\lim_{n\to\infty}\int_S\int_{\{0,1\}}\phi_n(s,da)d\nu^*_\theta(s)&=\lim_{n\to\infty}\int_S \max_{s\in K}\phi_n(s)dm_\theta(K)\notag\\
		&=\int_S\lim_{n\to\infty}\max_{s\in K}\phi_n(s)dm_\theta(K)\notag\\
		&=\int_S\max_{s\in K}\phi(s)dm_\theta(K)\notag\\
		&=\int_S\phi (s)d\nu^*_\theta(s)\notag\\
		&=\int_S\int_{\{0,1\}}\phi (s,a)d\nu^*_\theta(s).
	\end{align}
	where the second equality follows from the convergence of $\max_{s\in K}\phi_n(s)$ and the dominated convergence theorem.

	Consider any finite Borel measure $\mu$ on $\Theta$. The result above and dominated convergence theorem imply
	\begin{align}
		\lim_{n\to\infty}\int_\Theta \int_S\int_{\{0,1\}}\phi_n(s,da)d\nu^*_\theta(s)I_{\Theta_0}(\theta)d\mu(\theta)&=\int_\Theta\lim_{n\to\infty}\int_S\int_{\{0,1\}}\phi_n(s,da)d\nu^*_\theta(s)I_{\Theta_0}(\theta)d\mu(\theta)\\
		&=\int_\Theta\int_S\int_{\{0,1\}}\phi(s,da)d\nu^*_\theta I_{\Theta_0}(\theta)d\mu(\theta).
	\end{align}
	By a similar argument, one can also show
	\begin{align}
		\lim_{n\to\infty}\int_\Theta \int_S\int_{\{0,1\}}\phi_n(s,da)d\nu_\theta(s)I_{\Theta_1}(\theta)d\mu(\theta)&=\int_\Theta\lim_{n\to\infty}\int_S\int_{\{0,1\}}\phi_n(s,da)d\nu_\theta(s)I_{\Theta_1}(\theta)d\mu(\theta)\\
		&=\int_\Theta\int_S\int_{\{0,1\}}\phi(s,da)d\nu_\theta I_{\Theta_1}(\theta)d\mu(\theta).
	\end{align}
	Note that $\Theta$ is a compact set in a metric space. Corollary 14.15 in \cite{aliprantis2006infinite} then ensures that the dual space of $\mathcal C(\Theta)$ is the set of finite Borel measures on $\Theta$. Combining these results and noting that $R(\theta,\phi)=\zeta\int_S\int_{\{0,1\}}\phi(s,da)d\nu^*_\theta I_{\Theta_0}(\theta)+(1-\int_S\int_{\{0,1\}}\phi(s,da)d\nu_\theta)I_{\Theta_1}(\theta)$, it follows that $R(\cdot,\phi_n)\to R(\cdot,\phi)$ in $\mathcal C(\Theta)$ with respect to the weak topology. This establishes that $g$ is continuous. Hence, $M_2$ is the continuous image of a compact set. Conclude that $M_2$ is compact.
\end{proof}

\section{Details on Examples 1 and 2}\label{sec:implementation}

\subsection{Example 1: Binary response game}
In this section, we provide details on Example \ref{ex:game} including the computation of the belief function, least favorable pair, and minimax test. Recall that $S=\{(0,0),(1,1),(1,0),(0,1)\}$. There exist 14 subsets to be considered (without considering the empty set and $S$). One can then compute the lower and upper bounds of the probability of each event by mimicking the calculation in \eqref{eq:belcalc}. 
The results are summarized in Table \ref{Table:Probbound}.
\begin{table}[ht]
	\begin{center}
		\caption{The lower and upper probability bounds}\label{Table:Probbound}
		\begin{tabular}
			{lll} \hline \hline $A$ & $\nu_{\theta}(A)=\min P(A)$ & $\nu_{\theta}^{*}=\max P(A)$\\
			\hline $A_1= \{ (0,0)\}$ & $ \frac{1}{4}$ & $ \frac{1}{4} $ \\

			$A_2= \{ (1,1) \} $ & $\Phi(\theta^{(1)}) \Phi(\theta^{(2)})$ & $\Phi(\theta^{(1)}) \Phi(\theta^{(2)})$ \\

			$A_3=\{(0,1)\}$ & $\frac{1}{4}-\Phi(\theta^{(1)}) \Phi(\theta^{(2)})+\frac{\Phi(\theta^{(2)})}{2}$ & $\frac{1}{2}(1-\Phi(\theta^{(1)}))$ \\

			$A_4=\{(1,0)\}$ & $\frac{1}{4}-\Phi(\theta^{(1)}) \Phi(\theta^{(2)})+\frac{\Phi(\theta^{(1)})}{2}$ & $\frac{1}{2}(1-\Phi(\theta^{(2)}))$ \\

			$A_5=\{(0,0),(1,1)\}$ & $\frac{1}{4}+\Phi(\theta^{(1)}) \Phi(\theta^{(2)})$ & $\frac{1}{4}+\Phi(\theta^{(1)}) \Phi(\theta^{(2)})$ \\

			$A_6=\{(0,0),(0,1)\}$ & $\frac{1}{2}-\Phi(\theta^{(1)}) \Phi(\theta^{(2)})+\frac{\Phi(\theta^{(2)})}{2}$ & $\frac{3}{4}-\frac{\Phi(\theta^{(1)})}{2}$ \\

			$A_7=\{(0,0),(1,0)\}$ & $\frac{1}{2}-\Phi(\theta^{(1)}) \Phi(\theta^{(2)})+\frac{\Phi(\theta^{(1)})}{2}$ & $\frac{3}{4}-\frac{\Phi(\theta^{(2)})}{2}$\\

			$A_8=A_6^{c}=\{(1,1),(1,0)\}$ & $\frac{1}{4}+\frac{\Phi(\theta^{(1)})}{2}$ & $\frac{1}{2}(1-\Phi(\theta^{(2)}))+ \Phi(\theta^{(1)}) \Phi(\theta^{(2)})$ \\

			$A_9=A_7^{c}=\{(1,1),(0,1)\}$ & $\frac{1}{4}+\frac{\Phi(\theta^{(2)})}{2}$ & $\frac{1}{2}(1-\Phi(\theta^{(1)}))+ \Phi(\theta^{(1)}) \Phi(\theta^{(2)})$\\

			$A_{10}= A_5^{c}=\{(1,0),(0,1)\}$ & $\frac{3}{4}-\Phi(\theta^{(1)}) \Phi(\theta^{(2)})$ & $\frac{3}{4}-\Phi(\theta^{(1)}) \Phi(\theta^{(2)})$\\

			$A_{11}=A_1^{c}=\{(1,1),(1,0),(0,1)\}$ & $\frac{3}{4}$ & $\frac{3}{4}$ \\

			$A_{12}=A_2^{c}=\{(0,0),(1,0),(0,1)\}$ & $1-\Phi(\theta^{(1)}) \Phi(\theta^{(2)})$ & $1-\Phi(\theta^{(1)}) \Phi(\theta^{(2)})$ \\

			$A_{13}=A_3^{c}=\{(0,0),(1,1),(1,0)\}$ & $\frac{1}{2}+\frac{\Phi(\theta^{(1)})}{2}$ & $\frac{3}{4}-\frac{\Phi(\theta^{(2)})}{2}+ \Phi(\theta^{(1)}) \Phi(\theta^{(2)})$ \\

			$A_{14}=A_4^{c}=\{(0,0),(1,1),(0,1)\}$ & $\frac{1}{2}+\frac{\Phi(\theta^{(2)})}{2}$ & $\frac{3}{4}-\frac{\Phi(\theta^{(1)})}{2}+\Phi(\theta^{(1)}) \Phi(\theta^{(2)})$ \\
			\ $S=\{(0,0),(1,1),(1,0),(0,1)\}$ & 1 & 1\\
			\hline\hline
		\end{tabular}
	\end{center}
\end{table}

\subsubsection{LFP}
The following proposition characterizes the LFP and minimax optimal test for testing $H_0:\theta=(0,0)$ against $H_1:\theta=\theta_1$, where $\theta_1<0$. 
\begin{proposition}\label{prop:entry_game}
	Let $(S,U,\Theta,G)$ be as defined in Example \ref{ex:game}. Suppose that $u$ follows the bivariate standard normal distribution. Let $\alpha\in (0,1/4)$, $\theta_0=(0,0)'$, and $\theta_1<0$. 
Then, across the three subcases below, the density of $Q_0$ is  $(q_0(0,0),q_0(1,1),q_0(1,0),q_0(0,1))=(\frac{1}{4},\frac{1}{4}, \frac{1}{4}, \frac{1}{4})$.

	\noindent Case I: If $\Phi(\theta_1^{(1)})(1-\Phi(\theta_1^{(2)})) \le \frac{1}{4}$ and $\Phi(\theta_1^{(2)})(1-\Phi(\theta_1^{(1)})) \le \frac{1}{4}$, the density of $Q_1$ is
	\begin{multline}
(q_1(0,0),q_1(1,1),q_1(1,0),q_1(0,1))\\
=\Big(\frac{1}{4},\Phi(\theta_1^{(1)})\Phi(\theta_1^{(2)}), \frac{3-4\Phi(\theta_1^{(1)})\Phi(\theta_1^{(2)})}{8},\frac{3-4\Phi(\theta_1^{(1)})\Phi(\theta_1^{(2)})}{8}\Big),
\end{multline}
and the minimax test is
\begin{align}
	\phi(s)=
		\begin{cases}
			2\alpha & s=(1,0) \text{ or }(0,1)\\
			0& \text{otherwise};
		\end{cases}
		\label{eq:game_minmax1} \end{align}

	\smallskip
	\noindent Case II: If $\Phi(\theta_1^{(1)})(1-\Phi(\theta_1^{(2)})) > \frac{1}{4}$, the density of $Q_1$ is
\begin{multline}
	(q_1(0,0),q_1(1,1),q_1(1,0),q_1(0,1))\\
	=\Big(\frac{1}{4},\Phi(\theta_1^{(1)})\Phi(\theta_1^{(2)}),\frac{1}{4}-\Phi(\theta_1^{(1)}) (\Phi(\theta_1^{(2)})-\frac{1}{2}),\frac{1}{2}(1-\Phi(\theta_1^{(1)}))\Big),
\end{multline}
and the minimax test is	
\begin{align}	
\phi(s)=
		\begin{cases}
			4\alpha & s=(1,0)\\
			0& \text{otherwise};
		\end{cases}
		\label{eq:game_minmax2} \end{align}

	\smallskip
	\noindent Case III: If $\Phi(\theta_1^{(2)})(1-\Phi(\theta_1^{(1)})) > \frac{1}{4}$, the density of $Q_1$ is
\begin{multline}
	(q_1(0,0),q_1(1,1),q_1(1,0),q_1(0,1))\\
	=\Big(\frac{1}{4},\Phi(\theta_1^{(1)})\Phi(\theta_1^{(2)}), \frac{1}{2}(1-\Phi(\theta_1^{(2)})),\frac{1}{4}-\Phi(\theta_1^{(2)}) (\Phi(\theta_1^{(1)})-\frac{1}{2})\Big),
\end{multline}
and the minimax test is	
\begin{align}
\phi(s)=
		\begin{cases}
			4\alpha & s=(0,1)\\
			0& \text{otherwise}.
		\end{cases}
		\label{eq:game_minmax3} \end{align}
\end{proposition}

\begin{proof}
	[\rm Proof of Proposition \ref{prop:entry_game}] First, observe that the upper and lower probabilities of $A_1,\cdots, A_4$ fully characterize the constraints in the convex program. To see this, observe that, for example,
	\begin{multline}
		\nu_\theta(A_5)=m_\theta(G(u|\theta)\subset \{(0,0),(1,1)\})\\
		=m_\theta(G(u|\theta)= \{(0,0)\})+m_\theta(G(u|\theta)= \{(1,1)\})=\nu_\theta(A_1)+\nu_\theta(A_2)=\frac{1}{4}+\Phi(\theta^{(1)}) \Phi(\theta^{(2)}),
	\end{multline}
	where we note that the additivity of $\nu_\theta$ for this event is due to the form of the correspondence in \eqref{eq:defG} and does not hold in general. One can compute $\nu_\theta(A_6)$ and $\nu_\theta(A_7)$ similarly. The upper bounds $\nu^*_\theta(A_j)$ for $j=8,\cdots,14$ can then be computed using the conjugacy of $\nu_\theta$ and $\nu^*_\theta.$ Similarly, the upper bounds $\nu^*_\theta(A_j)$ for $j=1,\cdots,7$ imply the lower bounds $\nu_\theta(A_j)$ for $j=8,\cdots,14$.

	In sum, it suffices to impose the constraints that arise from the upper and lower probabilities of $A_1,\cdots,A_4$. Further, $q_0(0,0)=q_1(0,0)=1/4$ regardless of the parameter value. This allows to simplify the convex program as
	\begin{align}
		\min_{(q_0,q_1)}-\ln&(\frac{q_0(1,1)}{q_0(1,1)+q_1(1,1)})(q_0(1,1)+q_1(1,1)) -\ln(\frac{q_0(1,0)}{q_0(1,0)+q_1(1,0)})(q_0(1,0)+q_1(1,0))\notag\\
		&\qquad-\ln(\frac{q_0(0,1)}{q_0(0,1)+q_1(0,1)})(q_0(0,1)+q_1(0,1))\notag\\
		\text{s.t.}~ & \frac{1}{4}-\Phi(\theta_1^{ (1)}) \Phi(\theta_1^{(2)})+\frac{\Phi(\theta_1^{(2)})}{2} \leq q_1(0,1) \leq \frac{1}{2}(1-\Phi(\theta_1^{(1)}))\label{eq:ex1ineq1}\\
		& \frac{1}{4}-\Phi(\theta_1^{(1)}) \Phi(\theta_1^{(2)})+\frac{\Phi(\theta_1^{(1)})}{2} \leq q_1(1,0) \leq \frac{1}{2}(1-\Phi(\theta_1^{(2)}))\label{eq:ex1ineq2}\\
		& q_1(1,1)=\Phi(\theta_1^{(1)}) \Phi(\theta_1^{(2)})\label{eq:ex1const1}\\
		& q_1(1,1)+q_1(1,0)+q_1(0,1)=\frac{3}{4}\label{eq:ex1const2}\\
		& q_0(1,1)=q_0(1,0)=q_0(0,1)=\frac{1}{4}.
	\label{eq:ex1const3} \end{align}
	Note that \eqref{eq:ex1const1}-\eqref{eq:ex1const3} imply that the values of $q_0(1,1),q_0(1,0),q_0(0,1)$, and $q_1(1,1)$ are determined uniquely. Hence, it remains to optimize the problem with respect to $q_1(1,0)$ and $q_1(0,1)$. For this, let $y=q_1(1,0)$. Then, one may write $q_1(0,1)=3/4-\Phi(\theta_1^{(1)}) \Phi(\theta_1^{(2)})-y$ due to \eqref{eq:ex1const2} and \eqref{eq:ex1const3}. Hence, the problem reduces to an optimization problem with a single control variable. Using this, define the Lagrangian by
	\begin{multline}
		\mathcal L(y,\lambda)\equiv-\ln \bigg(\frac{1/4}{1/4+y}\bigg)(\frac{1}{4}+y)-\ln \bigg(\frac{1/4}{\frac{1}{4}+\frac{3}{4}-\Phi(\theta_1^{(1)})\Phi(\theta_1^{(2)})-y}\bigg)(\frac{1}{4}+\frac{3}{4}-\Phi(\theta_1^{(1)}) \Phi(\theta_1^{(2)})-y)\\- \lambda_1 \Big(\frac{1}{2}-\frac{\Phi(\theta_1^{(2)})}{2}-y\Big)- \lambda_2 \Big(y-\frac{1}{4}-\frac{\Phi(\theta_1^{(1)})}{2}+\Phi(\theta_1^{(1)}) \Phi(\theta_1^{(2)})\Big).
	\end{multline}
	By Theorem 28.3 in \cite{Rockafellar:1970aa}, the saddle point of the Lagrangian characterizes the optimal solution of the original problem. The Karush-Kuhn-Tucker (KKT) conditions are as follows:
	\begin{align}
		&1-\ln\bigg(\frac{1/4}{1/4+y} \bigg) -1+\ln \bigg(\frac{1/4}{1-\Phi(\theta_1^{(1)})\Phi(\theta_1^{(2)})-y} \bigg)+\lambda_1-\lambda_2=0 \label{eq:kkt1}\\
		&\lambda_1 \Big(\frac{1}{2}-\frac{\Phi(\theta_1^{(2)})}{2}-y\Big)=0\label{eq:kkt2}\\
		&\lambda_2 \Big(y-\frac{1}{4}-\frac{\Phi(\theta_1^{(1)})}{2}+\Phi(\theta_1^{(1)}) \Phi(\theta_1^{(2)})\Big)=0\label{eq:kkt3}\\
		&\lambda_1,\lambda_2\ge 0.
	\label{eq:kkt4} \end{align}
	Below, we consider three cases depending on the values of the Lagrange multipliers.

	\noindent Case 1 ($\lambda_1=\lambda_2=0$): Suppose that $\lambda_1=0$ and $\lambda_2=0$.  Then, the solution from \eqref{eq:kkt1} is
	\begin{equation}
		y=\frac{3}{8}-\frac{\Phi(\theta_1^{(1)}) \Phi(\theta_1^{(2)})}{2}.\label{eq:y_q1}
	\end{equation}
	Substituting this into the complementary slackness conditions \eqref{eq:kkt2} and \eqref{eq:kkt3} yields
	\begin{equation}
		\frac{1}{2}-\frac{\Phi(\theta_1^{(2)})}{2}-\frac{3}{8}+\frac{\Phi(\theta_1^{(1)}) \Phi(\theta_1^{(2)})}{2} \ge0,
	\end{equation}
	and
	\begin{equation}
		\frac{3}{8}-\frac{\Phi(\theta_1^{(1)}) \Phi(\theta_1^{(2)})}{2}-\frac{1}{4}-\frac{\Phi(\theta_1^{(1)})}{2}+\Phi(\theta_1^{(1)})\Phi(\theta_1^{(2)})\ge 0,
	\end{equation}
which can be simplified as
	\begin{equation}\label{eq:unicondition11}
		\Phi(\theta_1^{(2)})(1-\Phi(\theta_1^{(1)})) \le \frac{1}{4},\qquad \Phi(\theta_1^{(1)})(1-\Phi(\theta_1^{(2)})) \le  \frac{1}{4}.
	\end{equation}
By $y=q_1(1,0)$, $q_1(0,1)=3/4-\Phi(\theta_1^{(1)}) \Phi(\theta_1^{(2)})-y$, and \eqref{eq:y_q1}, the least favorable pair is
	\begin{align}
		(q_0(0,0),q_0(1,1),q_0(1,0),q_0(0,1))&=\Big(\frac{1}{4},\frac{1}{4}, \frac{1}{4}, \frac{1}{4}\Big)\label{eq:sol1}\\
		(q_1(0,0),q_1(1,1),q_1(1,0),q_1(0,1))&=\Big(\frac{1}{4},\Phi(\theta_1^{(1)})\Phi(\theta_1^{(2)}),\label{eq:sol2} \frac{3-4\Phi(\theta_1^{(1)})\Phi(\theta_1^{(2)})}{8},\frac{3-4\Phi(\theta_1^{(1)})\Phi(\theta_1^{(2)})}{8}\Big).
	\end{align}
	Hence, the likelihood-ratio statistic is given by
	\begin{align}
		\Lambda(s)&=
		\begin{cases}
			1& s=(0,0)\\
			4\Phi(\theta_1^{(1)})\Phi(\theta_1^{(2)})& s=(1,1)\\
			\frac{3-4\Phi(\theta_1^{(1)})\Phi(\theta_1^{(2)})}{2}&s=(1,0)\\
			\frac{3-4\Phi(\theta_1^{(1)})\Phi(\theta_1^{(2)})}{2}&s=(0,1).
		\end{cases}
	\end{align}
	Hence, under $Q_0$, $\pi(s)$ is a discrete random variable supported on $\{4\Phi(\theta_1^{(1)})\Phi(\theta_1^{(2)}),1,\frac{3-4\Phi(\theta_1^{(1)})\Phi(\theta_1^{(2)})}{2}\}$ with probability mass function $(1/4,1/4,1/2)$. For $\theta^{(j)}<0,j=1,2$, one has
	\begin{align}
		4\Phi(\theta_1^{(1)})\Phi(\theta_1^{(2)})<1<\frac{3-4\Phi(\theta_1^{(1)})\Phi(\theta_1^{(2)})}{2}.
	\end{align}
	The largest value of the support is therefore $\frac{3-4\Phi(\theta_1^{(1)})\Phi(\theta_1^{(2)})}{2}$. Setting $C$ to a higher than this value yields $Q_0(\pi(s)>C)=0$, which violates the level-$\alpha$ requirement. Therefore, set $C=\frac{3-4\Phi(\theta_1^{(1)})\Phi(\theta_1^{(2)})}{2}$. Solving
	\begin{align}
		\alpha=E_{Q_0}[\phi(s)]=\gamma Q_0(\pi(s)\ge C)=\gamma Q_0(s=(1,0) \cup s=(0,1))=\frac{\gamma}{2},
	\end{align}
	one obtains $\gamma=2\alpha.$ This gives the minimax test in \eqref{eq:game_minmax1}.

	\noindent Case 2 ($\lambda_1=0$ and $\lambda_2>0$): Suppose that $\lambda_1=0$ and $\lambda_2>0$. By the complementary slackness condition \eqref{eq:kkt3}, the solution is obtained at the lower bound $y=\frac{1}{4}-\Phi(\theta_1^{(1)})\Phi(\theta_1^{(2)})+\frac{\Phi(\theta_1^{(1)})}{2}$. The KKT conditions then reduce to
	\begin{equation}
		-\ln\bigg(\frac{\frac{1}{4}}{\frac{1}{4}+y} \bigg) +\ln \bigg(\frac{\frac{1}{4}}{1-\Phi(\theta_1^{(1)})\Phi(\theta_1^{(2)})-y} \bigg)-\lambda_2=0.
	\end{equation}
	Substituting $y=\frac{1}{4}-\Phi(\theta_1^{(1)}) \Phi(\theta_1^{(2)})+\frac{\Phi(\theta_1^{(1)})}{2}$ into this yields
	\begin{equation}
		\lambda_2=\ln \bigg(\frac{\frac{1}{2}-\Phi(\theta_1^{(1)})\Phi(\theta_1^{(2)})+\frac{\Phi(\theta_1^{(1)})}{2}}{\frac{3}{4}-\frac{\Phi(\theta_1^{(1)})}{2}} \bigg).
	\label{eq:ratio} \end{equation}
	The difference between the numerator and denominator in the logarithm above is
	\begin{equation}
		\frac{1}{2}-\Phi(\theta_1^{(1)}) \Phi(\theta_1^{(2)})+\frac{\Phi(\theta_1^{(1)})}{2}-(\frac{3}{4}-\frac{\Phi(\theta_1^{(1)})}{2})=\Phi(\theta_1^{(1)})(1-\Phi(\theta_1^{(2)})-\frac{1}{4}.
	\end{equation}
	Therefore, $\lambda_2 > 0$ if and only if
	\begin{equation}\label{eq:unicondition21}
		\Phi(\theta_1^{(1)})(1-\Phi(\theta_1^{(2)}) > \frac{1}{4}.
	\end{equation}
	Similarly, the complementary slackness condition \eqref{eq:kkt2} is satisfied with $\lambda_1=0$ and $y=\frac{1}{4}-\Phi(\theta_1^{(1)}) \Phi(\theta_1^{(2)})+\frac{\Phi(\theta_1^{(1)})}{2}$.
	Note that the constraint $\frac{1}{2}-\frac{\Phi(\theta_1^{(2)})}{2}-y\ge 0$ is trivially satisfied because
	\begin{align}
	\frac{1}{2}-\frac{\Phi(\theta_1^{(2)})}{2}-y&=	\dfrac{1}{2}-\dfrac{\Phi(\theta_1^{(2)})}{2}-\bigg(\frac{1}{4}-\Phi(\theta_1^{(1)})\Phi(\theta_1^{(2)})+\frac{\Phi(\theta_1^{(1)})}{2} \bigg) \\
	&=(\frac{1}{2}-\Phi(\theta_1^{(1)}))(\frac{1}{2}-\Phi(\theta_1^{(2)}))\ge0,\label{eq:unicondition22}
	\end{align}
where the last inequality follows from $\theta_1^{(j)}\le 0$ for $j=1,2.$
	Hence, if  $\Phi(\theta_1^{(1)})(1-\Phi(\theta_1^{(2)})) > \dfrac{1}{4}$, the least favorable pair is
	\begin{align}
		(q_0(0,0),q_0(1,1),q_0(1,0),q_0(0,1))&=\Big(\frac{1}{4},\frac{1}{4}, \frac{1}{4}, \frac{1}{4}\Big)\label{eq:sol1}\\
		(q_1(0,0),q_1(1,1),q_1(1,0),q_1(0,1))&=\Big(\frac{1}{4},\Phi(\theta_1^{(1)})\Phi(\theta_1^{(2)}),\label{eq:sol2} \frac{1}{4}-\Phi(\theta_1^{(1)}) (\Phi(\theta_1^{(2)})-\frac{1}{2}),\frac{1}{2}(1-\Phi(\theta_1^{(1)}))\Big),
	\end{align}
	where we used $y=q_1(1,0)$, $q_1(0,1)=3/4-\Phi(\theta_1^{(1)}) \Phi(\theta_1^{(2)})-y$, and $y=\frac{1}{4}-\Phi(\theta_1^{(1)}) \Phi(\theta_1^{(2)})+\frac{\Phi(\theta_1^{(1)})}{2}$.
	Hence, the likelihood ratio statistic is given by
	\begin{align}
		\Lambda(s)&=
		\begin{cases}
			1& s=(0,0)\\
			4\Phi(\theta_1^{(1)})\Phi(\theta_1^{(2)})& s=(1,1)\\
			1-4\Phi(\theta_1^{(1)}) (\Phi(\theta_1^{(2)})-\frac{1}{2})&s=(1,0)\\
			2(1-\Phi(\theta_1^{(1)}))&s=(0,1).
		\end{cases}
	\end{align}
	By \eqref{eq:unicondition21} and $\theta^{(j)}_1<0,j=1,2$, it is straightforward to show
	\begin{align}
		4\Phi(\theta_1^{(1)})\Phi(\theta_1^{(2)})<1<2(1-\Phi(\theta_1^{(1)}))\le 1-4\Phi(\theta_1^{(1)}) (\Phi(\theta_1^{(2)})-\frac{1}{2}).
	\end{align}
	As in Case 1, we set $C=1-4\Phi(\theta_1^{(1)}) (\Phi(\theta_1^{(2)})-\frac{1}{2})$ and note that $Q_0(\pi(s)\ge C)=Q_0(s=(1,0))=1/4>\alpha$ by our assumption. This implies that we should set $\gamma=4\alpha.$ This gives the minimax test in \eqref{eq:game_minmax2}.

	\noindent Case 3 ($\lambda_1>0$ and $\lambda_2=0$): The argument is similar to the one for Case 2. Hence, we omit the proof.
\end{proof}

\subsubsection{Exact Conditional UMP Tests in the Entry Game}\label{ssec:ump_game_results}
The symmetric path used in the main text (in Section \ref{ssec:ump_game}) is chosen for exposition. The same
conditional-UMP logic can be applied to other directed paths, provided that
the path is restricted to a subset on which the LFP regime does not change
and the conditional likelihood ratios are monotone in a common statistic.

For example, if a path $w\mapsto\theta_w$ remains in Case I of
Proposition~\ref{prop:entry_game}, the relevant ordering statistic is the monopoly-versus-duopoly
comparison used in the main text. If the path remains in Case II, the
conditional ordering statistic is the number of $(1,0)$ outcomes, after
conditioning on the selection-invariant counts of $(0,0)$, $(1,1)$, and total
monopoly outcomes. If the path remains in Case III, the analogous statistic is
the number of $(0,1)$ outcomes. Thus, for any such fixed-regime path, the
high-level conditions of Proposition~\ref{prop:cond-ump} can be verified
with the corresponding conditioning and ordering statistics. We demonstrate this below.

For $\theta=(\theta^{(1)},\theta^{(2)})$ with $\theta^{(1)}<0$ and $\theta^{(2)}<0$, define
\[
    a(\theta)=\Phi(\theta^{(1)}),
    \qquad
    b(\theta)=\Phi(\theta^{(2)}),
\]
where $\Phi$ is the standard normal distribution function.  We write $a_w=a(\theta_w)$ and $b_w=b(\theta_w)$ along a directed path $w\mapsto \theta_w$.

The least favorable pair for testing $\theta_0$ against $\theta$ has three regimes.  Define
\begin{align}
    \Theta_I
    &=
    \left\{
    \theta<0:
    a(\theta)\{1-b(\theta)\}\leq \frac14,
    \quad
    b(\theta)\{1-a(\theta)\}\leq \frac14
    \right\},
    \label{eq:ThetaI}\\
    \Theta_{II}
    &=
    \left\{
    \theta<0:
    a(\theta)\{1-b(\theta)\}>\frac14
    \right\},
    \label{eq:ThetaII}\\
    \Theta_{III}
    &=
    \left\{
    \theta<0:
    b(\theta)\{1-a(\theta)\}>\frac14
    \right\}.\label{eq:ThetaIII}
\end{align}
For later use, let
\[
    p_M(\theta)
    =\frac34-a(\theta)b(\theta)
\]
be the probability, under parameter $\theta$, that the market produces a monopoly outcome. Also define the probabilities of the regions in which the monopoly outcome is uniquely determined:
\begin{align*}
    \ell_{10}(\theta)
    &=\frac14-a(\theta)b(\theta)+\frac{a(\theta)}2,
    \\
    \ell_{01}(\theta)
    &=\frac14-a(\theta)b(\theta)+\frac{b(\theta)}2.
\end{align*}
Thus $\ell_{10}(\theta)$ is the probability of the latent region in which $(1,0)$ is the unique equilibrium outcome, and $\ell_{01}(\theta)$ is the analogous probability for $(0,1)$.

For a sample $s^n=(s_1,\ldots,s_n)$, define the outcome counts
\[
    N_{ab}(s^n)=\sum_{i=1}^n 1\{s_i=(a,b)\},
    \qquad (a,b)\in S,
\]
and define the monopoly count $N_M(s^n)=N_{10}(s^n)+N_{01}(s^n).$

Let $\{\theta_w:w\in\cW_1\}$ be any directed path with $\theta_w^{(1)}<0$ and $\theta_w^{(2)}<0$ for all $w\in\cW_1$.  The path may be linear, for example $\theta_w=-w(\zeta_1,\zeta_2)$, or nonlinear.  The important requirement is that it stays in one of the three parameter regions on the chosen $\cW_1$. Then, the following characterization holds.

\begin{proposition}[Exact conditional UMP tests along fixed-regime paths]
\label{thm:entry-fixed-regime}
Consider the binary entry game described above, with null
$\theta_0=(0,0)$ and $U_i\stackrel{i.i.d.}{\sim}N(0,I_2)$ across markets.
Let $\{\theta_w:w\in\cW_1\}$ be a directed family of alternatives with
$\theta_w^{(1)}<0$ and $\theta_w^{(2)}<0$ for all $w\in\cW_1$. Suppose that
the path stays entirely in one of the three LFP regimes
\eqref{eq:ThetaI}--\eqref{eq:ThetaIII}. Then the corresponding conditional
LR test, defined below with critical values and boundary randomization
chosen to have conditional level $\alpha$ under the least favorable null law,
has exact conditional robust level and is conditionally UMP in power
guarantee on the directed family.

\begin{enumerate}[]
\item \textbf{Case I:} Suppose $\theta_w\in\Theta_I$ for all $w\in\cW_1$.
Let
\[
    K_n^I=N_{11}+N_{10}+N_{01},
    \qquad
    T_n^I=-N_{11}.
\]
Conditional on $K_n^I=k$, the least favorable null law is
\[
    N_{11}\mid K_n^I=k
    \sim
    \Bin\left(k,\frac13\right),
\]
and the least favorable alternative law is
\[
    N_{11}\mid K_n^I=k
    \sim
    \Bin\left(k,p_I(w)\right),
    \qquad
    p_I(w)=\frac{4a_wb_w}{3}<\frac13.
\]
Thus the exact conditional UMP test rejects for small values of $N_{11}$
conditional on $K_n^I$, or equivalently for large values of $T_n^I=-N_{11}$.

\item \textbf{Case II:} Suppose $\theta_w\in\Theta_{II}$ for all $w\in\cW_1$.  Let
\[
    K_n^{II}=(N_{00},N_{11},N_M),
    \qquad
    T_n^{II}=N_{10}.
\]
For $K_n^{II}=(n_{00},n_{11},m)$, the least favorable null law is
\[
    N_{10}\mid K_n^{II}=(n_{00},n_{11},m)
    \sim
    \Bin\left(m,\frac12\right),
\]
and the least favorable alternative law is
\[
    N_{10}\mid K_n^{II}=(n_{00},n_{11},m)
    \sim
    \Bin\left(m,p_{II}(w)\right),
\]
where
\[
    p_{II}(w)
    =
    \frac{\ell_{10}(\theta_w)}{p_M(\theta_w)}
    =
    \frac{\frac14-a_wb_w+\frac{a_w}{2}}{\frac34-a_wb_w}
    >\frac12.
\]
Thus the exact conditional UMP test rejects for large values of $N_{10}$ conditional on $(N_{00},N_{11},N_M)$.

\item \textbf{Case III:} Suppose $\theta_w\in\Theta_{III}$ for all $w\in\cW_1$.  Let
\[
    K_n^{III}=(N_{00},N_{11},N_M),
    \qquad
    T_n^{III}=N_{01}.
\]
For $K_n^{III}=(n_{00},n_{11},m)$, the least favorable null law is
\[
    N_{01}\mid K_n^{III}=(n_{00},n_{11},m)
    \sim
    \Bin\left(m,\frac12\right),
\]
and the least favorable alternative law is
\[
    N_{01}\mid K_n^{III}=(n_{00},n_{11},m)
    \sim
    \Bin\left(m,p_{III}(w)\right),
\]
where
\[
    p_{III}(w)
    =
    \frac{\ell_{01}(\theta_w)}{p_M(\theta_w)}
    =
    \frac{\frac14-a_wb_w+\frac{b_w}{2}}{\frac34-a_wb_w}
    >\frac12.
\]
Thus the exact conditional UMP test rejects for large values of $N_{01}$ conditional on $(N_{00},N_{11},N_M)$.
\end{enumerate}
\end{proposition}

\begin{proof}
We verify the conditional least favorable laws and the monotone
likelihood-ratio ordering in each fixed regime, and then apply
Proposition~\ref{prop:cond-ump}.

Consider Case I.  The statistic
\[
    K_n^I=N_{11}+N_{10}+N_{01}
\]
is selection-invariant: it counts markets with at least one entrant.
The statistic $N_{11}$ is also selection-invariant, because the duopoly
outcome is uniquely determined by the latent variables. Under the null,
conditional on $K_n^I=k$,
\[
    N_{11}\mid K_n^I=k
    \sim
    \Bin\left(k,\frac13\right).
\]
Under $\theta_w$, the probability of duopoly is $a_wb_w$, while the
probability of at least one entrant is $3/4$. Hence, conditional on
$K_n^I=k$,
\[
    N_{11}\mid K_n^I=k
    \sim
    \Bin\left(k,\frac{a_wb_w}{3/4}\right)
    =
    \Bin\left(k,\frac{4a_wb_w}{3}\right).
\]
Since $\theta_w^{(1)}<0$ and $\theta_w^{(2)}<0$, we have
$a_w<1/2$ and $b_w<1/2$, so $4a_wb_w/3<1/3$.  The likelihood ratio of
$\Bin(k,p_I(w))$ relative to $\Bin(k,1/3)$ is therefore decreasing in
$N_{11}$, or equivalently increasing in $T_n^I=-N_{11}$.  This verifies the
conditional least favorable law and monotone likelihood-ratio conditions in
Case I.

Next, consider Case II.  Condition on
\[
    K_n^{II}=(N_{00},N_{11},N_M)=(n_{00},n_{11},m).
\]
This conditioning statistic is selection-invariant: $N_{00}$ and $N_{11}$
are uniquely determined by the latent variables, and $N_M=N_{10}+N_{01}$
counts observations in the monopoly region, regardless of how the two
monopoly outcomes are selected.

Under the null, the model is complete and, conditional on the $m$ monopoly
observations, the two monopoly outcomes are equally likely. Hence
\[
    N_{10}\mid K_n^{II}=(n_{00},n_{11},m)
    \sim
    \Bin\left(m,\frac12\right).
\]

Under $\theta_w$, decompose the monopoly region into the unique $(1,0)$
region, the unique $(0,1)$ region, and the ambiguous monopoly region.  The
unique $(1,0)$ region has probability $\ell_{10}(\theta_w)$, and the total
monopoly region has probability
\[
    p_M(\theta_w)=\frac34-a_wb_w.
\]
Let $A_{10}$ denote the number of observations in the unique $(1,0)$ region
among the $m$ monopoly observations, and let $R_{10}$ denote the number of
ambiguous monopoly observations selected into $(1,0)$.  Note that
\[
    N_{10}=A_{10}+R_{10},
    \qquad
    R_{10}\geq 0.
\]
Therefore, for any upper-tail rejection rule in $N_{10}$, the rejection
probability is minimized by selecting all ambiguous monopoly observations
into $(0,1)$, so that $R_{10}=0$. Under this least favorable selection,
\[
    N_{10}=A_{10},
    \qquad
    A_{10}\mid K_n^{II}=(n_{00},n_{11},m)
    \sim
    \Bin\left(m,\frac{\ell_{10}(\theta_w)}{p_M(\theta_w)}\right).
\]
Thus the least favorable alternative conditional law is
\[
    N_{10}\mid K_n^{II}=(n_{00},n_{11},m)
    \sim
    \Bin(m,p_{II}(w)),
\]
where
\[
    p_{II}(w)
    =
    \frac{\ell_{10}(\theta_w)}{p_M(\theta_w)}
    =
    \frac{\frac14-a_wb_w+\frac{a_w}{2}}
         {\frac34-a_wb_w}.
\]
Moreover,
\[
    \ell_{10}(\theta_w)-\frac12p_M(\theta_w)
    =
    \frac12\left[
    a_w(1-b_w)-\frac14
    \right]
    >0,
\]
where the inequality is due to the definition of $\Theta_{II}$. Hence
$p_{II}(w)>1/2$. The likelihood ratio of $\Bin(m,p_{II}(w))$ relative to
$\Bin(m,1/2)$ is increasing in $N_{10}$. This verifies the high-level
conditions of Proposition~\ref{prop:cond-ump} for Case II.

The proof for Case III is identical after interchanging the roles of the two
players, and is therefore omitted.

In all three cases, the least favorable null conditional law does not depend
on $w$, and the least favorable alternative conditional likelihood ratios are
monotone in the stated ordering statistic. Proposition~\ref{prop:cond-ump}
therefore yields exact conditional robust level and conditional UMP optimality
in power guarantee along the fixed-regime path.
\end{proof}

\subsection{Example 2: (Binary) Roy model}\label{sec_apdx:roy}
In this section, we provide details on Example \ref{ex:roy} illustrating the calculation of the least favorable pair and minimax tests. Recall that $S=\{(0,0),(0,1),(1,0),(1,1)\}$, and the sharp identifying restrictions are given by
	\begin{align}
		\theta^{(1,0)}&\le P(\{(1,0)\})\label{eq:roy_sir1}\\
		\theta^{(0,1)}&\le P(\{(1,1)\})\label{eq:roy_sir2}\\
		\theta^{(0,0)}&= P(\{(0,0)\})+P(\{(0,1)\}).\label{eq:roy_sir3}
	\end{align}
\subsubsection{LFP}
We start with the following characterization of the LFP for the hypotheses considered in the text. For this, let $\bar c>0$ be a known constant.
\begin{proposition}\label{prop:roy_LFP}
Let $(S,U,G,\Theta)$ be defined as in Example \ref{ex:roy}. Let $m_\theta$ be a discrete distribution on $U$ whose distribution is uniquely determined by $\theta=(\theta^{(0,0)},\theta^{(0,1)},\theta^{(1,0)})'.$ Suppose (i) $\theta^{(0,0)}_0=\theta^{(0,0)}_1=\bar c>0$ and $\theta^{(1,0)}_0<1-\bar c-\theta^{(0,1)}_0$; and (ii) $\theta^{(1,0)}_1>1-\bar c-\theta^{(0,1)}_0$.

\noindent
Case 1: If $\theta_1^{(1,0)}< 1-\bar c-\theta^{(0,1)}_1$, the following are the densities of an LFP $(Q_0,Q_1)\in\mathcal P_{\theta_0}\times\mathcal P_{\theta_1}$:
\begin{align}
	(q_0(0,0),q_0(0,1),q_0(1,0),q_0(1,1))'&=(\frac{\bar c}{2},\frac{\bar c}{2},1-\bar c-\theta^{(0,1)}_0,\theta^{(0,1)}_0)'\\
	(q_1(0,0),q_1(0,1),q_1(1,0),q_1(1,1))'&=(\frac{\bar c}{2},\frac{\bar c}{2},\theta^{(1,0)}_1,1-\bar c-\theta^{(1,0)}_1)'.
\end{align}

\noindent
Case 2: If $\theta_1^{(1,0)}= 1-\bar c-\theta^{(0,1)}_1$,  the following are the densities of an LFP $(Q_0,Q_1)\in\mathcal P_{\theta_0}\times\mathcal P_{\theta_1}$:
\begin{align}
	(q_0(0,0),q_0(0,1),q_0(1,0),q_0(1,1))'&=(\frac{\bar c}{2},\frac{\bar c}{2},1-\bar c-\theta^{(0,1)}_0,\theta^{(0,1)}_0)'\\
	(q_1(0,0),q_1(0,1),q_1(1,0),q_1(1,1))'&=(\frac{\bar c}{2},\frac{\bar c}{2},\theta^{(1,0)}_1,\theta^{(0,1)}_1)'.
\end{align}
\end{proposition}

Before giving a proof of the claim above, a few remarks are in order.
To simplify the analysis, we assume that $\theta^{(0,0)}_0=\theta^{(0,0)}_1$ are known. Recall that the sharp identifying restrictions in \eqref{eq:roy_sir1}-\eqref{eq:roy_sir3} imply
\begin{align}
	\theta^{(1,0)}_0\le P_0(\{(1,0)\})\le 1-\bar c-\theta^{(0,1)}_0.
\end{align}
The assumption $\theta^{(1,0)}_0<1-\bar c-\theta^{(0,1)}_0$ therefore makes the model incomplete at $\theta_0$. Finally, the assumption $\theta^{(1,0)}_1>1-\bar c-\theta^{(0,1)}_0$ is equivalent to $\nu_{\theta_1}(\{(1,0)\})>\nu_{\theta_0}(\{(1,0)\})$, which ensures that $\mathcal P_{\theta_0}\cap \mathcal P_{\theta_1}=\emptyset.$

\begin{proof}
The following convex program characterizes the LFP:
\begin{align}
	\min_{(p_0,p_1)\in \Delta^3\times\Delta^3}&\sum_{s\in\{(0,0),(0,1),(1,0),(1,1)\}}-\ln \big(\frac{p_0(s)}{p_0(s)+p_1(s)}\big)(p_0(s)+p_1(s))\\
	s.t.~& 	
	 \theta^{(1,0)}_0\le p_0(1,0)\notag\\
	&\theta^{(1,0)}_1\le p_1(1,0)\notag\\
	&\theta^{(0,1)}_0\le p_0(1,1)\notag\\
	&\theta^{(0,1)}_1\le p_1(1,1)\notag\\
	&\bar c= p_0(0,0)+p_0(0,1)\notag	\\
	&\bar c= p_1(0,0)+p_1(0,1)\notag	
\end{align}
First, we concentrate out $p_0(0,0),p_1(0,0),p_0(0,1),p_1(0,1)$ from the problem. The subset of the KKT conditions that involves these components are, for $s\in\{(0,0),(0,1)\}$
\begin{align}
	&-\frac{p_0(s)+p_1(s)}{p_0(s)}\frac{p_0(s)+p_1(s)-p_0(s)}{(p_0(s)+p_1(s))^2}(p_0(s)+p_1(s))-\ln \big(\frac{p_0(s)}{p_0(s)+p_1(s)}\big)-\lambda_1=0\\
	&-\frac{p_0(s)+p_1(s)}{p_0(s)}\frac{-p_0(s)}{(p_0(s)+p_1(s))^2}(p_0(s)+p_1(s))-\ln \big(\frac{p_0(s)}{p_0(s)+p_1(s)}\big)-\lambda_2=0\\
	&	p_0(0,0)+p_0(0,1)=\bar c	\label{eq:KKTeqc1}\\
	&	p_1(0,0)+p_1(0,1)=\bar c,	\label{eq:KKTeqc2}
\end{align}
for some $\lambda_1\ne 0$ and $\lambda_2\ne 0.$ The first two conditions can be simplified as
\begin{align}
	-\frac{p_1(s)}{p_0(s)}-\ln \big(\frac{p_0(s)}{p_0(s)+p_1(s)}\big)-\lambda_1&=0\\
	1-\ln \big(\frac{p_0(s)}{p_0(s)+p_1(s)}\big)-\lambda_2&=0.
\end{align}
This can be further simplified as
\begin{align}
\frac{p_1(s)}{p_0(s)}=\lambda_2-\lambda_1-1,~s\in\{(0,0),(0,1)\}
\end{align}
which implies $\frac{p_1(0,0)}{p_0(0,0)}=\frac{p_1(0,1)}{p_0(0,1)},$ and hence, if $p_0(0,0)\in (0,\bar c),$
\begin{align}
	\frac{p_1(0,0)}{p_1(0,1)}=\frac{p_0(0,0)}{p_0(0,1)}=\beta
\end{align}
for some $0< \beta<\infty$.
This, together with \eqref{eq:KKTeqc1} yields
\begin{align}
	(p_j(0,0),p_j(0,1))=(\frac{\beta \bar c}{1+\beta},\frac{\bar c}{1+\beta}), ~j=0,1.
\end{align}
For example, one can take the following as a solution:
\begin{align}
	(q_j(0,0),q_j(0,1))=(\frac{\bar c}{2},\frac{\bar c}{2}).
\end{align}
After concentrating out $p_0(0,0),p_1(0,0),p_0(0,1),p_1(0,1)$, the problem becomes
\begin{align}
	\min_{(p_0,p_1)\in \Delta^3\times\Delta^3}&\sum_{s\in\{(1,0),(1,1)\}}-\ln \big(\frac{p_0(s)}{p_0(s)+p_1(s)}\big)(p_0(s)+p_1(s))\\
	s.t.~& 	
	 \theta^{(1,0)}_0\le p_0(1,0) \notag\\
	&\theta^{(1,0)}_1\le p_1(1,0) \notag\\
	&\theta^{(0,1)}_0\le p_0(1,1)\notag\\
	&\theta^{(0,1)}_1\le p_1(1,1)\notag\\
	&p_0(1,0)+p_0(1,1)=1-\bar c\notag\\
	&p_1(1,0)+p_1(1,1)=1-\bar c.\notag
\end{align}
The KKT conditions are
\begin{align}
	&-\frac{p_1(1,0)}{p_0(1,0)}-\ln \big(\frac{p_0(1,0)}{p_0(1,0)+p_1(1,0)}\big)-\chi_1-\chi_5=0 \label{eq:KKT_Roy1}\\
	&-\frac{p_1(1,1)}{p_0(1,1)}-\ln \big(\frac{p_0(1,1)}{p_0(1,1)+p_1(1,1)}\big)-\chi_3-\chi_5=0 \label{eq:KKT_Roy2}\\
	&~~1-\ln \big(\frac{p_0(1,0)}{p_0(1,0)+p_1(1,0)}\big)-\chi_2-\chi_6=0						   \label{eq:KKT_Roy3}\\
	&~~1-\ln \big(\frac{p_0(1,1)}{p_0(1,1)+p_1(1,1)}\big)-\chi_4-\chi_6=0						   \label{eq:KKT_Roy4}\\
	&~~\chi_1(p_0(1,0)  - \theta^{(1,0)}_0)=0												   \label{eq:KKT_Roy5}\\
	&~~\chi_2(p_1(1,0)  - \theta^{(1,0)}_1)=0												   \label{eq:KKT_Roy6}\\
	&~~\chi_3(p_0(1,1)  - \theta^{(0,1)}_0)=0										   \label{eq:KKT_Roy7}\\
	&~~\chi_4(p_1(1,1)  - \theta^{(0,1)}_1)=0.									   \label{eq:KKT_Roy8}
\end{align}
where $\chi_j\ge 0$ for $j=1,\dots,4$, and the original inequality and equality constraints are also imposed.

\noindent
\textbf{Case 1: ($\chi_1=0,\chi_2>0,\chi_3>0,\chi_4=0$)}

Suppose $\chi_3>0$. Then, $p_0(1,1) =\theta^{(0,1)}_0$ by the complementary slackness condition \eqref{eq:KKT_Roy7}. It also implies $p_0(1,0)=1-\bar c-\theta^{(0,1)}_0$ by $ p_0(1,0)+ p_0(1,1)=1-\bar c$.
Further, $\chi_1=0$ because the constraints associated with $\chi_1$ and $\chi_3$ cannot bind simultaneously due to the assumption that $\theta_0^{(1,0)}<1-\bar c-\theta_0^{(0,1)}.$
Suppose further that $\chi_2>0$. Then, $p_1(1,0)=\theta^{(1,0)}_1$ and $p_1(1,1)=1-\bar c-\theta^{(1,0)}_1.$ We also assume that $\chi_4=0.$

Now note that \eqref{eq:KKT_Roy1}-\eqref{eq:KKT_Roy6} reduce to
\begin{align}
&-\frac{\theta^{(1,0)}_1}{1-\bar c-\theta^{(0,1)}_0}-\ln \big(\frac{1-\bar c-\theta^{(0,1)}_0}{1-\bar c-\theta^{(0,1)}_0+\theta^{(1,0)}_1}\big)-\chi_5=0	 \\
&-\frac{1-\bar c-\theta^{(1,0)}_1}{\theta^{(0,1)}_0}-\ln \big(\frac{\theta^{(0,1)}_0}{\theta^{(0,1)}_0+1-\bar c-\theta^{(1,0)}_1}\big)-\chi_3-\chi_5=0     \\
&~~1-\ln \big(\frac{1-\bar c-\theta^{(0,1)}_0}{1-\bar c-\theta^{(0,1)}_0+\theta^{(1,0)}_1}\big)-\chi_2-\chi_6=0						                     \\
&~~1-\ln \big(\frac{\theta^{(0,1)}_0}{\theta^{(0,1)}_0+1-\bar c-\theta^{(1,0)}_1}\big)-\chi_6=0				
\end{align}
It can be shown that, when $\theta_1^{(1,0)}>1-\bar c-\theta_0^{(0,1)}$, the system can be solved for $(\chi_2,\chi_3,\chi_5,\chi_6)$ that satisfies $\chi_j>0$ for $j=2,3,$ and hence the solution (the remaining components of LFP) is
\begin{align}
	(q_0(1,0),q_0(1,1))&=(1-\bar c-\theta^{(0,1)}_0,\theta^{(0,1)}_0)\\
	(q_1(1,0),q_1(1,1))&=(\theta^{(1,0)}_1,1-\bar c-\theta^{(1,0)}_1).
\end{align} 

\noindent
\textbf{Case 2: ($\chi_1=0,\chi_2>0,\chi_3>0,\chi_4>0$)}

The difference from Case 1 is that $\chi_4>0$ is assumed. By the complementary slackness condition, this implies $p_1(1,1)  = \theta^{(0,1)}_1$, which also equals $1-\bar c-\theta^{(1,0)}_1$ due to $\chi_2>0$ and $ p_1(1,0)+ p_1(1,1)=1-\bar c$. Therefore, the solutions are
\begin{align}
	(q_0(1,0),q_0(1,1))&=(1-\bar c-\theta^{(0,1)}_0,\theta^{(0,1)}_0)\\
	(q_1(1,0),q_1(1,1))&=(\theta^{(1,0)}_1,\theta^{(0,1)}_1).
\end{align}
The KKT conditions reduce to \begin{align}
&-\frac{\theta^{(1,0)}_1}{1-\bar c-\theta^{(0,1)}_0}-\ln \big(\frac{1-\bar c-\theta^{(0,1)}_0}{1-\bar c-\theta^{(0,1)}_0+\theta^{(1,0)}_1}\big)-\chi_5=0	 \\
&-\frac{\theta^{(0,1)}_1}{\theta^{(0,1)}_0}-\ln \big(\frac{\theta^{(0,1)}_0}{\theta^{(0,1)}_0+\theta^{(0,1)}_1}\big)-\chi_3-\chi_5=0     \\
&~~1-\ln \big(\frac{1-\bar c-\theta^{(0,1)}_0}{1-\bar c-\theta^{(0,1)}_0+\theta^{(1,0)}_1}\big)-\chi_2-\chi_6=0						                     \\
&~~1-\ln \big(\frac{\theta^{(0,1)}_0}{\theta^{(0,1)}_0+\theta^{(0,1)}_1}\big)-\chi_4-\chi_6=0.				
\end{align}
One may again show that there exists $(\chi_2,\dots,\chi_6)$ that solves the system with $\chi_j>0$ for $j=2,3,4,$ which is consistent with the assumption.
\end{proof}

\end{document}